\documentclass[reprint, floatfix, superscriptaddress, amsmath, amssymb, longbibliography]{revtex4-2}

\pdfoutput=1
\usepackage[utf8]{inputenc}
\usepackage[T1]{fontenc}

\usepackage{graphicx}
\usepackage{tikz}
\usetikzlibrary{patterns,calc,arrows.meta,positioning,decorations.pathmorphing,shapes.geometric,circuits.ee.IEC}

\usepackage{mathtools}
\usepackage{physics}
\usepackage{braket}
\usepackage{amsthm}
\usepackage{array}
\usepackage{bm}
\usepackage{multirow}

\usepackage{booktabs}

\usepackage{url}
\usepackage{hyperref}
\usepackage{cleveref}
\usepackage{tikz}
\usepackage[margin=2cm]{geometry}
\usepackage{pgfplots}
\pgfplotsset{compat=1.17}
\usetikzlibrary{circuits.ee.IEC,positioning,arrows.meta,shapes.geometric,calc,decorations.pathmorphing,patterns,fit,backgrounds,decorations.markings}

\usepackage{xcolor}

\newtheorem{theorem}{Theorem}[section]
\newtheorem{lemma}[theorem]{Lemma}
\newtheorem{proposition}[theorem]{Proposition}
\newtheorem{corollary}[theorem]{Corollary}
\newtheorem{definition}[theorem]{Definition}
\newtheorem{assumption}[theorem]{Assumption}
\newtheorem{remark}[theorem]{Remark}

\newcommand{\Yin}{Y_{\mathrm{in}}}

\tikzset{
    capacitor/.style={thick, decoration={markings, mark=at position 0.5 with {\node[transform shape] {$\parallel$};}}, postaction={decorate}},
    inductor/.style={thick, decoration={coil, aspect=0.4, segment length=3mm, amplitude=2mm}, decorate},
    resistor/.style={thick, decoration={zigzag, segment length=4mm, amplitude=1mm}, decorate},
    junction/.style={circle, fill=black, inner sep=0pt, minimum size=4pt},
    block/.style={draw, thick, minimum width=2cm, minimum height=1cm, align=center},
    arrow/.style={-{Stealth[length=2mm]}, thick},
    dashedarrow/.style={-{Stealth[length=2mm]}, thick, dashed},
}

\begin{document}
\raggedbottom

\title{Exact Multimode Quantization of Superconducting Circuits via Boundary Admittance and Continued Fractions}

\author{Mustafa Bakr}
\email{mustafa.bakr@physics.ox.ac.uk}
\affiliation{Clarendon Laboratory, Department of Physics, University of Oxford}
\author{Robin Wopalenski}
\affiliation{Clarendon Laboratory, Department of Physics, University of Oxford}

\begin{abstract}
Accurate extraction of linearized quantum circuit models from electromagnetic simulations is essential for the design of superconducting circuits. We present a quantization framework based on the driving-point admittance $Y_{\mathrm{in}}(s)$ seen by a Josephson junction embedded in an arbitrary passive linear environment. By taking the Schur complement of the nodal admittance matrix, we show that the linearized coupled system obeys an eigenvalue-dependent boundary condition, $s Y_{\mathrm{in}}(s) + 1/L_J = 0$, whose roots determine the dressed linear mode frequencies. This boundary condition admits an exact continued fraction representation: any positive-real admittance can be realized as a canonical Cauer ladder, yielding a tridiagonal (Jacobi) structure that enables certified convergence bounds via interlacing theorems.For the full nonlinear Hamiltonian, we treat Josephson junctions in the charge basis, where each cosine potential is exactly tridiagonal, and couple them to cavity modes in the Fock basis; in the general multi-junction case this yields a block-tridiagonal structure solvable by matrix continued fractions, enabling systematic diagonalization across all coupling regimes from dispersive through ultrastrong and deep strong coupling. The resulting quantization procedure is: (i)~compute or measure $Y_{\mathrm{in}}(s)$, (ii)~solve the boundary condition to obtain dressed eigenfrequencies, (iii)~synthesize an equivalent passive network, and (iv)~quantize while retaining the full cosine nonlinearity of the Josephson junction. We prove that junction participation decays as $\mathcal{O}(\omega_n^{-1})$ at high frequencies for any circuit with finite shunt capacitance, ensuring ultraviolet convergence of perturbative corrections without imposed cutoffs. This framework bridges classical microwave network theory and circuit quantum electrodynamics, enabling Hamiltonian parameters, including corrections beyond the rotating-wave and dispersive approximations, to be extracted directly from simulated or measured admittance data.

\end{abstract}

\maketitle


\section{Introduction}
\label{sec:introduction}

A Josephson junction~\cite{Josephson1962} embedded in an electromagnetic environment experiences that environment as a frequency-dependent load. At low frequencies, the surrounding circuit appears predominantly inductive; at high frequencies, predominantly capacitive; and near resonances, strongly reactive. The starting point of this paper is that this entire influence can be compressed into a single scalar object that fully determines the junction's linear response: the driving-point admittance seen at the junction port $\Yin(s)$. Once $\Yin(s)$ is known, either from measurement or from electromagnetic simulation followed by network reduction, it can be used to construct a quantitative quantum description of the coupled system.

The picture is simple. Consider a junction shunted by its capacitance $C_J$ and connected to an otherwise arbitrary passive linear network built from resonators, transmission lines, and coupling elements. When the junction oscillates at angular frequency $\omega$, it injects current into the network, and the ratio of current to voltage at the port is $\Yin(s)$. This back-action reshapes the junction dynamics. Near an environmental resonance, the admittance is large and the junction hybridizes strongly with the resonant structure into dressed normal modes. Far from resonances, the junction behaves as a weakly perturbed anharmonic oscillator with small admixtures of environmental degrees of freedom. A crucial consequence is that, under a finite effective shunt capacitance at the junction port, high-frequency environmental modes are naturally suppressed at the junction. For $\omega \gg \omega_p$, where $\omega_p$ is the junction plasma frequency~\cite{Koch2007}, the shunt capacitance presents a low impedance that effectively short-circuits the junction port. In this regime, environmental eigenfunctions have vanishing junction participation: the junction node becomes energetically expensive to excite, suppressing high-frequency mode amplitudes at the junction. This physical suppression provides an intrinsic ultraviolet regularization within passive circuit models with finite junction-port capacitance, a condition satisfied by standard superconducting devices. This ultraviolet decoupling is consistent with exact multimode quantization analyses of superconducting circuits, where treating the electromagnetic environment as a passive infinite-dimensional system yields finite Hamiltonians without ad-hoc cutoffs~\cite{ParraRodriguez2018, Malekakhlagh2016, Malekakhlagh2017}. Closely related ultraviolet-convergent treatments of capacitive coupling to infinite-dimensional circuits were developed earlier in~\cite{Paladino2003, Gely2017}.

A central thesis of this paper is that the continued fraction is not merely a computational tool but a natural mathematical structure for quantum circuits with localized nonlinearity coupled to linear environments. From network theory, Cauer showed that any positive-real driving-point admittance $\Yin(s)$ admits a canonical ladder realization, corresponding to an exact (finite or infinite) continued-fraction expansion~\cite{Brune1931, Cauer1958}. The physical circuit topology is a chain of LC sections with nearest-neighbor coupling, which upon quantization yields a nearest-neighbor (tridiagonal, or block-tridiagonal) chain representation for the linear environment. From spectral theory, the resolvent of a tridiagonal (Jacobi) matrix is exactly a continued fraction whose poles are eigenvalues and whose residues encode eigenvector components~\cite{Stieltjes1894, Chihara1978, Simon2005OPUC}. This correspondence between network realizability and spectral theory is the bridge connecting classical circuit analysis to circuit QED. The same chain structure appears in quantum dissipation theory, where a system coupled to a harmonic bath maps to a semi-infinite chain via the Lanczos algorithm~\cite{Caldeira1983, Chin2010, Prior2010}. For circuit QED, the admittance $\Yin(s)$ encodes the complete information about the linear environment seen at the junction port; combined with the junction parameters, it defines a Hamiltonian whose spectrum can be obtained nonperturbatively, with certified convergence, across coupling regimes. 

The framework developed here begins from $\Yin(s)$ and proceeds to an exact circuit Hamiltonian that retains the full cosine nonlinearity of the Josephson potential. From this Hamiltonian one obtains the dressed spectrum, including avoided crossings and level repulsion~\cite{JaynesCummings1963, CohenTannoudji1992}. The familiar circuit-QED parameters, qubit-mode coupling strengths~\cite{Blais2004, Wallraff2004}, transmon anharmonicity~\cite{Koch2007}, and dispersive shifts~\cite{Schuster2005, Boissonneault2009}, emerge as controlled limits in appropriate parameter regimes. Because the construction retains the full cosine potential without invoking the rotating-wave approximation~\cite{Kockum2019, Blais2021}, corrections beyond the RWA and dispersive limits are obtained by systematic expansion of the exact Hamiltonian. Dissipative rates follow from combining the admittance description with Fermi's golden rule~\cite{Dirac1927, Fermi1950}, recovering the Purcell effect~\cite{Purcell1946} and its multimode generalizations~\cite{Houck2008, Clerk2010, BakrPurcell2025}. The central mathematical result is that the Schur complement reduction, which eliminates internal degrees of freedom to yield the driving-point admittance, constructs an eigenvalue-dependent boundary condition in the sense of Sturm-Liouville spectral theory. This equivalence, made explicit in Theorem~\ref{thm:bc_equivalence}, provides both conceptual clarity and practical computational tools: the dressed mode frequencies are roots of the boundary condition equation $s\Yin(s) + 1/L_J = 0$, and the mode shapes follow from network synthesis of the positive-real function $\Yin(s)$. The continued fraction framework provides exact implicit expressions whose numerical evaluation admits certified convergence across coupling regimes. Following the coupling classification reported in~\cite{Kockum2019}, we work uniformly from the dispersive regime ($g/|\Delta| \lesssim 0.1$) where the Jaynes-Cummings model applies, through strong coupling where vacuum Rabi oscillations are resolved, to ultrastrong coupling ($g/\omega_r \gtrsim 0.1$) where counter-rotating terms become significant and the ground state contains virtual photons, and into deep strong coupling ($g/\omega_r \gtrsim 1$) where light and matter become strongly entangled in the ground state. The quantum Rabi model conserves parity $\hat{P} = e^{i\pi \hat{a}^\dagger \hat{a}} \otimes \hat{\sigma}_z$, and in the parity eigenbasis its Hamiltonian becomes tridiagonal, precisely the continued fraction structure (see Section~\ref{sec:cf_framework}). This observation, combined with Braak's proof of integrability~\cite{Braak2011}, shows that continued-fraction representations of the resolvent/spectral function yield the exact spectrum of the quantum Rabi model at arbitrary coupling.

The practical advantage of this approach is that circuit parameters are not introduced as phenomenological fit parameters. They are inferred from the environment as encoded in $\Yin(s)$, so that the quantum description remains anchored to the measured or simulated electromagnetic structure while avoiding the computational burden of full field quantization. The Schur complement, a classical tool in linear algebra for block Gaussian elimination~\cite{Zhang2005, Haynsworth1968}, has found widespread application in computational electromagnetics. In microwave engineering, the coupled-integral-equation technique uses Schur complement reductions to analyze waveguide discontinuities and multi-port junctions~\cite{Amari1996, Amari1999}. Similar reduction procedures appear in domain decomposition methods for scattering problems~\cite{Turc2017} and in iterative solvers for coupled electromagnetic-circuit problems~\cite{Chou2016}. The same mathematical structure underlies circuit quantization: the driving-point admittance obtained by eliminating internal nodes simultaneously encodes the boundary condition for dressed modes. 

\begin{figure*}[t]
\centering
\begin{tikzpicture}[node distance=2cm, scale=0.95, transform shape]
    \tikzstyle{block} = [draw, thick, rectangle, minimum width=2.2cm, minimum height=1.2cm, align=center, rounded corners=3pt]
    \tikzstyle{arrow} = [-{Stealth[length=2.5mm]}, thick]
    
    \node[block, fill=blue!12] (Y) {Measured/Simulated\\Admittance\\$\Yin(\omega)$};
    \node[block, fill=blue!12, right=1.5cm of Y] (schur) {Schur\\Complement\\Reduction};
    \node[block, fill=green!12, right=1.5cm of schur] (synth) {Cauer/Foster\\Network\\Synthesis};
    \node[block, fill=orange!12, right=1.5cm of synth] (quant) {Canonical\\Quantization};
    \node[block, fill=red!12, right=1.5cm of quant] (obs) {Physical\\Observables};
    
    \draw[arrow] (Y) -- (schur);
    \draw[arrow] (schur) -- (synth);
    \draw[arrow] (synth) -- (quant);
    \draw[arrow] (quant) -- (obs);
    
    \node[below=0.4cm of Y, text width=2.5cm, align=center, font=\footnotesize] {From EM simulation\\or S-parameters};
    \node[below=0.4cm of schur, text width=2.5cm, align=center, font=\footnotesize] {Eliminates internal\\nodes (Sec.~\ref{sec:schur})};
    \node[below=0.4cm of synth, text width=2.5cm, align=center, font=\footnotesize] {Continued fraction\\(Sec.~\ref{sec:synthesis})};
    \node[below=0.4cm of quant, text width=2.5cm, align=center, font=\footnotesize] {Full cosine\\(Sec.~\ref{sec:quantization})};
    \node[below=0.4cm of obs, text width=2.5cm, align=center, font=\footnotesize] {$\omega_q$, $g$, $\chi$,\\$\alpha$, $\Gamma$};
    
    \node[above=0.3cm of schur, font=\footnotesize\itshape, text=blue!70!black] {PR preserved};
    \node[above=0.3cm of synth, font=\footnotesize\itshape, text=green!60!black] {Tridiagonal $H$};
    \node[above=0.3cm of quant, font=\footnotesize\itshape, text=orange!70!black] {$\hat{H} = \sum_n \omega_n \hat{a}_n^\dagger\hat{a}_n - E_J\cos\hat{\phi}$};
    
    \node[draw, thick, dashed, rounded corners, minimum width=12cm, minimum height=1.2cm, below=2.2cm of synth] (regimes) {};
    \node[above=0.1cm, font=\bfseries\small] at (regimes.north) {Full Cosine Retained; No RWA at Hamiltonian Level};
    \node[font=\footnotesize] at (regimes.center) {
        Dispersive: $g/|\Delta| < 0.1$ \quad\quad
        Strong: $g/|\Delta| \gtrsim 0.1$ \quad\quad
        USC: $g/\omega_r \gtrsim 0.1$ \quad\quad
        DSC: $g/\omega_r \gtrsim 1$
    };
    
\end{tikzpicture}
\caption{Overview of the boundary condition framework. The measured or simulated admittance $\Yin(\omega)$ is reduced to the junction port via Schur complement, synthesized into an equivalent circuit via Cauer continued fraction expansion, and quantized to yield the exact Hamiltonian. Physical observables are extracted by numerical diagonalization or perturbative expansion. The Hamiltonian retains the full cosine nonlinearity; standard circuit QED emerges as the dispersive limit. The dashed box indicates that the framework applies uniformly from dispersive through deep strong coupling: different regimes are distinguished by which terms in the cosine expansion are numerically significant, not by modifications to the quantization procedure.}
\label{fig:framework_overview}
\end{figure*}
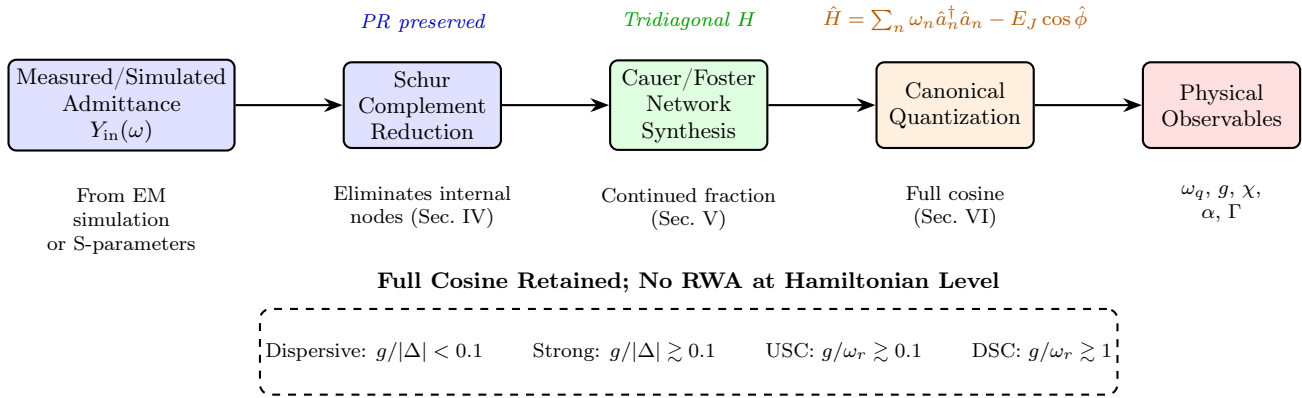

The present work is motivated by three points where the standard treatment of superconducting quantum circuits remains either implicit or incomplete, and aims to provide a boundary-condition formulation that makes the dressed-mode structure, ultraviolet behavior, and multi-mode interference directly computable from $\Yin(s)$ with explicit convergence guarantees. The theoretical description of superconducting quantum circuits is often introduced through the Jaynes--Cummings Hamiltonian~\cite{JaynesCummings1963} and its dispersive limit~\cite{Koch2007, Blais2004, Blais2021}. The transmon qubit design~\cite{Koch2007} and the circuit-QED architecture~\cite{Wallraff2004, Schuster2007} have become the dominant paradigm for superconducting quantum computing, building on the first demonstration of coherent control in a single-Cooper-pair box~\cite{Nakamura1999}. Impedance-based quantization methods~\cite{Nigg2012, Solgun2014, Solgun2015, ParraRodriguez2018, ParraRodriguez2019, Vool2017, Minev2021, Lu2025} provide a more systematic foundation by tying the quantum model directly to a measured or simulated linear response function. In particular, multiport impedance quantization~\cite{Solgun2014, Solgun2015, Egusquiza2022, Labarca2024, Alghadeer2025} establishes that broad classes of positive-real impedance matrices can be synthesized and quantized to yield an exact Hamiltonian.

First, the reduction of a multiport electromagnetic environment to a single-port admittance at the junction is typically presented as an engineering convenience. Here we show that the Schur-complement reduction constructs an eigenvalue-dependent boundary condition at the junction node. In that language, the dressed linear mode condition,
\begin{equation}
s\,\Yin(s)+\frac{1}{L_J}=0,
\end{equation}
is not merely a root-finding recipe but the spectral condition imposed by the junction termination. This connection to Sturm--Liouville theory~\cite{Fulton1977, Walter1973, Zettl2005}, developed further in~\cite{BakrQM2025}, provides a clean conceptual bridge between circuit reduction and operator-theoretic mode structure. Early studies focused on affine and single-pole eigenparameter dependence at the boundary~\cite{Fulton1977, Walter1973}. The general rational case, under Nevanlinna–Herglotz restrictions, was developed in detail by Binding, Browne, and Watson~\cite{Binding2002I, Binding2002II} and by Dijksma, Langer, and de Snoo~\cite{Dijksma1987}. For distributed environments such as transmission line resonators, we extend this formulation to include the spatial dependence of the admittance $\Yin(x,s)$, where the junction position $x_J$ selects the relevant boundary data and determines mode-dependent coupling strengths (see Section~\ref{sec:schur}).

Second, multimode treatments are frequently accompanied by ultraviolet worries: naive perturbative sums over modes appear to diverge unless an external cutoff is imposed. The resolution is that junction participation decays at high frequency in any circuit with finite shunt capacitance. We show that when $\Yin$ includes a finite shunt capacitance at the junction port, a condition satisfied by physical realization within passive lumped-element circuit theory, the junction participation scales as $\phi_n^J=\mathcal{O}(\omega_n^{-1})$ for large mode frequency $\omega_n$. This yields $\lambda_n\propto \omega_n^{-3/2}$ for the zero-point flux participation and implies absolute convergence of the perturbative sums that generate Kerr terms~\cite{Nigg2012}, Lamb shifts~\cite{Lamb1947, Bethe1947}, dispersive shifts~\cite{Schuster2005}, and related quantities. The result is stated and proved at the theorem level within the assumptions of lumped passive circuit theory~\cite{Newcomb1966, Anderson1973} (Theorem~\ref{thm:uv_main} in Section~\ref{sec:uv}). We extend the proof to distributed multimode environments through spectral measure arguments, establishing ultraviolet convergence for any passive linear environment satisfying the finite-capacitance assumption.

Third, the usual derivations of the coupling $g$~\cite{Blais2004}, anharmonicity $\alpha=-E_C/\hbar$~\cite{Koch2007}, and dispersive shift $\chi=g^2\alpha/[\Delta(\Delta+\alpha)]$~\cite{Koch2007, Boissonneault2009} are often developed inside the Jaynes--Cummings framework~\cite{JaynesCummings1963} and only then interpreted in circuit terms. We reverse this logic. Starting from the exact Hamiltonian
\begin{equation}
\hat{H}=\sum_n \hbar\omega_n \hat{a}_n^\dagger \hat{a}_n - E_J\cos\left(\frac{\hat{\Phi}_J}{\varphi_0}\right),
\end{equation}
we expand the cosine in a controlled way and recover the standard formulas as limits of the exact circuit model. This non-circular route makes clear which approximations are invoked, where they enter, and how they can be relaxed systematically. The operator-level content of parameters such as $g$ is thereby made explicit: $g$ is the coefficient of the bilinear operator $(\hat{b}+\hat{b}^\dagger)(\hat{a}+\hat{a}^\dagger)$ in the cosine expansion, and the rotating-wave approximation projects this onto the excitation-exchange sector $\hat{a}^\dagger\hat{b} + \hat{a}\hat{b}^\dagger$. The validity of this projection requires conditions we state precisely in Section~\ref{sec:coupling_regimes}. Furthermore, we derive explicit transcendental equations for dressed frequencies and provide closed-form solutions that enable circuit QED engineers to compute all relevant parameters directly from design specifications.

A subtlety that pervades circuit QED but is often left implicit concerns the choice of basis for quantizing nonlinear elements. The transmon qubit and the cavity resonator, though both described by oscillator-like Hamiltonians, require different quantum mechanical treatments due to the topology of their configuration spaces. The superconducting phase $\varphi$ across a Josephson junction is a compact variable living on a circle, with $\varphi$ and $\varphi + 2\pi$ representing the same physical state; this compactness reflects the discreteness of Cooper pair number. The natural basis is the charge basis $\{|n\rangle\}$ where $n \in \mathbb{Z}$ counts Cooper pairs, and the cosine potential couples only nearest-neighbor charge states, yielding the Mathieu equation~\cite{Koch2007}. In contrast, a linear resonator mode has a non-compact configuration variable with flux ranging over the entire real line, and the natural basis is the Fock basis of photon number eigenstates. For coupled systems, the correct procedure treats the transmon in the charge basis to obtain eigenstates and matrix elements, cavity modes in the Fock basis, and builds the full Hamiltonian in the product space, the approach employed by numerical packages such as scqubits~\cite{Groszkowski2021}. The continued fraction framework unifies these treatments: both charge-basis transmon and Fock-basis cavity yield tridiagonal Hamiltonians, and the boundary admittance encodes the linear sector while the Josephson nonlinearity is treated where the cosine is naturally tridiagonal.

The paper is organized as follows. Section~\ref{sec:circuit_variables} establishes circuit variables and Josephson constitutive relations. Section~\ref{sec:cf_framework} develops the continued fraction framework, establishing the equivalence between Cauer network synthesis, Jacobi matrix spectral theory, and the boundary condition formulation. Section~\ref{sec:schur} derives the Schur complement reduction, proves its equivalence to an eigenvalue-dependent boundary condition (Theorem~\ref{thm:bc_equivalence}), and develops the graphical solution method along with explicit transcendental equations for dressed frequencies. Section~\ref{sec:synthesis} treats network synthesis in Foster and Brune canonical forms. Section~\ref{sec:quantization} carries out canonical quantization to obtain the exact multimode Hamiltonian with the full cosine nonlinearity retained. Section~\ref{sec:uv} proves the ultraviolet convergence theorem (Theorem~\ref{thm:uv_main}), and extending to distributed multimode environments. Section~\ref{sec:multimode_purcell} analyzes multi-mode interference and Purcell suppression directly from the boundary-condition admittance. Section~\ref{sec:coupling_regimes} delineates coupling regimes with explicit validity conditions. Section~\ref{sec:observables} derives physical observables from the continued fraction structure, and Section~\ref{sec:usc_spinboson} treats ultrastrong coupling and spin-boson physics with experimental validation. Section~\ref{sec:parity_validation} establishes parity protection mechanisms and validates the matrix continued fraction against experimental data for a two-mode transmon in the strong nonlinear coupling regime.


\section{Circuit Variables and Constitutive Relations}
\label{sec:circuit_variables}

The quantum description of electrical circuits requires dynamical variables that admit a Lagrangian formulation and subsequent canonical quantization. The natural choice is the node flux, a time-integrated voltage that serves as the generalized coordinate for circuit dynamics. This section establishes the node flux formalism, derives the constitutive relations for linear and nonlinear elements, and states the convention for treating the junction capacitance that is maintained throughout. The central outcome is a consistent convention for treating the junction capacitance (Assumption~\ref{assumption:cj_convention}) that avoids double-counting and ensures well-defined canonical quantization.

\subsection{The Node Flux as Generalized Coordinate}
\label{sec:node_flux}

Consider a circuit network with $N$ nodes, where one node is designated as the ground reference with zero potential. The node flux at node $i$ is defined as
\begin{equation}
\Phi_i(t) \equiv \int_{-\infty}^{t} V_i(t') \, dt',
\label{eq:node_flux_def}
\end{equation}
where $V_i(t)$ is the voltage at node $i$ relative to ground. This definition implies the voltage-flux relation $V_i(t) = \dot{\Phi}_i(t)$. In the Laplace domain, the full transform relation is $\tilde{V}_i(s) = s \tilde{\Phi}_i(s) - \Phi_i(0^-)$; assuming zero initial conditions $\Phi_i(0^-) = 0$, this simplifies to $\tilde{V}_i(s) = s \tilde{\Phi}_i(s)$.

The node flux has dimensions of magnetic flux (Weber = Volt$\cdot$second). This is not coincidental: by Faraday's law, the time-integrated voltage around a closed loop equals the magnetic flux threading the loop. For superconducting circuits, the node flux is directly related to the gauge-invariant phase difference across circuit elements, with the conversion factor being the reduced flux quantum
\begin{equation}
\varphi_0 \equiv \frac{\Phi_0}{2\pi} = \frac{\hbar}{2e} \approx 0.329\,\mathrm{mWb}.
\label{eq:reduced_flux_quantum}
\end{equation}
We use $\varphi_0$ (with subscript zero) throughout to distinguish from the phase variable $\varphi$.

The physical interpretation becomes concrete when we consider how the flux appears in the energy expressions. For a capacitor, the stored energy depends on voltage, hence on $\dot{\Phi}$, making capacitance play the role of mass in the mechanical analogy. For an inductor, the stored energy depends on current, which is proportional to $\Phi$ itself, making inductance play the role of a spring constant. This identification of capacitance with kinetic energy and inductance with potential energy is the foundation of the Lagrangian formulation.

\subsection{Constitutive Relations for Linear Elements}
\label{sec:linear_elements}

The current through a capacitor with capacitance $C$ is proportional to the rate of change of voltage:
\begin{equation}
I_C = C\frac{dV}{dt} = C\ddot{\Phi}.
\label{eq:capacitor_current}
\end{equation}
The energy stored in the capacitor is
\begin{equation}
U_C = \frac{1}{2}CV^2 = \frac{1}{2}C\dot{\Phi}^2,
\label{eq:capacitor_energy}
\end{equation}
which depends on $\dot{\Phi}$ and therefore contributes to the kinetic energy in the Lagrangian.

The current through an inductor with inductance $L$ is given by integrating the voltage:
\begin{equation}
I_L = \frac{1}{L}\int V \, dt = \frac{\Phi_L}{L},
\label{eq:inductor_current}
\end{equation}
where $\Phi_L$ denotes the branch flux across the inductor (the difference of node fluxes at its terminals). Note the absence of time derivatives: the inductor current equals the branch flux divided by inductance, a relation that holds instantaneously in the time domain. In the Laplace domain with zero initial conditions, $\tilde{I}_L(s) = \tilde{\Phi}_L(s)/L$.
The energy stored in the inductor is
\begin{equation}
U_L = \frac{1}{2}LI_L^2 = \frac{\Phi_L^2}{2L},
\label{eq:inductor_energy}
\end{equation}
which depends on $\Phi_L$ alone and contributes to the potential energy.

A resistor with resistance $R$ dissipates energy and cannot be directly incorporated into a conservative Lagrangian. Its current is
\begin{equation}
I_R = \frac{V}{R} = \frac{\dot{\Phi}}{R}.
\label{eq:resistor_current}
\end{equation}
The treatment of dissipation requires coupling to a bath of oscillators, which we address in Section~\ref{sec:synthesis}.

The constitutive relations for inductors and capacitors imply that any LC ladder network has a natural chain structure. Consider a cascade of $N$ LC sections where node $n$ has flux $\Phi_n$, with $\Phi_0 = 0$ representing the grounded input. The Lagrangian is
\begin{equation}
\mathcal{L} = \sum_{n=1}^{N} \frac{1}{2}C_n \dot{\Phi}_n^2 - \sum_{n=1}^{N} \frac{(\Phi_n - \Phi_{n-1})^2}{2L_n},
\label{eq:lc_chain_lagrangian}
\end{equation}
where the first sum represents kinetic energy stored in shunt capacitors and the second represents potential energy stored in series inductors connecting adjacent nodes. The equations of motion couple only nearest-neighbor fluxes, yielding a tridiagonal dynamical matrix. Upon quantization, this becomes a tridiagonal Hamiltonian---the continued fraction structure developed in Section~\ref{sec:cf_framework}.

\subsection{The Josephson Junction: Constitutive Relations}
\label{sec:josephson}

The Josephson junction is characterized by two fundamental relations predicated by Josephson in 1962~\cite{Josephson1962}. The first Josephson relation gives the supercurrent through the junction as a function of the superconducting phase difference $\varphi$ across the barrier:
\begin{equation}
I_J = I_c \sin\varphi,
\label{eq:josephson_1}
\end{equation}
where $I_c$ is the critical current. The second Josephson relation connects the voltage across the junction to the time derivative of the phase:
\begin{equation}
V_J = \varphi_0 \frac{d\varphi}{dt}, \quad \text{where} \quad \varphi_0 \equiv \frac{\hbar}{2e} = \frac{\Phi_0}{2\pi}.
\label{eq:josephson_2}
\end{equation}
Integrating the second relation gives $\varphi = \Phi_J/\varphi_0$, where $\Phi_J$ is the branch flux across the junction. The current-flux relation becomes
\begin{equation}
I_J = I_c \sin\left(\frac{\Phi_J}{\varphi_0}\right).
\label{eq:josephson_current_flux}
\end{equation}

The potential energy is obtained by computing the work done against the junction current:
\begin{equation}
U_J = \int I_J \, d\Phi_J = -I_c\varphi_0\cos\left(\frac{\Phi_J}{\varphi_0}\right) = -E_J\cos\left(\frac{\Phi_J}{\varphi_0}\right),
\label{eq:josephson_energy}
\end{equation}
where $E_J \equiv I_c\varphi_0 = \hbar I_c/(2e)$ is the Josephson energy. We have chosen the constant of integration so that $U_J(\Phi_J = 0) = -E_J$, the minimum of the potential.

The cosine potential has minima at $\Phi_J = 2\pi n \varphi_0$ for integer $n$. Near any minimum, the potential is approximately harmonic:
\begin{equation}
U_J \approx -E_J + \frac{E_J}{2\varphi_0^2}\Phi_J^2 = -E_J + \frac{\Phi_J^2}{2L_J},
\label{eq:josephson_harmonic}
\end{equation}
where the Josephson inductance is
\begin{equation}
L_J \equiv \frac{\varphi_0^2}{E_J} = \frac{\varphi_0}{I_c}.
\label{eq:josephson_inductance}
\end{equation}
This linearization is valid for $|\Phi_J| \ll \varphi_0$, or equivalently $|\varphi| \ll 1$. In the transmon regime where $E_J/E_C \gg 1$, the phase fluctuations are small and the linearization provides a good starting point, though the exact treatment developed in this paper retains the full cosine.

The phase variable $\varphi = \Phi_J/\varphi_0$ is fundamentally different from the flux in a linear inductor: it is a compact variable with $\varphi \sim \varphi + 2\pi$, reflecting the quantization of Cooper pair number~\cite{Likharev1985, Averin1986}. The conjugate variable is the Cooper pair number operator $\hat{n} = -i\partial/\partial\varphi$ with eigenvalues $n \in \mathbb{Z}$. The cosine potential acts as a nearest-neighbor hopping term in the charge basis:
\begin{equation}
\cos\hat{\varphi} = \frac{1}{2}\sum_{n=-\infty}^{\infty} \big(|n+1\rangle\langle n| + |n\rangle\langle n+1|\big).
\label{eq:cosine_charge_basis}
\end{equation}
This tridiagonal structure in the charge basis is the continued fraction representation of the transmon Hamiltonian. The resulting eigenvalue problem is equivalent to the Mathieu equation~\cite{Koch2007, Cottet2002}, whose solutions give the transmon spectrum $\omega_{01} \approx \sqrt{8E_JE_C} - E_C$ and anharmonicity $\alpha \approx -E_C$ in the transmon limit $E_J/E_C \gg 1$.

\subsection{The Junction Capacitance Convention}
\label{sec:cj_convention_subsec}

A physical Josephson junction has a geometric capacitance $C_J$ arising from the parallel-plate structure of the electrodes separated by the insulating barrier. This capacitance participates in the circuit dynamics and must be accounted for in the quantization. However, there is freedom in how to incorporate $C_J$: it can be included in the linear electromagnetic environment described by $Y_{\mathrm{in}}(s)$, or it can be treated as a separate term in the junction Hamiltonian with a charging energy $Q_J^2/(2C_J)$.

\begin{assumption}[Junction Capacitance Inclusion]
\label{assumption:cj_convention}
The driving-point admittance $Y_{\mathrm{in}}(s)$ seen by the junction includes the junction capacitance $C_J$. Under this convention, the linear electromagnetic environment encompasses all linear elements connected to the junction node, including $C_J$. The Josephson junction then contributes only the nonlinear cosine potential to the Hamiltonian:
\begin{equation}
H_J = -E_J\cos\left(\frac{\hat{\Phi}_J}{\varphi_0}\right).
\label{eq:junction_hamiltonian}
\end{equation}
No separate charging energy term $Q_J^2/(2C_J)$ appears.
\end{assumption}

This convention, which has become standard in circuit QED following Ref.~\cite{Nigg2012}, has three important consequences that justify its adoption. First, it avoids double-counting: if $C_J$ were included both in $Y_{\mathrm{in}}(s)$ and as a separate charging term, the junction capacitance would be counted twice, leading to incorrect eigenfrequencies. Second, it ensures that the junction flux $\Phi_J$ is a coordinate within the synthesized network with a well-defined conjugate momentum. When $C_J$ is part of the network, the momentum $Q_J = \partial\mathcal{L}/\partial\dot{\Phi}_J$ arises naturally from the Lagrangian formulation, with $Q_J$ representing the charge on the total capacitance connected to the junction node. Third, it simplifies the boundary condition equation: the condition $sY_{\mathrm{in}}(s) + 1/L_J = 0$ for the dressed mode frequencies involves only the admittance and the Josephson inductance, with no additional terms. The alternative convention, excluding $C_J$ from $Y_{\mathrm{in}}(s)$ and adding a charging term, is mathematically valid if handled consistently. However, mixing elements of both conventions leads to double-counting. We maintain Assumption~\ref{assumption:cj_convention} throughout this paper. Although the charging energy $E_C = e^2/(2C_\Sigma)$ does not appear explicitly in the Hamiltonian under this convention, it determines the transmon spectrum through the ratio $E_J/E_C$. The capacitance $C_\Sigma$ entering $E_C$ is the total capacitance at the junction node, which equals the coefficient of $s$ in the high-frequency expansion of $Y_{\mathrm{in}}(s)$. The plasma frequency is $\omega_p = 1/\sqrt{L_J C_\Sigma} = \sqrt{8E_J E_C}/\hbar$, and the anharmonicity is $\alpha \approx -E_C/\hbar$.

\subsection{High-Frequency Behavior of the Admittance}
\label{sec:high_freq}

An important consequence of including $C_J$ in $Y_{\mathrm{in}}(s)$ is a constraint on the high-frequency behavior:
\begin{equation}
Y_{\mathrm{in}}(i\omega) = i\omega C_\Sigma + \mathcal{O}(1) \quad \text{as } \omega \to \infty,
\label{eq:high_freq_admittance}
\end{equation}
where $C_\Sigma \geq C_J > 0$ is the total shunt capacitance at the junction node. This behavior reflects the physical fact that at high frequencies, the capacitive admittance dominates over inductive or resistive contributions. The condition $C_\Sigma > 0$ is essential for the ultraviolet convergence proved in Section~\ref{sec:uv}. It ensures that the junction is not connected to a purely inductive environment, which would lead to divergent zero-point fluctuations. For any physical circuit containing a Josephson junction with finite $C_J > 0$, Assumption~\ref{assumption:cj_convention} automatically guarantees Eq.~\eqref{eq:high_freq_admittance}.

In the Cauer continued fraction representation (Section~\ref{sec:cauer_synthesis}), the high-frequency behavior $Y_{\mathrm{in}}(s) \sim sC_\Sigma$ corresponds to the leading term of a Type~II expansion, where the first element is a shunt capacitor $C_1 = C_\Sigma$. For a finite ladder with $N$ sections, the continued fraction terminates and the total shunt capacitance at the input is determined by the network topology. For an infinite ladder representing a continuum of modes, the continued fraction converges provided the coefficients satisfy appropriate growth conditions~\cite{Wall1948}. The physical requirement of finite junction capacitance translates to a convergent continued fraction, ensuring both ultraviolet regularity and well-defined spectral properties.


\section{The Continued Fraction Framework}
\label{sec:cf_framework}

This section develops the continued fraction representation as the fundamental mathematical structure underlying circuit quantization. We establish the equivalence between three perspectives: Cauer network synthesis from classical circuit theory, Jacobi matrix spectral theory from mathematical analysis, and the boundary condition formulation from Sturm-Liouville theory. This equivalence is not merely formal; it provides accurate computational methods across coupling regimes. Quantities traditionally obtained perturbatively (dispersive shifts, Lamb shifts, Purcell rates) can be computed nonperturbatively from the same underlying continued-fraction structure, and reduce to the usual perturbative formulas in their regimes of validity.

\subsection{Cauer Canonical Form and Ladder Networks}
\label{sec:cauer_synthesis}

The Cauer synthesis procedure, developed in the foundational works of network theory~\cite{Cauer1926, Cauer1958, Guillemin1957}, provides a canonical realization of any positive-real admittance function as a ladder network. We begin with the key definition.

\begin{definition}[Positive-Real Function]
\label{def:positive_real}
A function $Y(s)$ of the complex frequency $s = \sigma + i\omega$ is \emph{positive-real} if:
\begin{enumerate}
    \item[(i)] $Y(s)$ is analytic in the open right half-plane $\mathrm{Re}(s) > 0$;
    \item[(ii)] $Y(s^*) = Y^*(s)$ (reality condition); and
    \item[(iii)] $\mathrm{Re}[Y(s)] \geq 0$ for $\mathrm{Re}(s) > 0$.
\end{enumerate}
\end{definition}

Positive-real functions characterize passive linear networks: condition (iii) ensures that the network cannot generate power, since the time-averaged power dissipated is $P = \frac{1}{2}|V|^2 \mathrm{Re}[Y] \geq 0$. The fundamental theorem of network synthesis~\cite{Brune1931, Guillemin1957} states that every positive-real function is realizable as a passive network of resistors, inductors, and capacitors. For lossless networks (no resistors), the function is purely imaginary on the imaginary axis: $Y(i\omega) = iB(\omega)$ where $B(\omega)$ is the real-valued susceptance.

\begin{theorem}[Cauer Continued Fraction Expansion]
\label{thm:cauer_expansion}
Let $Y(s)$ be a real-rational driving-point admittance of a lossless passive one-port network with $N$ finite resonances. Then $Y(s)$ admits a unique Type~I continued fraction expansion:
\begin{equation}
Y(s) = \cfrac{1}{L_1 s + \cfrac{1}{C_1 s + \cfrac{1}{L_2 s + \cfrac{1}{C_2 s + \ddots}}}}
\label{eq:cauer_type1}
\end{equation}
or equivalently a unique Type~II expansion with a leading capacitive term:
\begin{equation}
Y(s) = C_0 s + \cfrac{1}{L_1 s + \cfrac{1}{C_1 s + \cfrac{1}{L_2 s + \ddots}}}.
\label{eq:cauer_type2}
\end{equation}
The coefficients $\{L_n, C_n\}_{n=1}^{N}$ are uniquely determined by $Y(s)$ within each form and are all positive.
\end{theorem}

\begin{proof}[Proof sketch]
The expansion follows from repeated pole extraction. For a lossless network, the impedance $Z(s) = 1/Y(s)$ has simple poles on the imaginary axis at the resonant frequencies $\{\pm i\omega_k\}$, with positive residues guaranteed by Foster's reactance theorem~\cite{Foster1924}. For Type~I, extract a series inductor by computing $L_1 = \lim_{s\to\infty} Z(s)/s$, then form the remainder $Y_1(s) = 1/(Z(s) - L_1 s)$. For Type~II, extract a shunt capacitor by computing $C_0 = \lim_{s\to\infty} Y(s)/s$, then form the remainder $Z_1(s) = 1/(Y(s) - C_0 s)$. Each extraction removes one resonance, reducing the rational function degree by two. The process terminates after $N$ steps for a network with $N$ resonances. Uniqueness within each form follows from the uniqueness of the extraction sequence---see Guillemin~\cite{Guillemin1957}, Chapter~4, for the complete treatment.
\end{proof}

\begin{figure}[t]
\centering
\begin{tikzpicture}[scale=0.9]
    \draw[thick] (0,0) node[left] {Port} -- (1,0);
    
    \draw[thick, decoration={coil, aspect=0.4, segment length=2.5mm, amplitude=1.5mm}, decorate] (1,0) -- (2.5,0);
    \node[above] at (1.75,0.3) {\footnotesize $L_1$};
    
    \draw[thick] (2.5,0) -- (2.5,-0.3);
    \draw[thick] (2.2,-0.3) -- (2.8,-0.3);
    \draw[thick] (2.2,-0.45) -- (2.8,-0.45);
    \draw[thick] (2.5,-0.45) -- (2.5,-0.8);
    \node[right] at (2.85,-0.37) {\footnotesize $C_1$};
    
    \draw[thick] (2.3,-0.8) -- (2.7,-0.8);
    \draw[thick] (2.35,-0.9) -- (2.65,-0.9);
    \draw[thick] (2.4,-1.0) -- (2.6,-1.0);
    
    \draw[thick] (2.5,0) -- (3.5,0);
    \draw[thick, decoration={coil, aspect=0.4, segment length=2.5mm, amplitude=1.5mm}, decorate] (3.5,0) -- (5,0);
    \node[above] at (4.25,0.3) {\footnotesize $L_2$};
    
    \draw[thick] (5,0) -- (5,-0.3);
    \draw[thick] (4.7,-0.3) -- (5.3,-0.3);
    \draw[thick] (4.7,-0.45) -- (5.3,-0.45);
    \draw[thick] (5,-0.45) -- (5,-0.8);
    \node[right] at (5.35,-0.37) {\footnotesize $C_2$};
    
    \draw[thick] (4.8,-0.8) -- (5.2,-0.8);
    \draw[thick] (4.85,-0.9) -- (5.15,-0.9);
    \draw[thick] (4.9,-1.0) -- (5.1,-1.0);
    
    \draw[thick] (5,0) -- (6,0);
    \node at (6.5,0) {$\cdots$};
    
    \node[below] at (3.25,-1.5) {Cauer Type I Ladder};
\end{tikzpicture}
\caption{Cauer Type~I canonical form realizing the admittance $Y(s)$ of Eq.~\eqref{eq:cauer_type1}. The network consists of a cascade of series inductors $L_k$ and shunt capacitors $C_k$. Each LC section represents one rung of the ladder, with the continued fraction structure reflecting the recursive topology: the admittance looking into rung $k$ depends on all subsequent rungs via the nested fraction. This nearest-neighbor coupling is the classical origin of the tridiagonal quantum Hamiltonian, upon quantization, the flux across inductor $L_k$ couples only to the charges on capacitors $C_{k-1}$ and $C_k$.}
\label{fig:cauer_ladder}
\end{figure}

The physical circuit corresponding to Eq.~\eqref{eq:cauer_type1} is shown in Fig.~\ref{fig:cauer_ladder}. The ladder topology directly encodes the continued fraction structure: each series inductor $L_k$ followed by a shunt capacitor $C_k$ corresponds to one level of nesting in the fraction. The impedance at the port depends recursively on the impedance of all downstream elements, precisely mirroring the nested structure of Eq.~\eqref{eq:cauer_type1}. This recursive, nearest-neighbor connectivity is preserved under quantization, yielding a nearest-neighbor (tridiagonal) structure in the canonical coordinate representation of the linear sector.

\subsection{Jacobi Matrices and the Resolvent}
\label{sec:jacobi_resolvent}

The connection between continued fractions and spectral theory is well-established in mathematical analysis, originating with Stieltjes' seminal work on the moment problem~\cite{Stieltjes1894}. The key insight is that continued fractions encode the spectral properties of tridiagonal matrices.

\begin{definition}[Jacobi Matrix]
\label{def:jacobi_matrix}
A \emph{Jacobi matrix} is a symmetric tridiagonal matrix of the form
\begin{equation}
\mathrm{J} = \begin{pmatrix}
a_0 & b_1 & 0 & 0 & \cdots \\
b_1 & a_1 & b_2 & 0 & \cdots \\
0 & b_2 & a_2 & b_3 & \cdots \\
\vdots & \vdots & \vdots & \vdots & \ddots
\end{pmatrix}
\label{eq:jacobi_matrix}
\end{equation}
with real diagonal elements $\{a_n\}_{n \geq 0}$ and positive off-diagonal elements $\{b_n > 0\}_{n \geq 1}$.
\end{definition}

Jacobi matrices arise naturally from orthogonal polynomial theory: the three-term recurrence relation satisfied by orthogonal polynomials $\{P_n(x)\}$ with respect to a measure $\mu$,
\begin{equation}
x P_n(x) = b_{n+1} P_{n+1}(x) + a_n P_n(x) + b_n P_{n-1}(x),
\label{eq:three_term_recurrence}
\end{equation}
corresponds precisely to the eigenvalue equation $\mathrm{J} \mathrm{v} = x \mathrm{v}$ for the Jacobi matrix with coefficients $\{a_n, b_n\}$~\cite{Chihara1978, Szego1939}. Here $b_n$ is the off-diagonal element connecting rows $n-1$ and $n$ in the matrix, consistent with Definition~\ref{def:jacobi_matrix}.

\begin{theorem}[Resolvent-Continued Fraction Identity]
\label{thm:resolvent_cf}
Let $\mathrm{J}$ be an $N \times N$ Jacobi matrix with elements $\{a_n, b_n\}$, and let $|0\rangle$ denote the first basis vector. The $(0,0)$ element of the resolvent $(z\mathrm{I} - \mathrm{J})^{-1}$ is given by the continued fraction
\begin{multline}
G(z) \equiv 
\langle 0 | (z\mathrm{I} - \mathrm{J})^{-1} | 0 \rangle \\
= \cfrac{1}{z - a_0
 - \cfrac{b_1^2}{z - a_1
 - \cfrac{b_2^2}{z - a_2
 - \ddots
 - \cfrac{b_{N-1}^2}{z - a_{N-1}}}}}
\label{eq:resolvent_cf}
\end{multline}
The poles of $G(z)$ are the eigenvalues $\{E_k\}$ of $\mathrm{J}$, and the residues encode the squared amplitudes of the normalized eigenvectors at site $0$.
\end{theorem}

\begin{proof}
Partition the Jacobi matrix as
\begin{equation}
\mathrm{J} = \begin{pmatrix}
a_0 & \mathrm{v}^\top \\
\mathrm{v} & \mathrm{J}'
\end{pmatrix},
\end{equation}
where $\mathrm{v} = (b_1, 0, 0, \ldots)^\top$ and $\mathrm{J}'$ is the $(N-1) \times (N-1)$ trailing principal submatrix. By the Schur complement formula for block matrix inversion,
\begin{equation}
[(z\mathrm{I} - \mathrm{J})^{-1}]_{00} = \frac{1}{z - a_0 - \mathrm{v}^\top (z\mathrm{I}' - \mathrm{J}')^{-1} \mathrm{v}}.
\end{equation}
Since $\mathrm{v}$ has only one non-zero element $b_1$ in the first position,
\begin{equation}
\mathrm{v}^\top (z\mathrm{I}' - \mathrm{J}')^{-1} \mathrm{v} = b_1^2 [(z\mathrm{I}' - \mathrm{J}')^{-1}]_{00} = b_1^2 G'(z),
\end{equation}
where $G'(z)$ is the $(0,0)$ resolvent of the reduced matrix $\mathrm{J}'$. This yields the recursion $G(z) = 1/(z - a_0 - b_1^2 G'(z))$. Iterating $N-1$ times, with the terminal condition $G_{N-1}(z) = 1/(z - a_{N-1})$ for the $1 \times 1$ matrix, produces the continued fraction~\eqref{eq:resolvent_cf}.
\end{proof}

he spectral content of $G(z)$ is revealed by its partial fraction expansion:
\begin{equation}
G(z) = \sum_{k=1}^{N} \frac{|\psi_k(0)|^2}{z - E_k},
\label{eq:resolvent_spectral}
\end{equation}
where $\{E_k\}_{k=1}^{N}$ are the eigenvalues of $\mathrm{J}$ and $\psi_k(0) = \langle 0 | E_k \rangle$ is the first component of the $k$-th normalized eigenvector (satisfying $\sum_{n=0}^{N-1} |\psi_k(n)|^2 = 1$). In the circuit QED context, the eigenvalues correspond to dressed mode frequencies, and the squared amplitudes $|\psi_k(0)|^2$ encode the port weight of each mode at the junction node. This connection to participation ratios underlies the Purcell decay formula discussed in Section~\ref{sec:complex_poles}.

\subsection{Equivalence of Representations}
\label{sec:cf_equivalence}

We now establish the central equivalence connecting classical network theory to circuit QED. This theorem unifies the three perspectives introduced at the beginning of this section.

\begin{theorem}[Network-Spectral Equivalence]
\label{thm:network_spectral}
Let $Y(s)$ be the driving-point admittance of a lossless passive one-port network with $N$ resonances. The following three characterizations are equivalent:
\begin{enumerate}
    \item[(i)] $Y(s)$ admits the Cauer continued fraction expansion~\eqref{eq:cauer_type1} with positive coefficients $\{L_n, C_n\}_{n=1}^{N}$.
    \item[(ii)] The function $G(\omega^2) \equiv -Y(i\omega)/(i\omega)$ equals the resolvent~\eqref{eq:resolvent_cf} of an $N \times N$ Jacobi matrix $\mathrm{J}$ whose elements are determined by the LC parameters through the mapping given in Appendix~\ref{app:lc_jacobi}.
    \item[(iii)] The eigenfrequencies of the coupled junction-environment system (in the linearized, harmonic approximation) are the positive roots $\omega > 0$ of the boundary condition
    \begin{equation}
    i\omega Y(i\omega) + \frac{1}{L_J} = 0,
    \label{eq:boundary_condition_thm}
    \end{equation}
    where $L_J = \varphi_0^2/E_J$ is the Josephson inductance.
\end{enumerate}
\end{theorem}

\begin{proof}
{(i) $\Leftrightarrow$ (ii):} For a lossless network, $Y(i\omega) = iB(\omega)$ where $B(\omega)$ is the real susceptance. Define $G(\omega^2) = B(\omega)/\omega$, which is well-defined for $\omega \neq 0$. The Cauer expansion~\eqref{eq:cauer_type1} evaluated at $s = i\omega$ gives
\begin{equation}
Y(i\omega) = \cfrac{1}{i\omega L_1 + \cfrac{1}{i\omega C_1 + \cfrac{1}{i\omega L_2 + \ddots}}}.
\end{equation}
Extracting the imaginary part and dividing by $\omega$, we obtain a continued fraction in $\omega^2$ that has the structure of Eq.~\eqref{eq:resolvent_cf} with $z = \omega^2$. The explicit mapping between Cauer parameters $\{L_n, C_n\}$ and Jacobi matrix elements $\{a_n, b_n\}$ involves products of adjacent LC values and is given in Appendix~\ref{app:lc_jacobi}. The correspondence is bijective: given $\{a_n, b_n\}$, one recovers $\{L_n, C_n\}$ uniquely, and vice versa.

{(ii) $\Leftrightarrow$ (iii):} The boundary condition~\eqref{eq:boundary_condition_thm} can be rewritten as
\begin{equation}
\omega B(\omega) + \frac{1}{L_J} = 0 \quad \Longleftrightarrow \quad G(\omega^2) = -\frac{1}{\omega^2 L_J}.
\end{equation}
Using the spectral representation~\eqref{eq:resolvent_spectral},
\begin{equation}
\sum_{k=1}^{N} \frac{|\psi_k(0)|^2}{\omega^2 - \omega_k^2} = -\frac{1}{\omega^2 L_J},
\end{equation}
where $\omega_k^2 = E_k$ are the eigenvalues of the Jacobi matrix (squared bare frequencies). This equation has $N+1$ positive roots $\{\tilde{\omega}_m\}_{m=0}^{N}$, which are the dressed mode frequencies. By the Cauchy interlacing theorem, these roots strictly interlace with the bare frequencies: $0 < \tilde{\omega}_0 < \omega_1 < \tilde{\omega}_1 < \omega_2 < \cdots < \omega_N < \tilde{\omega}_N$. The dressed frequencies are equivalently the eigenvalues of the $(N+1) \times (N+1)$ Jacobi matrix obtained by augmenting $\mathrm{J}$ with the junction degree of freedom.
\end{proof}

This theorem reveals that the circuit engineer synthesizing $Y(s)$ as a ladder network, the mathematician applying Jacobi matrix spectral theory, and the physicist solving the boundary condition for dressed modes are all performing equivalent calculations in different representations. The continued fraction is the unifying mathematical structure.

\subsection{Extension to the Full Nonlinear Problem: The Quantum Rabi Model}
\label{sec:qrm_parity}

The continued fraction structure extends beyond the linear (harmonic) sector treated above to include the full nonlinear dynamics of a qubit coupled to a cavity. While Theorem~\ref{thm:network_spectral} applies to the linearized junction (replaced by inductance $L_J$), the same tridiagonal structure appears in the exact quantum problem when the Josephson cosine potential is retained.

The quantum Rabi model (QRM)~\cite{Rabi1936, Rabi1937} describes a two-level system coupled to a harmonic oscillator without the rotating-wave approximation:
\begin{equation}
\hat{H}_{\mathrm{QRM}} = \hbar\omega_c \hat{a}^\dagger \hat{a} + \frac{\hbar\omega_q}{2} \hat{\sigma}_z + \hbar g(\hat{a} + \hat{a}^\dagger)\hat{\sigma}_x.
\label{eq:qrm_hamiltonian}
\end{equation}
Unlike the Jaynes-Cummings model, the QRM does not conserve excitation number $\hat{N} = \hat{a}^\dagger\hat{a} + \frac{1}{2}(1 + \hat{\sigma}_z)$. However, it does conserve parity $\hat{P} = e^{i\pi \hat{a}^\dagger \hat{a}} \hat{\sigma}_z$, with eigenvalues $\pm 1$~\cite{Casanova2010}.

\begin{proposition}[Tridiagonal Structure of the QRM]
\label{prop:qrm_tridiagonal}
In the parity eigenbasis, the QRM Hamiltonian restricted to each parity sector is tridiagonal. Specifically:
\begin{itemize}
    \item \textit{Even parity} ($\hat{P} = +1$): Basis states $|+,0\rangle, |-,1\rangle, |+,2\rangle, |-,3\rangle, \ldots$
    \item \textit{Odd parity} ($\hat{P} = -1$): Basis states $|-,0\rangle, |+,1\rangle, |-,2\rangle, |+,3\rangle, \ldots$
\end{itemize}
where $|\pm\rangle = (|g\rangle \pm |e\rangle)/\sqrt{2}$ are eigenstates of $\hat{\sigma}_x$. In each sector, the Hamiltonian matrix has diagonal elements
\begin{equation}
d_n = \hbar\omega_c n + (-1)^{n+p} \frac{\hbar\omega_q}{2},
\label{eq:qrm_diagonal}
\end{equation}
where $p = 0$ for even parity and $p = 1$ for odd parity, and off-diagonal elements
\begin{equation}
t_n = \hbar g\sqrt{n+1}.
\label{eq:qrm_offdiagonal}
\end{equation}
\end{proposition}

\begin{proof}
The parity operator $\hat{P} = e^{i\pi \hat{a}^\dagger\hat{a}}\hat{\sigma}_z$ satisfies $[\hat{P}, \hat{H}_{\mathrm{QRM}}] = 0$, verified by noting that $e^{i\pi \hat{a}^\dagger\hat{a}} \hat{a} e^{-i\pi \hat{a}^\dagger\hat{a}} = -\hat{a}$ and $\hat{\sigma}_z \hat{\sigma}_x \hat{\sigma}_z = -\hat{\sigma}_x$. The Hilbert space therefore decomposes as $\mathcal{H} = \mathcal{H}_+ \oplus \mathcal{H}_-$.

In the $\hat{\sigma}_x$ eigenbasis with Fock states, a general state $|\sigma, n\rangle$ (where $\sigma = \pm$ and $n \in \{0, 1, 2, \ldots\}$) has parity $\hat{P}|\sigma, n\rangle = \sigma(-1)^n |\sigma, n\rangle$. States with even parity satisfy $\sigma(-1)^n = +1$, yielding the sequence $|+,0\rangle, |-,1\rangle, |+,2\rangle, \ldots$; odd parity gives $|-,0\rangle, |+,1\rangle, |-,2\rangle, \ldots$.

The coupling term $\hbar g(\hat{a} + \hat{a}^\dagger)\hat{\sigma}_x$ acts as $\hbar g(\hat{a} + \hat{a}^\dagger)$ on $|+\rangle$ states and $-\hbar g(\hat{a} + \hat{a}^\dagger)$ on $|-\rangle$ states (since $\hat{\sigma}_x|\pm\rangle = \pm|\pm\rangle$). Since $\hat{a}$ and $\hat{a}^\dagger$ only connect $|n\rangle \leftrightarrow |n \pm 1\rangle$, and consecutive states in each parity chain alternate between $|+\rangle$ and $|-\rangle$, the coupling is strictly nearest-neighbor within each chain. The matrix elements follow from $\langle n \pm 1|(\hat{a} + \hat{a}^\dagger)|n\rangle = \sqrt{n+1}$ or $\sqrt{n}$.
\end{proof}

The tridiagonal structure means the QRM eigenproblem in each parity sector is equivalent to a continued fraction equation. The eigenvalue condition $\det(E\mathrm{I} - \mathrm{H}_{\pm}) = 0$ becomes a transcendental equation whose roots give exact eigenvalues at arbitrary coupling strength $g/\omega_c$. This includes the ultrastrong coupling (USC, $g/\omega_c \gtrsim 0.1$) and deep strong coupling (DSC, $g/\omega_c \gtrsim 1$) regimes where the rotating-wave approximation fails~\cite{FornDiaz2019, Kockum2019}. The analytic solution of the QRM via continued fractions was achieved by Braak~\cite{Braak2011}, who showed that the spectrum is determined by zeros of transcendental ``$G$-functions'' constructed from the continued fraction. This breakthrough demonstrated that the QRM, despite having only a discrete $\mathbb{Z}_2$ symmetry, is exactly solvable.

\subsection{Dressed Operators and Observable-Dependent Spectra}
\label{sec:dressed_operators}

In the USC and DSC regimes, the bare operators $\hat{a}, \hat{a}^\dagger, \hat{\sigma}_\pm$ no longer correspond to physical observables. The ground state $|G\rangle$ of the full Hamiltonian contains a nonzero population of bare photons: $\langle G | \hat{a}^\dagger\hat{a} | G \rangle > 0$. Physical observables must be constructed from the exact eigenstates~\cite{Kockum2019, Napoli2024}.

\begin{definition}[Dressed Operators]
\label{def:dressed_operators}
Let $\{|E_k\rangle\}_{k=0}^{\infty}$ be the exact eigenstates of the full Hamiltonian, ordered by energy. The \emph{dressed annihilation-like operator} is
\begin{equation}
\hat{X}^{+} = \sum_{j > i} X_{ij} |E_i\rangle \langle E_j|, \quad X_{ij} = \langle E_i | (\hat{a} + \hat{a}^\dagger) | E_j \rangle,
\label{eq:dressed_X}
\end{equation}
which only contains lowering transitions ($E_i < E_j$). Similarly, the dressed momentum operator is
\begin{equation}
\hat{P}^{+} = \sum_{j > i} P_{ij} |E_i\rangle \langle E_j|, \quad P_{ij} = i\langle E_i | (\hat{a}^\dagger - \hat{a}) | E_j \rangle.
\label{eq:dressed_P}
\end{equation}
\end{definition}

\begin{table*}[t]
\caption{The continued fraction framework for circuit QED: correspondence between mathematical structure and physical quantities (left), and key observables computed from the framework (right). Here $G(z) = \langle 0|(\mathrm{I}z-\mathrm{J})^{-1}|0\rangle$ is the Jacobi resolvent, $G_\pm(E)$ are the Braak $G$-functions for the quantum Rabi model, $\hat{n}$ is the Cooper pair number operator, $V_\mathrm{zpf} = \sqrt{\hbar\omega_r/2C_r}$ is the resonator zero-point voltage, and $\Delta = \omega_q - \omega_r$ is the qubit-resonator detuning. The dispersive shift formula $\chi \approx g^2/\Delta - g^2/(\Delta+\alpha)$ is valid in the dispersive regime $g \ll |\Delta|$; see Section~\ref{sec:transcendental} for the general expression.}
\label{tab:cf_framework}
\begin{ruledtabular}
\begin{tabular}{ll|ll}
\multicolumn{2}{c|}{\textrm{CF--Physical Dictionary}} & 
\multicolumn{2}{c}{\textrm{Observables}} \\
\midrule
Admittance $Y(s)$ & Environment seen by junction & 
Qubit frequency $\omega_{01}$ & $(E_1 - E_0)/\hbar$ \\[2pt]
CF poles $\{\omega_k\}$ & Bare mode frequencies & 
Anharmonicity $\alpha$ & $\omega_{12} - \omega_{01}$ \\[2pt]
CF residues $\{|c_k|^2\}$ & Mode weights at port & 
Dispersive shift $\chi$ & $g^2/\Delta - g^2/(\Delta+\alpha)$ (dispersive) \\[2pt]
Jacobi diagonal $\{a_n\}$ & On-site energies & 
Coupling $g$ & $(2e/\hbar)|\langle 0|\hat{n}|1\rangle|\,V_\mathrm{zpf}$ \\[2pt]
Jacobi off-diagonal $\{b_n\}$ & Nearest-neighbor couplings & 
Purcell rate $\gamma_\kappa$ & $\kappa(g/\Delta)^2$ \\[2pt]
Poles of $G(z)$ & Eigenvalues (linear systems) & 
Lamb shift $\delta\omega$ & $\mathrm{Re}[\Sigma(\omega)]$; see Eq.~\eqref{eq:frequency_shift_multimode} \\[2pt]
Zeros of $G_\pm(E)$ & Eigenvalues (QRM, parity $\pm$) & & \\[2pt]
$|\psi_k(0)|^2$ & Junction participation in mode $k$ & & \\[2pt]
CF truncation at level $N$ & $N$-mode approximation & & \\
\end{tabular}
\end{ruledtabular}
\end{table*}

The matrix elements $X_{ij}$ and $P_{ij}$ are computed directly from the continued fraction eigenvectors. A crucial result from Refs.~\cite{Kockum2019, Napoli2024} is that transition rates depend on which physical operator couples the system to the measurement apparatus:
\begin{equation}
\Gamma_{i \to j}^{(\mathrm{flux})} \propto |X_{ij}|^2, \quad \Gamma_{i \to j}^{(\mathrm{charge})} \propto |P_{ij}|^2.
\label{eq:dressed_rates}
\end{equation}
For inductive (flux) coupling to a transmission line, the rate is governed by $|X_{ij}|^2$; for capacitive (charge) coupling, by $|P_{ij}|^2$. In the QRM parity-chain representation, these matrix elements follow from the tridiagonal eigenvectors. More generally, in the circuit Hamiltonian one computes them from the exact eigenstates obtained by diagonalization, with the continued-fraction resolvent providing an efficient route to the required spectral weights.

In the weak coupling regime ($g/\omega_c \ll 0.1$), the distinction is immaterial: $|X_{01}|^2 \approx |P_{01}|^2 \approx 1$ for the fundamental $|E_0\rangle \to |E_1\rangle$ transition. In USC/DSC, however, the matrix elements can differ dramatically. Generically $|X_{01}|^2 \gg |P_{01}|^2$, leading to observable-dependent spectra: the same physical system exhibits different transition rates depending on the measurement scheme. This phenomenon has been observed experimentally in flux qubit systems~\cite{FornDiaz2017} and is a direct consequence of the continued fraction structure encoding the exact eigenstates. 


\subsection{Summary: The Continued Fraction Dictionary}
\label{sec:cf_summary}

Table~\ref{tab:cf_framework} summarizes the continued fraction framework, establishing a dictionary between mathematical structure and physical content. The left columns translate continued fraction elements into circuit QED quantities: the input admittance $Y(s)$ encodes the electromagnetic environment seen by the Josephson junction, its poles give the bare mode frequencies, and the associated residues quantify each mode's weight at the junction port. Upon quantization, the Cauer ladder maps to a Jacobi (tridiagonal) Hamiltonian whose diagonal elements are on-site energies and whose off-diagonal elements are nearest-neighbor couplings. The eigenvalue structure depends on the system: for linear networks, eigenvalues appear as poles of the resolvent $G(z)$, while for the quantum Rabi model they are zeros of Braak's transcendental $G$-functions in each parity sector. In both cases, the eigenvector component $|\psi_k(0)|^2$ at site zero gives the junction's participation in eigenmode $k$, the key quantity governing transition matrix elements and decay rates. 

The right columns of Table~\ref{tab:cf_framework} collect the experimentally relevant observables that emerge from this framework. The qubit frequency $\omega_{01}$ and anharmonicity $\alpha$ follow directly from the first few eigenvalues. The dispersive shift $\chi = g^2/\Delta - g^2/(\Delta+\alpha)$, which enables quantum nondemolition readout, arises from the interplay of the $|0\rangle \leftrightarrow |1\rangle$ and $|1\rangle \leftrightarrow |2\rangle$ virtual transitions---for a harmonic oscillator these would cancel exactly, but the transmon's anharmonicity leaves a finite residual. The Purcell decay rate $\gamma_\kappa = \kappa(g/\Delta)^2$ in the dispersive limit, and more generally the Lamb shift, are likewise determined by the framework's eigenstructure. 

The continued fraction framework thus provides a unified computational approach to circuit QED that: (1) reduces to standard perturbative results in the dispersive regime; (2) extends seamlessly to USC/DSC where perturbation theory fails; and (3) connects directly to classical network synthesis, enabling the design of circuits with targeted spectral properties.


\section{The Schur Complement as Boundary Condition}
\label{sec:schur}

The central mathematical operation in reducing a multiport electromagnetic environment to a single-port description is the Schur complement. This section derives the Schur complement formula and proves that it yields precisely the frequency-dependent boundary condition at the junction node. This equivalence connects standard circuit reduction to eigenparameter-dependent boundary conditions in the sense of Sturm-Liouville theory~\cite{Zettl2005, Dijksma1987}. The boundary-condition viewpoint was developed for dispersive readout in a transmon-resonator setting in Ref.~\cite{BakrQM2025}; here we show the same structure arises generally from the driving-point admittance $Y_{\mathrm{in}}(s)$. The Schur complement reduction is intimately connected to the continued fraction structure developed in Section~\ref{sec:cf_framework}. When the admittance matrix arises from a ladder network, the Schur complement that eliminates internal nodes reconstructs the continued fraction expansion term by term. For a Cauer ladder with $N$ sections, eliminating the $k$-th internal node via Schur complement is equivalent to evaluating the $k$-th level of the continued fraction. This connection to Belevitch's cascade synthesis for multiport networks~\cite{Belevitch1960} provides a unified perspective on network reduction and spectral analysis.

\begin{figure*}[t]
\centering
\begin{tikzpicture}[scale=0.85]
    \begin{scope}[shift={(0,0)}]
        \node[draw, thick, rounded corners, minimum width=3.2cm, minimum height=2.8cm, fill=gray!8] (net1) at (0,0) {};
        \node[font=\small] at (0,0.4) {Multiport};
        \node[font=\small] at (0,-0.1) {Network};
        \node[font=\footnotesize, black] at (0,-0.7) {$\mathrm{Y}(s)$};
        
        \draw[thick] (-2.8,0) -- (-1.6,0);
        \node[circle, fill=black, inner sep=1.5pt] at (-1.6,0) {};
        \node[left, font=\small] at (-2.8,0) {$J$};
        
        \draw[thick] (1.6,0) -- (2.8,0);
        \node[circle, fill=black, inner sep=1.5pt] at (1.6,0) {};
        
        \node[below, font=\small] at (0,-1.9) {(a) Full multiport network};
    \end{scope}
    
    \draw[-{Stealth[length=3mm]}, very thick] (3.5,0) -- (5.5,0);
    \node[above, font=\footnotesize] at (4.5,0.3) {Schur};
    \node[below, font=\footnotesize] at (4.5,-0.3) {complement};
    
    \begin{scope}[shift={(8.5,0)}]
        \node[draw, thick, rounded corners, minimum width=2.8cm, minimum height=2.2cm, fill=blue!8] (net2) at (0,0) {};
        \node[font=\small] at (0,0.2) {$Y_{\mathrm{in}}(s)$};
        \node[font=\footnotesize, gray] at (0,-0.3) {Continued fraction};
        
        \draw[thick] (-2.5,0) -- (-1.4,0);
        \node[circle, fill=black, inner sep=1.5pt] at (-1.4,0) {};
        \node[left, font=\small] at (-2.5,0) {$J$};
        
        \node[below, font=\small] at (0,-1.9) {(b) Driving-point admittance};
    \end{scope}
\end{tikzpicture}
\caption{Schur complement reduction. (a) A multiport network with junction port $J$ and internal nodes $\bar{J}$. (b) The Schur complement eliminates internal nodes, yielding the scalar driving-point admittance $Y_{\mathrm{in}}(s)$ at the junction. The positive-real property is preserved under this reduction. For ladder networks, the Schur complement reconstructs the Cauer continued fraction expansion.}
\label{fig:schur_complement}
\end{figure*}
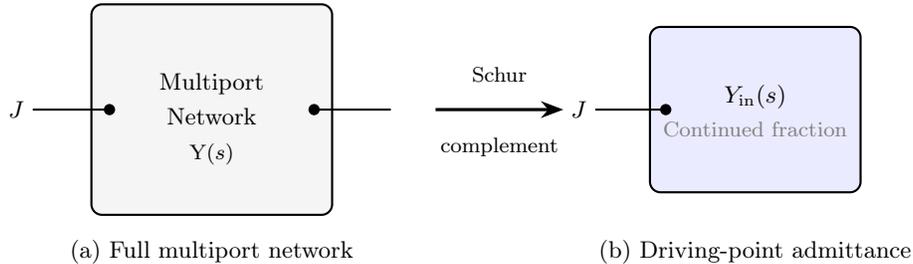

\subsection{Multiport Admittance and the Schur Complement}
\label{sec:multiport}
Consider a linear electromagnetic network with $N$ ports, characterized by an admittance matrix $\mathrm{Y}(s)$ relating port currents to port voltages in the Laplace domain:
\begin{equation}
\mathrm{I}(s) = \mathrm{Y}(s) \cdot \mathrm{V}(s).
\label{eq:admittance_relation}
\end{equation}
We partition the ports into two groups: the junction port $J$, and all remaining ports collectively denoted $\bar{J}$. The admittance matrix takes the block form
\begin{equation}
\mathrm{Y}(s) = \begin{pmatrix} Y_{JJ}(s) & \mathrm{Y}_{J\bar{J}}(s) \\[4pt] \mathrm{Y}_{\bar{J}J}(s) & \mathrm{Y}_{\bar{J}\bar{J}}(s) \end{pmatrix},
\label{eq:Y_partition}
\end{equation}
where $Y_{JJ}$ is the scalar self-admittance at port $J$, $\mathrm{Y}_{J\bar{J}}$ is a row vector of transfer admittances from ports $\bar{J}$ to port $J$, $\mathrm{Y}_{\bar{J}J}$ is a column vector of transfer admittances from port $J$ to ports $\bar{J}$, and $\mathrm{Y}_{\bar{J}\bar{J}}$ is the admittance matrix among the remaining ports. Throughout this section, we use boldface for matrices and vectors to distinguish them from scalar quantities.
The constitutive relation for the full network is
\begin{equation}
\begin{pmatrix} I_J \\ \mathrm{I}_{\bar{J}} \end{pmatrix} = \begin{pmatrix} Y_{JJ} & \mathrm{Y}_{J\bar{J}} \\[4pt] \mathrm{Y}_{\bar{J}J} & \mathrm{Y}_{\bar{J}\bar{J}} \end{pmatrix} \begin{pmatrix} V_J \\ \mathrm{V}_{\bar{J}} \end{pmatrix}.
\label{eq:constitutive_full}
\end{equation}

Now suppose the ports $\bar{J}$ are internal nodes of the network with no external current injection. This constraint is expressed by setting $\mathrm{I}_{\bar{J}} = \mathrm{0}$, which from the second block row of Eq.~\eqref{eq:constitutive_full} gives
\begin{equation}
\mathrm{0} = \mathrm{Y}_{\bar{J}J}(s) V_J + \mathrm{Y}_{\bar{J}\bar{J}}(s) \mathrm{V}_{\bar{J}}.
\label{eq:internal_constraint}
\end{equation}
Solving for the internal voltages yields
\begin{equation}
\mathrm{V}_{\bar{J}} = -\mathrm{Y}_{\bar{J}\bar{J}}^{-1}(s) \mathrm{Y}_{\bar{J}J}(s) V_J.
\label{eq:internal_voltages}
\end{equation}
Substituting into the first block row gives
\begin{align}
I_J &= Y_{JJ} V_J + \mathrm{Y}_{J\bar{J}} \mathrm{V}_{\bar{J}} \nonumber \\
&= Y_{JJ} V_J - \mathrm{Y}_{J\bar{J}} \mathrm{Y}_{\bar{J}\bar{J}}^{-1} \mathrm{Y}_{\bar{J}J} V_J \nonumber \\
&= \left[ Y_{JJ} - \mathrm{Y}_{J\bar{J}} \mathrm{Y}_{\bar{J}\bar{J}}^{-1} \mathrm{Y}_{\bar{J}J} \right] V_J.
\label{eq:schur_derivation}
\end{align}
The quantity in brackets is the Schur complement of $\mathrm{Y}_{\bar{J}\bar{J}}$ in $\mathrm{Y}$:
\begin{equation}
Y_{\mathrm{in}}(s) = Y_{JJ}(s) - \mathrm{Y}_{J\bar{J}}(s) \, \mathrm{Y}_{\bar{J}\bar{J}}^{-1}(s) \, \mathrm{Y}_{\bar{J}J}(s).
\label{eq:schur_complement}
\end{equation}

This formula has a direct physical interpretation. The first term $Y_{JJ}$ is the self-admittance at port $J$ with all other ports open-circuited. The second term accounts for the current that flows from port $J$ through the network to the other ports, induces voltages there according to the internal admittance matrix, and returns to port $J$ through the transfer admittances. The Schur complement thus represents the effective admittance seen by an observer at port $J$ when all other ports are internal to the network with no external excitation.

Figure~\ref{fig:schur_complement} illustrates this reduction schematically: the multiport network with internal nodes is reduced to a single driving-point admittance at the junction port. The key insight is that no information relevant to the junction dynamics is lost in this reduction; the entire influence of the electromagnetic environment on the junction is encoded in the scalar function $Y_{\mathrm{in}}(s)$. For a ladder network, the Schur complement reduction is equivalent to one step in the continued fraction evaluation. The recursive structure $Y_{\mathrm{reduced}}(s) = Y_{00}(s) - Y_{01}(s)Y_{10}(s)/[Y_{11}(s) + Y_{\mathrm{tail}}(s)]$ is precisely the CF recursion. Repeated application eliminates all internal nodes, yielding the full continued fraction expansion.

\subsection{Positive-Real Preservation}
\label{sec:pr_preservation}

The Schur complement preserves the positive-real property, which is essential for physical realizability of the reduced network.

\begin{proposition}[Positive-Real Preservation]
\label{prop:pr_preservation}
If the admittance matrix $\mathrm{Y}(s)$ is positive-real, representing a passive network, then the Schur complement $Y_{\mathrm{in}}(s)$ is also positive-real.
\end{proposition}

\begin{proof}
A matrix-valued function $\mathrm{Y}(s)$ is positive-real if it satisfies three conditions: analyticity in the open right half-plane $\mathrm{Re}(s) > 0$, the reality condition $\mathrm{Y}(s^*) = \mathrm{Y}^*(s)$, and the positive semidefiniteness condition $\mathrm{Y}(s) + \mathrm{Y}^\dagger(s) \geq 0$ for $\mathrm{Re}(s) > 0$.

For the Schur complement, analyticity in the right half-plane follows from analyticity of all component matrices and the fact that $\mathrm{Y}_{\bar{J}\bar{J}}(s)$ is invertible for $\mathrm{Re}(s) > 0$ when $\mathrm{Y}(s)$ is positive-real (since positive-real matrices have positive-definite real parts in the right half-plane). The reality condition is inherited directly from the block components.

For positive semidefiniteness, we use the characterization of Schur complements for positive semidefinite matrices. When $\mathrm{Y}(s) + \mathrm{Y}^\dagger(s) \geq 0$, the matrix can be written in block form with the $(1,1)$ block being $Y_{JJ} + Y_{JJ}^*$ and the $(2,2)$ block being $\mathrm{Y}_{\bar{J}\bar{J}} + \mathrm{Y}_{\bar{J}\bar{J}}^\dagger$. By the theory of Schur complements for Hermitian matrices~\cite{Zhang2005}, if the full matrix is positive semidefinite and the $(2,2)$ block is positive definite, then the Schur complement of the $(2,2)$ block satisfies the same semidefiniteness. Applying this to $\mathrm{Y}(s) + \mathrm{Y}^\dagger(s)$ yields
\begin{equation}
Y_{\mathrm{in}}(s) + Y_{\mathrm{in}}^*(s) \geq 0
\end{equation}
for $\mathrm{Re}(s) > 0$, completing the proof.
\end{proof}

This result is physically essential: it guarantees that $Y_{\mathrm{in}}(s)$ represents a realizable passive network regardless of the complexity of the original multiport structure. The subsequent network synthesis procedure can therefore construct an equivalent lumped-element circuit. Since the Schur complement preserves positive-realness, the Cauer continued fraction coefficients $\{L_n, C_n\}$ of $Y_{\mathrm{in}}(s)$ are all positive, ensuring that the synthesized ladder has physical (positive-valued) elements, the Jacobi matrix has positive off-diagonal elements, the eigenvalue problem is well-posed with real distinct eigenvalues, and the continued fraction converges.

\subsection{Connection to the Boundary Condition}
\label{sec:bc_connection}

We now establish the central result of this section: the Schur complement reduction produces the frequency-dependent boundary condition at the junction node. This connection reveals that circuit reduction and spectral theory are two perspectives on the same mathematical structure.

Consider the junction node in the synthesized circuit. Kirchhoff's current law states that the sum of currents entering the node must vanish. The current from the linear electromagnetic environment, flowing through the admittance $Y_{\mathrm{in}}$, is related to the junction voltage by
\begin{equation}
I_{\mathrm{env}}(t) = \mathcal{L}^{-1}\left[ Y_{\mathrm{in}}(s) \tilde{V}_J(s) \right],
\label{eq:env_current}
\end{equation}
where $\mathcal{L}^{-1}$ denotes the inverse Laplace transform. In the Laplace domain, using the voltage-flux relation $V_J = s\Phi_J$ (with zero initial conditions as discussed in Section~\ref{sec:node_flux}), this becomes
\begin{equation}
\tilde{I}_{\mathrm{env}}(s) = Y_{\mathrm{in}}(s) \tilde{V}_J(s) = s \, Y_{\mathrm{in}}(s) \tilde{\Phi}_J(s).
\label{eq:env_current_laplace}
\end{equation}

The current through the Josephson junction is determined by the first Josephson relation. Under our convention that $C_J$ is included in $Y_{\mathrm{in}}$ (Assumption~\ref{assumption:cj_convention}), the junction contributes only the supercurrent:
\begin{equation}
I_J(t) = I_c \sin\left(\frac{\Phi_J(t)}{\varphi_0}\right).
\label{eq:junction_current}
\end{equation}
Kirchhoff's current law at the junction node requires
\begin{equation}
I_{\mathrm{env}}(t) + I_J(t) = 0.
\label{eq:kcl}
\end{equation}

For the \emph{linearized} normal modes of the system, we seek solutions where the junction flux oscillates harmonically: $\Phi_J(t) = \phi_n^J e^{i\omega_n t}$ with small amplitude $|\phi_n^J| \ll \varphi_0$. In this regime, the Josephson current can be linearized as $\sin(\Phi_J/\varphi_0) \approx \Phi_J/\varphi_0$, giving $I_J \approx \Phi_J/L_J$ where $L_J = \varphi_0/I_c = \varphi_0^2/E_J$ is the Josephson inductance. Taking the Laplace transform of the linearized Kirchhoff equation yields
\begin{equation}
s \, Y_{\mathrm{in}}(s) \tilde{\Phi}_J(s) + \frac{\tilde{\Phi}_J(s)}{L_J} = 0.
\label{eq:linearized_kcl}
\end{equation}
For nontrivial solutions with $\tilde{\Phi}_J \neq 0$, the frequency must satisfy the boundary condition equation
\begin{equation}
s \, Y_{\mathrm{in}}(s) + \frac{1}{L_J} = 0.
\label{eq:boundary_condition}
\end{equation}

\begin{theorem}[Boundary Condition Equivalence]
\label{thm:bc_equivalence}
Let $\mathrm{Y}(s)$ be the nodal admittance matrix of a passive linear network, and let $Y_{\mathrm{in}}(s)$ be the Schur complement obtained by eliminating all nodes except a designated port node $J$. If a Josephson junction with inductance $L_J = \varphi_0^2/E_J$ terminates this port, then in the linearized (harmonic) approximation, the dressed mode frequencies $\omega_n$ of the coupled system are the positive solutions of
\begin{equation}
sY_{\mathrm{in}}(s) + \frac{1}{L_J} = 0 \quad \text{at } s = i\omega_n.
\end{equation}
This equation is the spectral condition imposed by the junction termination, with $Y_{\mathrm{in}}(s)$ encoding the frequency-dependent boundary data.
\end{theorem}

The boundary condition~\eqref{eq:boundary_condition} has a natural interpretation in Sturm-Liouville theory. For a differential operator on an interval with a boundary condition that depends on the eigenvalue parameter, the eigenvalues are determined by a transcendental equation involving the boundary data. Here, the ``boundary'' is the junction node, and the ``eigenvalue-dependent boundary condition'' is precisely Eq.~\eqref{eq:boundary_condition}, where the admittance $Y_{\mathrm{in}}(s)$ encodes how the electromagnetic environment responds at frequency $s$. This connection to Sturm-Liouville theory with eigenparameter-dependent boundary conditions~\cite{Fulton1977, Walter1973} provides powerful analytical tools for understanding the mode structure, including completeness of the eigenfunction expansion and the asymptotic distribution of eigenvalues. Furthermore, the boundary condition has a transparent continued fraction interpretation. Substituting the Cauer expansion and evaluating at $s = i\omega$, the condition becomes a CF equation whose roots are the dressed eigenfrequencies. The classical theory of continued fractions~\cite{Wall1948, Jones1980} provides key results: the roots of $sY_{\mathrm{in}}(s) + 1/L_J = 0$ interlace with the poles of $Y_{\mathrm{in}}(s)$; between any two consecutive bare resonator frequencies there exists exactly one dressed mode frequency; for $N$ bare modes there are exactly $N+1$ dressed modes; and the function $\omega \mapsto \omega \cdot \mathrm{Im}[Y_{\mathrm{in}}(i\omega)]$ is monotonic between poles, ensuring unique roots per interval.

\subsection{Illustrative Example: Quarter-Wave Resonator}
\label{sec:example_qw}

Consider a quarter-wave coplanar waveguide resonator of length $\ell$, shorted at $x = 0$ and open at $x = \ell$, coupled to a transmon at the open end. The resonator presents input admittance
\begin{equation}
Y_r(s) = \frac{1}{Z_0}\coth\left(\frac{s\ell}{v}\right),
\label{eq:resonator_admittance}
\end{equation}
which on the imaginary axis becomes $Y_r(i\omega) = -(i/Z_0)\cot(\omega\ell/v)$, purely reactive with poles at the bare resonator frequencies $\omega_n^{(0)} = (2n-1)\pi v/(2\ell)$ for $n = 1, 2, 3, \ldots$. Figure~\ref{fig:quarter_wave_circuit} shows the circuit topology.

The cotangent admittance can be expanded as an infinite continued fraction corresponding to a semi-infinite LC ladder~\cite{Pozar2011, Collin2001}. This allows the full machinery of continued fraction spectral theory to be applied to distributed resonators. The Cauer expansion yields ladder parameters $L_n \propto Z_0 \ell/(n^2 \pi^2 v)$ and $C_n \propto \ell/(Z_0 v)$, with the product $L_n C_n$ determining the resonant frequencies.

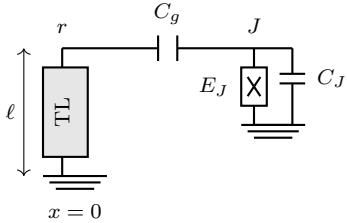
\begin{figure}[t]
\centering
\begin{tikzpicture}[scale=0.85]
    \draw[thick] (0,0) -- (0,2);
    \draw[thick, fill=gray!20] (-0.3,0.3) rectangle (0.4,1.7);
    \node[rotate=90] at (0,1) {\footnotesize TL};
    
    \draw[thick] (-0.3,0) -- (0.7,0);
    \draw[thick] (-0.2,-0.1) -- (0.6,-0.1);
    \draw[thick] (-0.1,-0.2) -- (0.5,-0.2);
    \node[below] at (0.2,-0.3) {\footnotesize $x=0$};
    
    \draw[thick] (0,2) -- (1.5,2);
    \draw[thick] (1.5,1.8) -- (1.5,2.2);
    \draw[thick] (1.8,1.8) -- (1.8,2.2);
    \draw[thick] (1.8,2) -- (3,2);
    \node[above] at (1.65,2.25) {\footnotesize $C_g$};
    
    \begin{scope}[shift={(3,2)}]
        \draw[thick] (0,0) -- (0,-0.3);
        \draw[thick, fill=white] (-0.2,-0.3) rectangle (0.2,-0.9);
        \draw[thick] (-0.1,-0.45) -- (0.1,-0.75);
        \draw[thick] (-0.1,-0.75) -- (0.1,-0.45);
        \draw[thick] (0,-0.9) -- (0,-1.2);
        \node[right] at (-1,-0.6) {\footnotesize $E_J$};
        
        \draw[thick] (0.6,0) -- (0.6,-0.4);
        \draw[thick] (0.4,-0.4) -- (0.8,-0.4);
        \draw[thick] (0.4,-0.55) -- (0.8,-0.55);
        \draw[thick] (0.6,-0.55) -- (0.6,-1.2);
        \draw[thick] (0,0) -- (0.6,0);
        \draw[thick] (0,-1.2) -- (0.6,-1.2);
        \node[right] at (0.85,-0.47) {\footnotesize $C_J$};
        
        \draw[thick] (-0.2,-1.2) -- (0.8,-1.2);
        \draw[thick] (-0.1,-1.3) -- (0.7,-1.3);
        \draw[thick] (0,-1.4) -- (0.6,-1.4);
    \end{scope}
    
    \node[above] at (0,2.1) {\footnotesize $r$};
    \node[above] at (3,2.1) {\footnotesize $J$};
    
    \draw[<->] (-0.6,0) -- (-0.6,2);
    \node[left] at (-0.6,1) {\footnotesize $\ell$};
\end{tikzpicture}
\caption{Quarter-wave transmission line resonator coupled to a transmon qubit. The resonator is shorted at $x=0$ and open at $x=\ell$. The coupling capacitor $C_g$ mediates the qubit-resonator interaction.}
\label{fig:quarter_wave_circuit}
\end{figure}

Applying the Schur complement with the junction as port $J$ yields:
\begin{equation}
Y_{\mathrm{in}}(s) = sC_J + \frac{sC_g Y_r}{Y_r + sC_g}.
\label{eq:schur_example}
\end{equation}
The first term is the junction capacitance; the second is the resonator admittance modified by the coupling capacitor. This can be written in continued fraction form, and the boundary condition $sY_{\mathrm{in}}(s) + 1/L_J = 0$ becomes the eigenvalue equation for the combined system.

\subsection{Graphical Solution and Mode Structure}
\label{sec:graphical}

For a lossless network where $Y_{\mathrm{in}}(i\omega) = iB(\omega)$, the boundary condition becomes
\begin{equation}
\omega B(\omega) = \frac{1}{L_J}.
\label{eq:bc_lossless}
\end{equation}
The graphical solution method is equivalent to finding roots of the continued fraction. The susceptance $B(\omega)$ diverges at the bare resonator frequencies, the poles of the continued fraction, creating vertical asymptotes that partition the frequency axis. Between consecutive poles, $B(\omega)$ passes through zero exactly once and changes sign, following from the positive-real property. The product $\omega B(\omega)$ varies continuously from $+\infty$ to $-\infty$ (or vice versa) between consecutive asymptotes, so by the intermediate value theorem there is exactly one intersection with the line $1/L_J$ per interval. Below the lowest bare resonator frequency, there is one additional root---the qubit-like mode---arising from the junction's inductive character.

\begin{figure}[t]
\centering
\includegraphics[width=0.95\columnwidth]{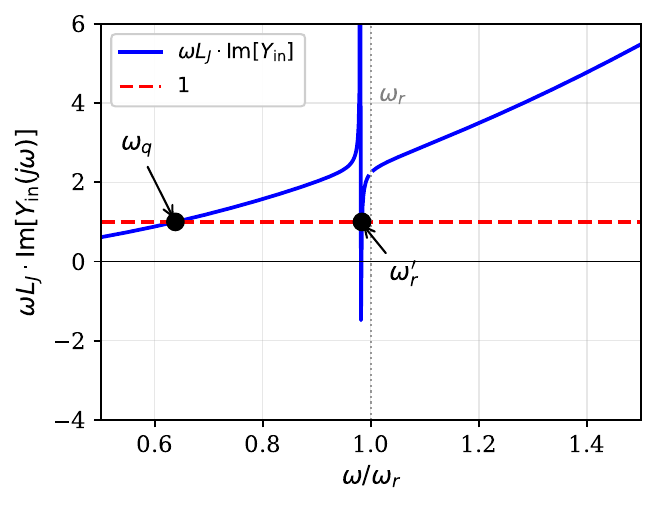}
\caption{Graphical solution of the boundary condition equation $sY_{\mathrm{in}}(s) + 1/L_J = 0$. The blue curve shows $\omega L_J \cdot \mathrm{Im}[Y_{\mathrm{in}}(i\omega)]$ versus frequency normalized to $\omega_r$. Dressed mode frequencies occur where this curve intersects unity (red dashed line). Two solutions emerge: the qubit mode $\omega_q$ at lower frequency and the dressed resonator mode $\omega_r'$ pushed below the bare frequency. The vertical asymptote at $\omega_r$ reflects the pole in the resonator admittance. The continued fraction structure guarantees exactly one root per interval between consecutive poles.}
\label{fig:bc_graphical}
\end{figure}

Figure~\ref{fig:bc_graphical} illustrates this graphical solution for a transmon coupled to a resonator. The blue curve shows the dimensionless product $\omega L_J B(\omega)$ as a function of frequency normalized to the bare resonator frequency. Mode frequencies occur where this curve intersects unity. The susceptance $B(\omega)$ diverges at the bare resonator frequencies, creating vertical asymptotes that partition the frequency axis into intervals. Between consecutive asymptotes, the product $\omega L_J B(\omega)$ varies continuously from $+\infty$ to $-\infty$ (or vice versa), guaranteeing exactly one zero crossing per interval by the intermediate value theorem. This topological argument ensures that the number of dressed modes equals the number of poles in the admittance plus one (for the junction mode).

Near an avoided crossing where the bare qubit frequency $\omega_p = 1/\sqrt{L_J C_\Sigma}$ approaches a bare resonator frequency $\omega_r$, the two solutions repel each other. The minimum separation equals the vacuum Rabi splitting $2g$. In the continued fraction framework, this gap between consecutive eigenvalues of the underlying Jacobi matrix is always positive for positive off-diagonal elements $b_n > 0$, a direct consequence of the positive-real property.


\subsection{Spatially-Dependent Admittance}
\label{sec:spatial_admittance}

For distributed networks, the admittance depends on the spatial location where the junction is connected. The driving-point admittance at position $x$ is $Y_{\mathrm{in}}(x,s) = \tilde{I}(x,s)/\tilde{V}(x,s)$, and a junction at position $x_J$ sees admittance $Y_{\mathrm{in}}(x_J,s)$. For a transmission line shorted at $x = 0$ and open at $x = \ell$, a junction at $x_J$ sees both sections in parallel:
\begin{equation}
Y_{\mathrm{in}}(x_J,s) = sC_J + \frac{1}{Z_0}\coth\left(\frac{sx_J}{v}\right) + \frac{1}{Z_0}\tanh\left(\frac{s(\ell-x_J)}{v}\right).
\label{eq:Yin_tline_explicit}
\end{equation}
The first hyperbolic cotangent represents the admittance looking toward the short at $x=0$; the hyperbolic tangent represents the admittance looking toward the open at $x=\ell$. The position-dependent admittance has a continued fraction expansion whose coefficients depend on $x_J$. The poles of each segment's CF contribute to the full spectrum, with residues that depend on $x_J$ through the mode functions. A junction at position $x_J$ couples to mode $n$ with strength proportional to $g_n(x_J) \propto |\psi_n(x_J)|$, where $\psi_n(x)$ is the normalized mode function. This provides a unified framework for analyzing position-dependent coupling: the junction position modifies the CF coefficients and hence the coupling to each mode.

\begin{theorem}[Spatially-Dependent Boundary Condition]
\label{thm:bc_spatial}
Let $Y_{\mathrm{in}}(x,s)$ be the driving-point admittance at position $x$ along a distributed passive linear network. If a Josephson junction with inductance $L_J$ is connected at position $x_J$, then in the linearized approximation the dressed mode frequencies are the positive solutions of $sY_{\mathrm{in}}(x_J,s) + 1/L_J = 0$ at $s = i\omega_n$.
\end{theorem}

\subsection{Transcendental Equations and Explicit Solutions}
\label{sec:transcendental}

We now derive the dressed mode frequencies of the linearized circuit, where the Josephson junction is replaced by its effective inductance $L_J = \varphi_0^2/E_J$. This linearized analysis yields the normal modes of the coupled system; the effects of the junction nonlinearity are treated in Section~\ref{sec:quantization}.

For a lossless network, the admittance seen by the junction takes the Foster form
\begin{equation}
Y_{\mathrm{in}}(i\omega) = i\omega C_\Sigma + \sum_{k=1}^{N} \frac{i\omega C_k (\omega_k^{(0)})^2}{(\omega_k^{(0)})^2 - \omega^2},
\label{eq:foster_general}
\end{equation}
where $C_\Sigma$ is the total shunt capacitance at the junction port (including the junction capacitance $C_J$ per Assumption~\ref{assumption:cj_convention}), and $\omega_k^{(0)}$ are the Foster pole frequencies of the linear environment, i.e., the resonant frequencies when the junction port is open-circuited.

The dressed normal-mode frequencies are determined by the boundary condition requiring that the total admittance at the junction port vanish:
\begin{equation}
Y_{\mathrm{in}}(i\omega) + \frac{1}{i\omega L_J} = 0.
\label{eq:boundary_condition_foster}
\end{equation}
Substituting Eq.~\eqref{eq:foster_general} and multiplying through by $i\omega L_J$ yields
\begin{equation}
1 - \omega^2 L_J C_\Sigma + \sum_{k=1}^{N} \frac{\omega^2 L_J C_k (\omega_k^{(0)})^2}{(\omega_k^{(0)})^2 - \omega^2} = 0.
\label{eq:bc_intermediate}
\end{equation}
Introducing dimensionless variables $u = \omega/\omega_p$ where $\omega_p = 1/\sqrt{L_J C_\Sigma}$ is the bare plasma frequency, $r_k = \omega_k^{(0)}/\omega_p$, and participation ratios $p_k = C_k/C_\Sigma$, Eq.~\eqref{eq:bc_intermediate} becomes
\begin{equation}
1 - u^2 + \sum_{k=1}^{N} \frac{u^2 p_k r_k^2}{r_k^2 - u^2} = 0.
\label{eq:bc_dimensionless}
\end{equation}
Substituting the Cauer continued-fraction realization of $Y_{\mathrm{in}}(s)$ into this boundary condition yields a scalar root-finding problem with an explicitly recursive structure. The $N+1$ roots $\{u_m\}_{m=0}^{N}$ satisfy the interlacing property
\begin{equation}
0 < u_0 < r_1 < u_1 < r_2 < \cdots < r_N < u_N,
\label{eq:interlacing}
\end{equation}
a consequence of the Cauchy interlacing theorem for positive-real functions.

For a single environmental mode ($N=1$), the boundary condition~\eqref{eq:bc_dimensionless} reduces to a quadratic equation in $\omega^2$ with exact solution
\begin{equation}
\omega_\pm^2 = \frac{\omega_p^2 + \omega_r^2(1+p)}{2} \pm \frac{1}{2}\sqrt{D},
\label{eq:exact_frequencies}
\end{equation}
where $\omega_r \equiv \omega_1^{(0)}$ is the bare resonator frequency, $p \equiv p_1 = C_1/C_\Sigma$ is the capacitive participation ratio, and the discriminant is
\begin{equation}
D = \bigl[\omega_p^2 + \omega_r^2(1+p)\bigr]^2 - 4\omega_p^2\omega_r^2.
\label{eq:discriminant}
\end{equation}
This result is exact for the linearized circuit at arbitrary coupling strength.

In the weak-coupling limit ($p \ll 1$), defining the renormalized plasma frequency $\bar{\omega}_p = \omega_p\sqrt{1-p}$, the mode-splitting parameter $\Omega = \sqrt{p\bar{\omega}_p\omega_r}$, and the detuning $\Delta = \bar{\omega}_p - \omega_r$, Eq.~\eqref{eq:exact_frequencies} reduces to
\begin{equation}
\omega_\pm \approx \frac{\bar{\omega}_p + \omega_r}{2} \pm \frac{1}{2}\sqrt{\Delta^2 + \Omega^2} + O(p^2).
\label{eq:weak_coupling_frequencies}
\end{equation}
At zero detuning ($\Delta = 0$), the mode splitting is $\omega_+ - \omega_- = \Omega$. Upon quantization and truncation of the junction to a two-level system, this normal-mode splitting maps to the vacuum Rabi splitting $2g$ (see Section~\ref{sec:quantization}). Table~\ref{tab:single_mode_formulas} summarizes the dressed frequencies in various limiting regimes.

\begin{table}[t]
\caption{Dressed normal-mode frequencies for a linearized junction coupled to a single resonator. Here $\bar{\omega}_p = \omega_p\sqrt{1-p}$, $\Omega = \sqrt{p\bar{\omega}_p\omega_r}$ is the mode-splitting parameter, $\Delta = \bar{\omega}_p - \omega_r$, and $D$ is the discriminant~\eqref{eq:discriminant}. The dispersive and resonant limits follow from weak-coupling expansion of the exact result.}
\label{tab:single_mode_formulas}
\begin{ruledtabular}
\begin{tabular}{lc}
\textrm{Regime} & \textrm{Dressed frequencies (linearized)} \\
\midrule
Exact & $\omega_\pm^2 = \dfrac{\omega_p^2 + \omega_r^2(1+p)}{2} \pm \dfrac{\sqrt{D}}{2}$ \\[10pt]
Weak coupling & $\omega_\pm \approx \dfrac{\bar{\omega}_p + \omega_r}{2} \pm \dfrac{1}{2}\sqrt{\Delta^2 + \Omega^2}$ \\[10pt]
Dispersive ($|\Delta| \gg \Omega$) & $\omega_- \approx \bar{\omega}_p + \dfrac{\Omega^2}{4\Delta}$, \quad $\omega_+ \approx \omega_r - \dfrac{\Omega^2}{4\Delta}$ \\[8pt]
Resonant ($\Delta = 0$) & $\omega_\pm = \dfrac{\bar{\omega}_p + \omega_r}{2} \pm \dfrac{\Omega}{2}$ \\[6pt]
Strong coupling & Solve Eq.~\eqref{eq:bc_dimensionless} numerically \\
\end{tabular}
\end{ruledtabular}
\end{table}

For a junction coupled to $N$ environmental modes, the dressed frequency near the bare plasma frequency can be obtained by expanding the boundary condition~\eqref{eq:bc_dimensionless} around $u \approx 1$. In the dispersive regime where $|u^2 - r_k^2| \gg p_k r_k^2$ for all modes, the root near $\omega_p$ satisfies
\begin{equation}
\omega_-^2 \approx \omega_p^2(1-p_\mathrm{tot}) + \sum_{k=1}^{N} \frac{p_k \omega_p^2 \omega_k^2}{\omega_p^2 - \omega_k^2},
\label{eq:omega_minus_multimode}
\end{equation}
where $p_\mathrm{tot} = \sum_k p_k$. To leading order in $p_k$, this gives the frequency shift
\begin{equation}
\delta\omega_- \equiv \omega_- - \omega_p \approx -\frac{\omega_p p_\mathrm{tot}}{2} + \sum_{k=1}^{N} \frac{p_k \omega_p \omega_k^2}{2(\omega_p^2 - \omega_k^2)}.
\label{eq:frequency_shift_multimode}
\end{equation}
The first term is a capacitive renormalization from the total participation, while the sum represents hybridization shifts from each mode, repulsive for $\omega_k > \omega_p$ and attractive for $\omega_k < \omega_p$.

Upon quantization and truncation to a two-level qubit (Section~\ref{sec:quantization}), the hybridization sum becomes the multimode Lamb shift with the identification $g_k^2/\Delta_k \leftrightarrow p_k\omega_p\omega_k^2/[2(\omega_p^2 - \omega_k^2)]$ in the appropriate limit. For an infinite number of modes, convergence of the frequency shift~\eqref{eq:frequency_shift_multimode} is guaranteed by the finite shunt capacitance $C_\Sigma$ at the junction port. Physically, $C_\Sigma$ acts as a low-pass filter that suppresses the junction's coupling to high-frequency modes, a point developed further in Section~\ref{sec:uv}.


\section{Network Synthesis}
\label{sec:synthesis}

Given the driving-point admittance $Y_{\mathrm{in}}(s)$, the next step is to construct an equivalent lumped-element circuit whose input admittance equals this function. This circuit represents the linear electromagnetic environment seen by the junction---the Josephson nonlinearity is not part of the synthesis and will be reintroduced only during quantization (Section~\ref{sec:quantization}). This inverse problem, known as network synthesis, provides the physical degrees of freedom that will be quantized. The key result, due to Brune~\cite{Brune1931}, is that any positive-real function admits such a realization. As established in Section~\ref{sec:cf_framework}, the Cauer continued fraction expansion provides a canonical ladder realization that yields a tridiagonal Hamiltonian upon quantization. The Foster and Cauer canonical forms represent dual approaches to network synthesis. The Foster form uses partial fraction expansion of $Y(s)$ to yield parallel LC branches, with each branch resonating independently at frequency $\omega_k$ and coupling to the junction through the branch capacitance $C_k$. This form directly reveals the mode structure. The Cauer form uses continued fraction expansion to yield a ladder network with series inductors and shunt capacitors, where adjacent elements couple through shared nodes, producing a tridiagonal Hamiltonian. Both forms are mathematically equivalent representations of the same positive-real function; the transformation between them is the correspondence between partial fraction and continued fraction expansions~\cite{Guillemin1957}.

\subsection{Foster Synthesis for Lossless Networks}
\label{sec:foster_synthesis}

For lossless (purely reactive) networks, the admittance is purely imaginary on the real frequency axis: $Y_{\mathrm{in}}(i\omega) = iB(\omega)$, where $B(\omega)$ is the susceptance. The positive-real conditions require that the poles of $Y_{\mathrm{in}}(s)$ lie on the imaginary axis at complex-conjugate pairs $s = \pm i\omega_k$, each pole be simple, and the residues at poles on the positive imaginary axis be real and positive~\cite{Guillemin1957}.

The Foster canonical form~\cite{Foster1924} expresses any lossless positive-real driving-point admittance as a partial fraction expansion:
\begin{equation}
Y(s) = sC_\infty + \frac{1}{sL_0} + \sum_{k=1}^{N} \frac{sC_k \omega_k^2}{s^2 + \omega_k^2}.
\label{eq:foster_form}
\end{equation}
The term $sC_\infty$ represents a shunt capacitance present when the admittance grows without bound as $s \to \infty$. The term $1/(sL_0)$ represents an inductor present when there is a pole at dc. Each term in the summation represents a series $LC$ branch resonating at $\omega_k = 1/\sqrt{L_k C_k}$.

The Foster partial fraction expansion is the spectral decomposition of the admittance function. Comparing with the resolvent formula $G(z) = \sum_k |\psi_k(0)|^2/(z - E_k)$ from Eq.~\eqref{eq:resolvent_spectral}, we identify the poles $\omega_k^2$ with eigenvalues of the underlying Jacobi operator, and the residues $C_k \omega_k^2$ with squared eigenvector components at the junction site. The correspondence is: pole $\omega_k$ maps to bare mode frequency, and residue $C_k \omega_k^2$ maps to junction participation $|\phi_k^J|^2$.

The circuit parameters are extracted from the residues of $Y(s)$. At a pole $s = i\omega_k$, the residue is computed as
\begin{equation}
\mathrm{Res}[Y(s), i\omega_k] = \lim_{s \to i\omega_k} (s - i\omega_k) \frac{sC_k \omega_k^2}{s^2 + \omega_k^2} = \frac{C_k \omega_k}{2}.
\label{eq:residue_result}
\end{equation}
Inverting this relation gives the branch capacitance and inductance:
\begin{equation}
C_k = \frac{2|\mathrm{Res}[Y(s), i\omega_k]|}{\omega_k}, \quad L_k = \frac{1}{\omega_k^2 C_k}.
\label{eq:foster_params}
\end{equation}

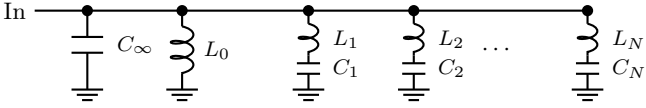
\begin{figure}[t]
\centering
\begin{tikzpicture}[scale=0.7]
    \draw[thick] (0,0) -- (1,0);
    \node[circle, fill=black, inner sep=1.5pt] at (1,0) {};
    \node[left, font=\small] at (0,0) {In};
    
    \draw[thick] (1,0) -- (1,-0.5);
    \draw[thick] (0.7,-0.5) -- (1.3,-0.5);
    \draw[thick] (0.7,-0.8) -- (1.3,-0.8);
    \draw[thick] (1,-0.8) -- (1,-1.5);
    \draw[thick] (0.7,-1.5) -- (1.3,-1.5);
    \draw[thick] (0.8,-1.6) -- (1.2,-1.6);
    \draw[thick] (0.9,-1.7) -- (1.1,-1.7);
    \node[right, font=\footnotesize] at (1.4,-0.65) {$C_\infty$};
    
    \draw[thick] (1,0) -- (10.5,0);
    
    \begin{scope}[shift={(2.8,0)}]
        \node[circle, fill=black, inner sep=1.5pt] at (0,0) {};
        \draw[thick] (0,0) -- (0,-0.3);
        \draw[thick, decoration={coil, aspect=0.4, segment length=2mm, amplitude=1.5mm}, decorate] (0,-0.3) -- (0,-1.2);
        \draw[thick] (0,-1.2) -- (0,-1.5);
        \draw[thick] (-0.3,-1.5) -- (0.3,-1.5);
        \draw[thick] (-0.2,-1.6) -- (0.2,-1.6);
        \draw[thick] (-0.1,-1.7) -- (0.1,-1.7);
        \node[right, font=\footnotesize] at (0.25,-0.75) {$L_0$};
    \end{scope}
    
    \begin{scope}[shift={(5.2,0)}]
        \node[circle, fill=black, inner sep=1.5pt] at (0,0) {};
        \draw[thick] (0,0) -- (0,-0.25);
        \draw[thick, decoration={coil, aspect=0.4, segment length=1.8mm, amplitude=1.3mm}, decorate] (0,-0.25) -- (0,-0.85);
        \draw[thick] (0,-0.85) -- (0,-1.0);
        \draw[thick] (-0.22,-1.0) -- (0.22,-1.0);
        \draw[thick] (-0.22,-1.2) -- (0.22,-1.2);
        \draw[thick] (0,-1.2) -- (0,-1.5);
        \draw[thick] (-0.3,-1.5) -- (0.3,-1.5);
        \draw[thick] (-0.2,-1.6) -- (0.2,-1.6);
        \draw[thick] (-0.1,-1.7) -- (0.1,-1.7);
        \node[right, font=\footnotesize] at (0.3,-0.55) {$L_1$};
        \node[right, font=\footnotesize] at (0.3,-1.1) {$C_1$};
    \end{scope}
    
    \begin{scope}[shift={(7.2,0)}]
        \node[circle, fill=black, inner sep=1.5pt] at (0,0) {};
        \draw[thick] (0,0) -- (0,-0.25);
        \draw[thick, decoration={coil, aspect=0.4, segment length=1.8mm, amplitude=1.3mm}, decorate] (0,-0.25) -- (0,-0.85);
        \draw[thick] (0,-0.85) -- (0,-1.0);
        \draw[thick] (-0.22,-1.0) -- (0.22,-1.0);
        \draw[thick] (-0.22,-1.2) -- (0.22,-1.2);
        \draw[thick] (0,-1.2) -- (0,-1.5);
        \draw[thick] (-0.3,-1.5) -- (0.3,-1.5);
        \draw[thick] (-0.2,-1.6) -- (0.2,-1.6);
        \draw[thick] (-0.1,-1.7) -- (0.1,-1.7);
        \node[right, font=\footnotesize] at (0.3,-0.55) {$L_2$};
        \node[right, font=\footnotesize] at (0.3,-1.1) {$C_2$};
    \end{scope}
    
    \node at (8.8,-0.75) {$\cdots$};
    
    \begin{scope}[shift={(10.5,0)}]
        \node[circle, fill=black, inner sep=1.5pt] at (0,0) {};
        \draw[thick] (0,0) -- (0,-0.25);
        \draw[thick, decoration={coil, aspect=0.4, segment length=1.8mm, amplitude=1.3mm}, decorate] (0,-0.25) -- (0,-0.85);
        \draw[thick] (0,-0.85) -- (0,-1.0);
        \draw[thick] (-0.22,-1.0) -- (0.22,-1.0);
        \draw[thick] (-0.22,-1.2) -- (0.22,-1.2);
        \draw[thick] (0,-1.2) -- (0,-1.5);
        \draw[thick] (-0.3,-1.5) -- (0.3,-1.5);
        \draw[thick] (-0.2,-1.6) -- (0.2,-1.6);
        \draw[thick] (-0.1,-1.7) -- (0.1,-1.7);
        \node[right, font=\footnotesize] at (0.3,-0.55) {$L_N$};
        \node[right, font=\footnotesize] at (0.3,-1.1) {$C_N$};
    \end{scope}
\end{tikzpicture}
\caption{Foster canonical form for a lossless positive-real admittance. The circuit consists of parallel branches connected to a common input node: a direct capacitor $C_\infty$ to ground capturing the high-frequency asymptotic behavior, an inductor $L_0$ to ground present when there is a pole at dc, and series $LC$ tanks tuned to the resonant frequencies $\omega_k = 1/\sqrt{L_k C_k}$. Each tank represents an environmental mode that can exchange energy with the junction. This parallel structure is dual to the Cauer ladder form.}
\label{fig:foster_synthesis}
\end{figure}

The physical interpretation of the Foster form, illustrated in Fig.~\ref{fig:foster_synthesis}, is that the admittance decomposes into parallel branches, each resonating at a distinct frequency. At frequency $\omega_k$, the $k$-th series $LC$ branch presents zero impedance, acting as a short circuit that dominates the response. Each branch represents an environmental mode that can exchange energy with the junction.

\subsection{Cauer Synthesis and the Ladder Network}
\label{sec:cauer_synthesis_detailed}

The Cauer synthesis provides an alternative canonical realization as a ladder network. Given a lossless positive-real admittance $Y(s)$, the Cauer expansion proceeds by alternating pole extractions. For Type~I (beginning with a series inductor), the algorithm is:

\begin{enumerate}
\item Form the impedance $Z(s) = 1/Y(s)$. If $\lim_{s \to \infty} Z(s)/s = L_1 > 0$, extract a series inductor: $Z_1(s) = Z(s) - sL_1$.
\item Form $Y_1(s) = 1/Z_1(s)$. If $\lim_{s \to \infty} Y_1(s)/s = C_1 > 0$, extract a shunt capacitor: $Y_2(s) = Y_1(s) - sC_1$.
\item Iterate: form $Z_2(s) = 1/Y_2(s)$, extract $L_2$, etc.
\item The process terminates when the remainder is constant.
\end{enumerate}

For Type~II (beginning with a shunt capacitor), interchange the roles of $Y$ and $Z$ in the first step: extract $C_0 = \lim_{s \to \infty} Y(s)/s$ first, then proceed with alternating impedance/admittance extractions.

The resulting circuit is a ladder of alternating series inductors and shunt capacitors. For Type~I, the expansion is:
\begin{equation}
Z(s) = sL_1 + \cfrac{1}{sC_1 + \cfrac{1}{sL_2 + \cfrac{1}{sC_2 + \ddots}}}
\end{equation}
which is equivalent to the admittance expansion in Eq.~\eqref{eq:cauer_type1}. The resonant frequencies of the ladder sections are $\omega_n = 1/\sqrt{L_n C_n}$, which encode the same spectral information as the Foster poles through a different parameterization.

The transformation between Foster and Cauer parameters is a classical result in network theory~\cite{Guillemin1957}. Given Foster parameters $\{C_k, \omega_k\}$, the Cauer parameters are obtained by evaluating the continued fraction expansion of the partial fraction sum. The explicit mapping between Cauer ladder parameters $\{L_n, C_n\}$ and Jacobi matrix elements $\{a_n, b_n\}$ is given in Appendix~\ref{app:lc_jacobi}. Both transformations are $O(N^2)$ operations for $N$ modes.

\subsection{Positive-Real Fitting and Passivity Enforcement}
\label{sec:passivity}

In practice, $Y_{\mathrm{in}}(s)$ is obtained from electromagnetic simulation or measurement at discrete frequencies. Before synthesis, the admittance must be fitted to a rational function. Vector fitting~\cite{Gustavsen1999} produces a fit of the form
\begin{equation}
Y_{\mathrm{fit}}(s) = d + se + \sum_{k=1}^{N} \frac{r_k}{s - p_k},
\label{eq:vector_fit}
\end{equation}
where the poles $p_k$ lie in the left half-plane for stable systems.

A critical subtlety is that vector fitting does not automatically enforce positive-realness. The fitting algorithm minimizes least-squares error but may produce $\mathrm{Re}[Y_{\mathrm{fit}}(i\omega)] < 0$ for some frequencies, corresponding to negative conductance. Such a function cannot represent a passive network, and quantizing it yields a Hamiltonian unbounded from below. The positive-real property is essential for both network realizability and continued fraction convergence. A positive-real function has a well-defined Cauer expansion with positive coefficients $\{L_n > 0, C_n > 0\}$, which guarantees that the synthesized ladder has physical (positive-valued) elements, the Jacobi matrix has positive off-diagonal elements, the continued fraction converges to a well-defined spectral measure, and the quantized Hamiltonian is bounded from below. Passivity violations manifest as negative Cauer coefficients.

When vector fitting violates passivity, enforcement must be applied before synthesis. Standard methods include residue perturbation~\cite{Gustavsen2001}, Hamiltonian perturbation, or convex optimization~\cite{GrivetTalocia2016}. The passivity enforcement algorithm of Gustavsen and Semlyen~\cite{Gustavsen2001} iteratively perturbs residues to ensure $\mathrm{Re}[Y(i\omega)] \geq 0$ at all frequencies while minimizing the perturbation to the fit. The number of poles $N$ determines the number of modes; ultraviolet convergence (Section~\ref{sec:uv}) ensures that physical observables converge as $N$ increases. The truncation of the Foster expansion to $N$ modes corresponds to truncating the Cauer continued fraction at level $N$. The truncated CF is the resolvent of a finite $N \times N$ Jacobi matrix whose eigenvalues approximate the first $N$ eigenvalues of the infinite operator. The truncation error is controlled by the decay of CF coefficients, governed by the high-frequency behavior of $Y_{\mathrm{in}}(s)$.

\subsection{Brune Synthesis for Lossy Networks}
\label{sec:brune_synthesis}

When the network contains dissipative elements, $Y_{\mathrm{in}}(s)$ is no longer purely imaginary on the frequency axis. The positive-real property generalizes to require $\mathrm{Re}[Y_{\mathrm{in}}(s)] \geq 0$ in the right half-plane. On the imaginary axis, $Y_{\mathrm{in}}(i\omega) = G(\omega) + iB(\omega)$, where the conductance $G(\omega) \geq 0$ represents dissipation.

The Brune synthesis algorithm~\cite{Brune1931} realizes any positive-real function as a network containing resistors, inductors, capacitors, and ideal transformers. The algorithm is iterative, reducing the admittance complexity at each stage by extracting a Brune section. Modern implementations for circuit QED applications are discussed in Refs.~\cite{Solgun2014, Solgun2015}.

The Brune synthesis extends the Cauer continued fraction to lossy networks. Each Brune section extracts one transmission zero (a frequency where the network is purely reactive), analogous to extracting one level of the lossless CF. The conductance $G(\omega)$ determines the dissipative coupling to the extracted mode. In the quantum theory, this dissipation appears as coupling to a bath, with Brune section parameters determining the bath spectral density (see Section~\ref{sec:complex_poles}). The chain structure of Brune sections is the lossy generalization of the Cauer ladder. 

The algorithm proceeds as follows. The first step identifies the frequency $\omega_0$ where $G(\omega) = \mathrm{Re}[Y_{\mathrm{in}}(i\omega)]$ attains its minimum value $G_0 \geq 0$. For $G_0 > 0$, the second step extracts a shunt conductance: $Y_1(s) = Y_{\mathrm{in}}(s) - G_0$. The third step extracts the susceptance at the critical frequency, forming $Y_2(s) = Y_1(s) - sB_0/\omega_0$. The fourth step analyzes the impedance $Z_2(s) = 1/Y_2(s)$, which has a pole at $s = i\omega_0$.

\begin{figure}[t]
\centering
\begin{tikzpicture}[scale=0.9]
    \draw[thick] (0,0) -- (1,0);
    \node[circle, fill=black, inner sep=1.5pt] at (1,0) {};
    \node[left] at (0,0) {In};
    
    \draw[thick] (1,0) -- (1,-0.5);
    \draw[thick, decoration={zigzag, segment length=3mm, amplitude=1mm}, decorate] (1,-0.5) -- (1,-1.5);
    \draw[thick] (1,-1.5) -- (1,-2);
    \node[right] at (1.3,-1) {$G_0$};
    
    \draw[thick] (0.7,-2) -- (1.3,-2);
    \draw[thick] (0.8,-2.1) -- (1.2,-2.1);
    \draw[thick] (0.9,-2.2) -- (1.1,-2.2);
    
    \draw[thick] (1,0) -- (2.5,0);
    \draw[thick] (2.5,-0.3) -- (2.5,0.3);
    \draw[thick] (2.8,-0.3) -- (2.8,0.3);
    \draw[thick] (2.8,0) -- (4,0);
    \node[above] at (2.65,0.4) {$C$};
    
    \node[circle, fill=black, inner sep=1.5pt] at (4,0) {};
    \draw[thick] (4,0) -- (4,-0.3);
    \draw[thick, decoration={coil, aspect=0.4, segment length=2mm, amplitude=1.5mm}, decorate] (4,-0.3) -- (4,-1.2);
    \draw[thick] (4,-1.2) -- (4,-2);
    \node[left] at (3.7,-0.75) {$L_a$};
    
    \fill (3.85,-0.4) circle (1.5pt);
    \fill (4.65,-0.4) circle (1.5pt);
    
    \draw[thick] (4.5,0) -- (4.5,-0.3);
    \draw[thick, decoration={coil, aspect=0.4, segment length=2mm, amplitude=1.5mm}, decorate] (4.5,-0.3) -- (4.5,-1.2);
    \draw[thick] (4.5,-1.2) -- (4.5,-2);
    \node[right] at (4.8,-0.75) {$L_b$};
    
    \draw[<->, thick, dashed] (4.1,-0.75) -- (4.4,-0.75);
    \node[below] at (4.25,-0.85) {\small $M$};
    
    \draw[thick] (4.5,0) -- (6,0);
    \node[right] at (6,0) {Out};
    
    \draw[thick] (4,-2) -- (4.5,-2);
    \draw[thick] (3.7,-2) -- (4.8,-2);
    \draw[thick] (3.8,-2.1) -- (4.7,-2.1);
    \draw[thick] (3.9,-2.2) -- (4.6,-2.2);
\end{tikzpicture}
\caption{Brune section for positive-real synthesis. The shunt conductance $G_0$ represents the minimum loss at the critical frequency. The series capacitor $C$ and mutually coupled inductors $L_a$, $L_b$ with mutual inductance $M$ form the reactive portion that extracts the transmission zero at $\omega_0$. The dots indicate the polarity of the mutual coupling. This structure is one rung of the lossy continued fraction ladder, with $G_0$ encoding the dissipative coupling to the bath at frequency $\omega_0$.}
\label{fig:brune_section}
\end{figure}
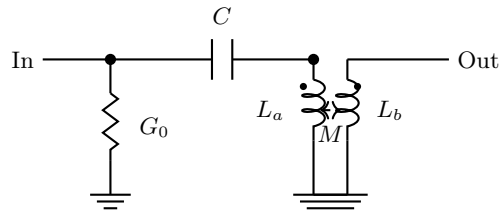

The Brune section, shown in Fig.~\ref{fig:brune_section}, extracts this pole while maintaining positive-realness. It consists of a series capacitor $C$ and two mutually coupled inductors $L_a$ and $L_b$ with mutual inductance $M$. The element values are determined by matching the pole residue and satisfying realizability constraints. The complete synthesis produces a cascade of Brune sections, each contributing a shunt conductance $G_k$ (dissipative channel) and reactive elements (mode at frequency $\omega_k$). Upon quantization, each section yields one Hamiltonian mode coupled to one dissipative bath channel.

When ideal transformers appear in the synthesis, they impose algebraic constraints rather than differential equations, reducing the number of independent degrees of freedom. During quantization, these constraints are handled through symplectic reduction or the ``effective Kirchhoff technique'' developed in Ref.~\cite{Solgun2015}, which eliminates transformer branches by a coordinate transformation.

\subsection{Complex Poles and Mode Linewidths}
\label{sec:complex_poles}

When dissipation is introduced, poles migrate from the imaginary axis into the left half-plane, acquiring finite lifetimes. The boundary condition equation $F(s) \equiv sY_{\mathrm{in}}(s) + 1/L_J = 0$ has solutions
\begin{equation}
s_n = -\frac{\kappa_n}{2} + i\omega_n,
\label{eq:complex_pole_form}
\end{equation}
where $\kappa_n > 0$ is the energy decay rate and $\omega_n$ is the oscillation frequency.

In the continued fraction framework, the migration of poles corresponds to the Jacobi matrix acquiring a non-Hermitian perturbation. The resolvent remains a continued fraction with complex coefficients:
\begin{equation}
G(z) = \cfrac{1}{z - (a_0 + i\gamma_0) - \cfrac{|b_1|^2}{z - (a_1 + i\gamma_1) - \ddots}}
\end{equation}
where $\gamma_n > 0$ encodes dissipation at level $n$. The poles at $z_n = \omega_n^2 - i\omega_n\kappa_n$ give the complex eigenfrequencies.

To derive the decay rate perturbatively, expand around the lossless pole. Writing $Y_{\mathrm{in}}(s) = iB(s) + G(s)$ where $G(s)$ is the (small) conductance, the perturbed pole satisfies
\begin{equation}
\delta s = -\frac{s_n^{(0)} G(\omega_n^{(0)})}{F_0'(s_n^{(0)})},
\label{eq:delta_s}
\end{equation}
where $s_n^{(0)} = i\omega_n^{(0)}$ is the unperturbed pole and $F_0(s) = sY_{\mathrm{in}}^{(0)}(s) + 1/L_J$ is the lossless boundary condition function. Define the effective modal capacitance:
\begin{equation}
C_n^{\mathrm{eff}} \equiv \frac{1}{2\omega_n^{(0)}} \frac{d(\omega B)}{d\omega}\bigg|_{\omega_n^{(0)}}.
\label{eq:Ceff_def}
\end{equation}
Evaluating the derivative of $F_0$ and taking the imaginary part of $\delta s$ yields the decay rate in the simple form
\begin{equation}
\kappa_n = \frac{G_n}{C_n^{\mathrm{eff}}},
\label{eq:kappa_final}
\end{equation}
where $G_n = G(\omega_n^{(0)})$ is the conductance at the mode frequency.

The effective modal capacitance is directly related to the continued fraction structure. In the spectral decomposition~\eqref{eq:resolvent_spectral}, the residue at pole $\omega_n$ is $|\psi_n(0)|^2 = 1/C_n^{\mathrm{eff}}$, the squared amplitude of the $n$-th eigenvector at the junction site. The decay rate formula is therefore
\begin{equation}
\kappa_n = G_n \cdot |\psi_n(0)|^2 = G_n \cdot p_n,
\label{eq:kappa_participation}
\end{equation}
where $p_n$ is the junction participation ratio in mode $n$. This is Fermi's golden rule: the decay rate equals the bath coupling (conductance) times the system-bath overlap (participation). This formula is equivalent to the $T_1$ expression derived in Ref.~\cite{Solgun2014} using the Caldeira-Leggett formalism. The formula reveals conditions for long qubit lifetime: minimize $G(\omega_q)$ at the qubit frequency and maximize $C_q^{\mathrm{eff}}$. Both conditions are satisfied by Purcell filters, which create a ``notch'' in the conductance at the qubit frequency while maintaining coupling at the readout frequency.


\section{Canonical Quantization}
\label{sec:quantization}

With the electromagnetic environment synthesized as a lumped circuit, we now construct the Lagrangian and Hamiltonian and carry out canonical quantization. The result is a quantum Hamiltonian retaining the full cosine nonlinearity of the Josephson potential, exact within the lumped-element circuit model obtained from finite Foster or Brune synthesis of $\Yin(s)$. The continued fraction structure established in Section~\ref{sec:cf_framework} provides both the natural basis for the eigenvalue problem and the computational framework for exact diagonalization across all coupling regimes. The quantization procedure preserves the tridiagonal structure at every stage: the synthesized Cauer ladder has nearest-neighbor coupling yielding a tridiagonal classical Hamiltonian, and upon quantization this structure becomes a block-tridiagonal quantum Hamiltonian enabling matrix continued fraction solution. In this section, we start by constructing the classical Lagrangian and Hamiltonian, transforms to normal modes, and connects participation ratios to admittance. We addresses basis selection for compact versus non-compact variables, and establish the exact quantum Hamiltonian with its block-tridiagonal structure. Further, we develop convergence theory, connects to the Mathieu equation, presents variational bounds, and closes the loop back to the boundary admittance.

\subsection{The Circuit Lagrangian}
\label{sec:lagrangian}

Consider a synthesized circuit with $M$ independent node fluxes $\Phi_1, \ldots, \Phi_M$ after eliminating any constraints from ideal transformers through symplectic reduction~\cite{Osborne2024}. The synthesized circuit is purely linear, arising from Foster or Brune synthesis of the positive-real admittance $\Yin(s)$. The Josephson junction contributes only the nonlinear cosine potential; its linear approximation (the Josephson inductance $L_J$) appeared in the boundary condition equation but is not part of the synthesized network. The capacitance matrix $\mathrm{C} \in \mathbb{R}^{M \times M}$ and inverse inductance matrix $\mathrm{L}^{-1} \in \mathbb{R}^{M \times M}$ are both real symmetric and positive-definite.

When the circuit is synthesized in Cauer (ladder) form, both $\mathrm{C}$ and $\mathrm{L}^{-1}$ are tridiagonal matrices. The Cauer ladder topology ensures that each node couples only to its nearest neighbors:
\begin{equation}
\mathrm{L}^{-1} = \begin{pmatrix}
L_1^{-1} & -L_1^{-1} & 0 & \cdots \\
-L_1^{-1} & L_1^{-1} + L_2^{-1} & -L_2^{-1} & \cdots \\
0 & -L_2^{-1} & L_2^{-1} + L_3^{-1} & \cdots \\
\vdots & \vdots & \vdots & \ddots
\end{pmatrix}.
\label{eq:L_inv_tridiagonal}
\end{equation}
This tridiagonal structure is preserved through normal mode transformation and quantization, yielding the Jacobi matrix form central to the continued fraction spectral theory.

The kinetic energy stored in the electric fields of all capacitors is
\begin{equation}
T = \frac{1}{2} \dot{{\Phi}}^T \mathrm{C} \dot{{\Phi}} = \frac{1}{2} \sum_{i,j=1}^{M} C_{ij} \dot{\Phi}_i \dot{\Phi}_j,
\label{eq:kinetic_energy}
\end{equation}
where $\dot{{\Phi}} = (\dot{\Phi}_1, \ldots, \dot{\Phi}_M)^T$ is the vector of node voltages. The potential energy stored in the magnetic fields of all linear inductors is
\begin{equation}
U_{\mathrm{lin}} = \frac{1}{2} {\Phi}^T \mathrm{L}^{-1} {\Phi} = \frac{1}{2} \sum_{i,j=1}^{M} (L^{-1})_{ij} \Phi_i \Phi_j.
\label{eq:linear_potential}
\end{equation}
The Josephson junction at node $J$ contributes the nonlinear potential
\begin{equation}
U_J = -E_J \cos\left(\frac{\Phi_J}{\varphi_0}\right),
\label{eq:junction_potential}
\end{equation}
where $\Phi_J$ is the flux at the junction node and $\varphi_0 = \hbar/(2e) = \Phi_0/(2\pi)$ is the reduced flux quantum, with $\Phi_0 = h/(2e)$ the magnetic flux quantum. The total Lagrangian is
\begin{equation}
\mathcal{L} = \frac{1}{2}\dot{{\Phi}}^T \mathrm{C} \dot{{\Phi}} - \frac{1}{2}{\Phi}^T \mathrm{L}^{-1} {\Phi} + E_J \cos\left(\frac{\Phi_J}{\varphi_0}\right).
\label{eq:lagrangian}
\end{equation}
This Lagrangian is exact in the sense that no approximation has been made to the Josephson nonlinearity.

\subsection{Conjugate Momenta and the Hamiltonian}
\label{sec:hamiltonian}

The conjugate momentum to the node flux $\Phi_i$ is
\begin{equation}
Q_i = \frac{\partial \mathcal{L}}{\partial \dot{\Phi}_i} = \sum_{j=1}^{M} C_{ij} \dot{\Phi}_j,
\label{eq:conjugate_momentum}
\end{equation}
or in matrix form ${Q} = \mathrm{C} \dot{{\Phi}}$. The physical interpretation is that $Q_i$ represents the total charge at node $i$. Since $\mathrm{C}$ is positive-definite and hence invertible, the Hamiltonian is obtained via the Legendre transformation:
\begin{multline}
H = \frac{1}{2} {Q}^T \mathrm{C}^{-1} {Q} + \frac{1}{2} {\Phi}^T \mathrm{L}^{-1} {\Phi} \\
- E_J \cos\left(\frac{\Phi_J}{\varphi_0}\right).
\label{eq:hamiltonian}
\end{multline}
Under Assumption~\ref{assumption:cj_convention}, the junction capacitance $C_J$ is already incorporated in the capacitance matrix $\mathrm{C}$. No separate charging energy term $Q_J^2/(2C_J)$ should be added; doing so would double-count the junction capacitance. The quadratic part of the Hamiltonian defines a generalized eigenvalue problem whose matrix pencil $(\mathrm{L}^{-1}, \mathrm{C})$ has the same spectral structure as a Jacobi matrix. The eigenfrequencies $\omega_n$ are the square roots of the eigenvalues of $\mathrm{C}^{-1/2} \mathrm{L}^{-1} \mathrm{C}^{-1/2}$. When both $\mathrm{C}$ and $\mathrm{L}^{-1}$ are tridiagonal (Cauer synthesis), this matrix is also tridiagonal, a Jacobi matrix, and the continued fraction formula~\eqref{eq:resolvent_cf} gives its resolvent exactly.

\subsection{Normal Mode Decomposition}
\label{sec:normal_modes}

The linear part of the Hamiltonian describes coupled harmonic oscillators. We diagonalize the quadratic form by transforming to normal mode coordinates through the generalized eigenvalue problem
\begin{equation}
\mathrm{L}^{-1} {\phi}_n = \omega_n^2 \mathrm{C} {\phi}_n,
\label{eq:generalized_eigenvalue}
\end{equation}
choosing eigenvectors $\mathrm{C}$-orthonormal so that the transformation ${\Phi} = \Psi{\xi}$ yields $\Psi^T \mathrm{C} \Psi = \mathrm{I}$ and $\Psi^T \mathrm{L}^{-1} \Psi = \Omega^2$, where $\Omega^2 = \mathrm{diag}(\omega_1^2, \ldots, \omega_M^2)$.

\begin{proposition}[Properties of the generalized eigenvalue problem]
\label{prop:gep_properties}
Let $\mathrm{C}$ and $\mathrm{L}^{-1}$ be real symmetric positive-definite $M \times M$ matrices. Then all eigenvalues $\omega_n^2$ are real and strictly positive, the eigenvectors ${\phi}_n$ can be chosen real, and they satisfy the $\mathrm{C}$-orthonormality condition ${\phi}_m^T \mathrm{C} {\phi}_n = \delta_{mn}$ as well as the $\mathrm{L}^{-1}$-diagonalization ${\phi}_m^T \mathrm{L}^{-1} {\phi}_n = \omega_n^2 \delta_{mn}$.
\end{proposition}

The columns of the modal matrix $\Psi$ are the eigenvectors of the generalized eigenvalue problem, which are also the eigenvectors of the corresponding Jacobi matrix. The junction participation $\phi_n^J = (\Psi)_{Jn}$ is the first component of the $n$-th eigenvector, which equals the residue of the resolvent at eigenvalue $\omega_n^2$.

The coordinate transformation ${\Phi}(t) = \Psi {\xi}(t)$ diagonalizes the linear Hamiltonian. The junction flux is expressed in terms of normal modes as
\begin{equation}
\Phi_J = \sum_{n=1}^{M} \xi_n \phi_n^J,
\label{eq:Phi_J_modes}
\end{equation}
where $\phi_n^J \equiv ({\phi}_n)_J$ is the junction participation of mode $n$. A mode with $\phi_n^J = 0$ does not participate in the junction dynamics, while a mode with large $|\phi_n^J|$ couples strongly.

The participation ratio $(\phi_n^J)^2$ has a direct continued fraction interpretation. From the spectral decomposition of the resolvent:
\begin{equation}
G(z) = \sum_n \frac{(\phi_n^J)^2}{z - \omega_n^2},
\end{equation}
the participation ratio is the residue at eigenvalue $\omega_n^2$. This connects to the boundary admittance through~\cite{Nigg2012,Minev2021}
\begin{equation}
(\phi_n^J)^2 = \frac{1}{C_n^{\mathrm{eff}}} = \frac{2\omega_n}{\frac{d}{d\omega}[\omega B(\omega)]|_{\omega_n}},
\label{eq:participation_from_Y}
\end{equation}
where $B(\omega) = \mathrm{Im}[\Yin(i\omega)]$ is the susceptance. The participation ratio is thus computable directly from the admittance derivative at the mode frequency, without explicit eigenvector calculation.

\subsection{Choice of Basis for Nonlinear Elements}
\label{sec:basis_choice}

\begin{remark}[Basis selection for compact vs.\ non-compact variables]
\label{remark:basis_selection}
The configuration space topology dictates the appropriate basis for each degree of freedom:

\textit{Transmon junction (compact phase).} The superconducting phase $\varphi = \Phi_J/\varphi_0$ is a compact variable living on the circle $S^1$, with $\varphi \sim \varphi + 2\pi$ representing the same physical state. This periodicity reflects the discrete spectrum of Cooper pair number: the conjugate variable $\hat{n} = -i\partial/\partial\varphi$ has integer eigenvalues $n \in \mathbb{Z}$. The charge basis $\{|n\rangle\}_{n \in \mathbb{Z}}$ is natural, and the cosine potential is exactly tridiagonal:
\begin{equation}
\langle m | \cos\hat{\varphi} | n \rangle = \frac{1}{2}(\delta_{m,n+1} + \delta_{m,n-1}).
\label{eq:cosine_charge_tridiag}
\end{equation}
The transmon Hamiltonian $\hat{H}_{\mathrm{tr}} = 4E_C \hat{n}^2 - E_J \cos\hat{\varphi}$ is therefore a Jacobi matrix in charge basis, with the continued fraction providing exact eigenvalues via the Mathieu characteristic values~\cite{Koch2007}.

\textit{Cavity mode (non-compact flux).} A linear resonator has flux $\Phi \in \mathbb{R}$ ranging over the entire real line. The Fock basis $\{|k\rangle\}_{k=0}^\infty$ of photon number eigenstates is natural. If a cosine potential $\cos[\lambda(\hat{a}+\hat{a}^\dagger)]$ acts on this mode, its matrix elements involve Laguerre polynomials (Proposition~\ref{prop:cosine_matrix}) and the Hamiltonian is not tridiagonal---it couples $|n\rangle$ to $|n \pm 2\rangle$, $|n \pm 4\rangle$, etc.\ with amplitudes decaying as $\lambda^{2k}/(k!)$.

\textit{Coupled transmon-resonator system.} The correct procedure uses the product basis:
\begin{multline}
\{|q, n_r\rangle\} = \{|g\rangle, |e\rangle, |f\rangle, \ldots\}_{\mathrm{transmon}} \\
\otimes \{|0\rangle, |1\rangle, |2\rangle, \ldots\}_{\mathrm{resonator}},
\end{multline}
where $|g\rangle$, $|e\rangle$, $|f\rangle$, \ldots are the transmon eigenstates obtained from charge-basis diagonalization. The coupling operators $\hat{n}$ and $\hat{\varphi}$ are expressed in the transmon eigenbasis:
\begin{equation}
n_{ij} = \langle \psi_i | \hat{n} | \psi_j \rangle, \quad \varphi_{ij} = \langle \psi_i | \hat{\varphi} | \psi_j \rangle,
\end{equation}
where $|\psi_j\rangle = \sum_n c_n^{(j)} |n\rangle$ are the charge-basis eigenvectors. This is the standard approach in circuit QED~\cite{Koch2007, Blais2021} and numerical packages~\cite{Groszkowski2021}.
\end{remark}

\subsection{The Exact Quantum Hamiltonian}
\label{sec:exact_hamiltonian}

Quantization proceeds by promoting the classical variables to operators satisfying canonical commutation relations. In the mode basis, we introduce creation and annihilation operators through
\begin{align}
\hat{\xi}_n &= \sqrt{\frac{\hbar}{2\omega_n}} (\hat{a}_n + \hat{a}_n^\dagger), \label{eq:ladder_ops_xi}\\
\hat{p}_n &= i\sqrt{\frac{\hbar\omega_n}{2}} (\hat{a}_n^\dagger - \hat{a}_n), \label{eq:ladder_ops_p}
\end{align}
with $[\hat{a}_m, \hat{a}_n^\dagger] = \delta_{mn}$. The junction flux operator is
\begin{equation}
\hat{\Phi}_J = \sum_{n=1}^{M} \Phi_n^{J,\mathrm{zpf}} (\hat{a}_n + \hat{a}_n^\dagger),
\label{eq:Phi_J_quantum}
\end{equation}
where $\Phi_n^{J,\mathrm{zpf}} = \phi_n^J \sqrt{\hbar/(2\omega_n)}$ is the zero-point flux at the junction from mode $n$, and the dimensionless coupling is $\lambda_n = \Phi_n^{J,\mathrm{zpf}}/\varphi_0$. The complete quantum Hamiltonian is
\begin{equation}
\hat{H} = \sum_{n=1}^{M} \hbar\omega_n \hat{a}_n^\dagger \hat{a}_n - E_J \cos\left( \frac{\hat{\Phi}_J}{\varphi_0} \right).
\label{eq:exact_hamiltonian}
\end{equation}
This is the central result of the quantization procedure. The Hamiltonian is exact within the lumped-element circuit model; no approximation has been made to the Josephson nonlinearity. The full cosine potential is retained, which is essential for correctly describing the anharmonicity of the transmon spectrum, for treating regimes of ultrastrong coupling where phase fluctuations are not perturbatively small, and for capturing tunneling processes between adjacent potential wells.

The Hamiltonian generates the continued fraction spectral problem. In the Fock basis for linear modes and charge basis for the transmon, the linear part is diagonal while the cosine potential couples adjacent states, producing a block-tridiagonal structure in the appropriate product basis. The continued fraction framework applies uniformly across all coupling regimes. In the dispersive regime ($\lambda_n \ll 1$), the tridiagonal approximation with bandwidth $P=1$ suffices, and the scalar continued fraction gives accurate eigenvalues. In the ultrastrong regime ($\lambda_n \sim 0.1$--$1$), higher bandwidth $P > 1$ is needed, and the matrix continued fraction provides exact results. In the deep strong regime ($\lambda_n > 1$), large Fock cutoff $N$ and bandwidth $P$ are required, with convergence monitored via projector leakage. Different regimes require different numerical parameters, not different methods.

\subsection{Matrix Elements of the Cosine Potential}
\label{sec:cosine_matrix}

To solve the eigenvalue problem for the Hamiltonian~\eqref{eq:exact_hamiltonian}, we evaluate the matrix elements of the cosine potential in the Fock basis. The choice of basis is crucial: for cavity modes, the Fock basis $\{|k\rangle\}_{k=0}^\infty$ is appropriate, with cosine matrix elements given by Laguerre polynomials. For the transmon alone, the charge basis $\{|n\rangle\}_{n \in \mathbb{Z}}$ is natural, where the cosine is exactly tridiagonal: $\langle m|\cos\hat{\varphi}|n\rangle = \frac{1}{2}(\delta_{m,n+1} + \delta_{m,n-1})$. For a single mode with dimensionless coupling $\lambda$, the cosine operator can be expressed using displacement operators as $\cos(\lambda(\hat{a} + \hat{a}^\dagger)) = \frac{1}{2}[\hat{D}(i\lambda) + \hat{D}(-i\lambda)]$~\cite{CahillGlauber1969}.

\begin{proposition}[Cosine matrix elements in Fock basis]
\label{prop:cosine_matrix}
For a linear oscillator mode in the Fock basis with dimensionless coupling $\lambda$, the matrix elements of the cosine potential are
\begin{multline}
\mathcal{C}_{nm} \equiv 
\langle n | \cos[\lambda(\hat{a} + \hat{a}^\dagger)] | m \rangle \\
=
\small
\begin{cases}
e^{-\lambda^2/2}
\sqrt{\dfrac{n_<!}{n_>!}}\,
\lambda^{|n-m|}
L_{n_<}^{(|n-m|)}(\lambda^2)
(-1)^{|n-m|/2}, 
& n-m \ \text{even}, \\[6pt]
0, 
& n-m \ \text{odd}.
\end{cases}
\label{eq:cosine_matrix_revised}
\end{multline}
where $n_< = \min(n,m)$, $n_> = \max(n,m)$, and $L_n^{(k)}(x)$ is the generalized Laguerre polynomial. The factor $(-1)^{|n-m|/2}$ arises from the phase of the displacement operator at pure imaginary argument.

These matrix elements apply to non-compact oscillator modes (cavities, transmission line resonators). For the transmon junction with compact phase variable, use the charge-basis representation Eq.~\eqref{eq:cosine_charge_tridiag}, which is exactly tridiagonal.
\end{proposition}

\begin{corollary}[Selection rules]
The cosine matrix element $\mathcal{C}_{nm}$ vanishes unless $n - m \equiv 0 \pmod{2}$. The Hilbert space decomposes into even and odd parity sectors, with the Hamiltonian block-diagonal in this decomposition.
\end{corollary}

Within each parity sector, the Hamiltonian is not tridiagonal in the Fock basis---the cosine couples $|0\rangle$ to $|2\rangle$, $|4\rangle$, etc. However, the coupling strength decays rapidly with $|n-m|$ for $\lambda < 1$: $|\mathcal{C}_{n,n+2k}| \sim \lambda^{2k}/(k!)$ for large $k$. This justifies the bandwidth truncation: keeping only couplings with $|n-m| \leq 2P$ gives a block-banded matrix that converges rapidly as $P \to \infty$.

For multimode systems with $\hat{\theta} = \sum_j \lambda_j (\hat{a}_j + \hat{a}_j^\dagger)$, the cosine matrix elements do not factorize as $\prod_j \mathcal{C}_{m_j n_j}^{(j)}$ because $\prod_j \cos(\lambda_j \hat{X}_j) \neq \cos(\sum_j \lambda_j \hat{X}_j)$. The correct formula is
\begin{multline}
\langle {m} | \cos\hat{\theta} | {n} \rangle = \frac{1}{2} \biggl[ \prod_{j=1}^{M} D_{m_j n_j}(+i\lambda_j) \\
+ \prod_{j=1}^{M} D_{m_j n_j}(-i\lambda_j) \biggr],
\label{eq:multimode_cosine_correct}
\end{multline}
where $D_{mn}(\alpha) = \langle m | \hat{D}(\alpha) | n \rangle$. The cosine potential conserves total excitation parity $(-1)^{\sum_j n_j}$, not the parity of each mode separately.

\subsection{Truncation Scheme and Block Structure}
\label{sec:truncation}

Practical computation requires truncation of the infinite-dimensional Hilbert space. We employ a two-parameter scheme.

\begin{definition}[Two-parameter truncation]
The truncated Hamiltonian $\hat{H}^{(N,P)}$ restricts to Fock states with $n \le N$ and retains cosine couplings only for $|\Delta n| \le 2P$:
\begin{multline}
\hat{H}^{(N,P)} = \sum_{n=0}^{N} n\hbar\omega |n\rangle\langle n| \\
- E_J \sum_{\substack{m,n \le N \\ |m-n| \le 2P}} \mathcal{C}_{mn} |m\rangle\langle n|.
\label{eq:H_NP}
\end{multline}
\end{definition}

The Fock cutoff $N$ determines the depth of the continued fraction, with truncation at $N$ levels corresponding to a finite CF with $\lfloor N/2 \rfloor$ terms per parity sector. The bandwidth $P$ determines the block size in the matrix continued fraction: $P=1$ gives scalar CFs, while $P > 1$ gives $P \times P$ matrix CFs.


For fixed bandwidth $P$, the Hamiltonian $\hat{H}^{(N,P)}$ admits a block-banded representation. Grouping Fock states within each parity sector into blocks of size $P$, the Schr\"odinger equation becomes the block recurrence
\begin{equation}
{A}_k {d}_k + {B}_k {d}_{k+1} + {B}_{k-1}^\dagger {d}_{k-1} = E\, {d}_k,
\label{eq:block_recurrence}
\end{equation}
where ${A}_k, {B}_k \in \mathbb{C}^{P \times P}$ are computed from the Laguerre matrix elements. This block recurrence is the matrix generalization of the Jacobi eigenvalue equation.

\begin{theorem}[Exact spectrum via matrix continued fraction]
\label{thm:exact_block_cf}
Let $K = \lfloor N/(2P) \rfloor$. The eigenvalues of $\hat{H}^{(N,P)}$ in each parity sector are the values of $E$ for which
\begin{equation}
\det[{G}_{00}(E)^{-1}] = 0,
\label{eq:secular_det}
\end{equation}
where the matrix $m$-function is given by the block continued fraction
\begin{multline}
{G}_{00}(E) = \bigl[{A}_0 - E{I} \\
- {B}_0^\dagger \bigl[{A}_1 - E{I} - {B}_1^\dagger \bigl[\cdots\bigr]^{-1} {B}_1\bigr]^{-1} {B}_0\bigr]^{-1}
\label{eq:matrix_cf}
\end{multline}
with the CF terminating at block $K$ via ${\Sigma}_{K+1} = {0}$.
\end{theorem}

The matrix continued fraction is the quantum analog of the Cauer continued fraction for network admittances (see Table~\ref{tab:cf_framework}). For block size $P = 1$, the matrices reduce to scalars and the matrix CF becomes the scalar CF
\begin{equation}
G^{(+)}(E) = \cfrac{1}{A_0 - E - \cfrac{|B_0|^2}{A_1 - E - \cfrac{|B_1|^2}{A_2 - E - \cdots}}}
\label{eq:scalar_cf}
\end{equation}
with eigenvalues at the zeros. This is the Jacobi matrix resolvent with $a_k = A_k$ and $b_k^2 = |B_{k-1}|^2$. The tridiagonal approximation is accurate only for $\lambda \lesssim 0.2$.

\subsection{Convergence Theory}
\label{sec:convergence}

\begin{theorem}[Bandwidth convergence]
\label{thm:bandwidth_convergence}
For fixed finite Fock cutoff $N$, the bandwidth-$P$ truncated Hamiltonian converges to the full-bandwidth Hamiltonian in operator norm:
\begin{equation}
\|\hat{H}^{(N,\infty)} - \hat{H}^{(N,P)}\|_{\mathrm{op}} \le \|\Delta\hat{H}^{(P)}\|_{\mathrm{HS}} \xrightarrow{P \to \infty} 0.
\end{equation}
\end{theorem}

The convergence follows from the super-exponential decay $\lambda^{2p}/(p!)$ of the cosine matrix elements for $|n-m| = 2p$. This quantifies the error from truncating the matrix CF at block size $P$: for $\lambda \ll 1$, the tridiagonal approximation ($P=1$) is accurate, while for $\lambda \sim 1$, larger $P$ is needed.

\begin{theorem}[Fock cutoff convergence]
\label{thm:fock_convergence}
For fixed bandwidth $P$, the eigenvalues $E_k^{(N)}$ of $\hat{H}^{(N,P)}$ satisfy the variational bounds $E_k^{(N)} \ge E_k$ and $E_k^{(N+1)} \le E_k^{(N)}$, so $E_k^{(N)} \downarrow E_k$ as $N \to \infty$.
\end{theorem}

The Fock cutoff determines the CF depth. Increasing $N$ adds more levels, improving the approximation. The monotonic convergence reflects the variational principle: longer continued fractions give tighter eigenvalue bounds.

\begin{lemma}[A posteriori error indicator]
\label{lem:leakage}
The projector leakage $\eta_k = \sum_{n > N - 2P} |\langle n | \psi_k^{(N)} \rangle|^2$ provides an a posteriori indicator: if $\eta_k \ll 1$, the eigenvalue $E_k^{(N)}$ is well-converged.
\end{lemma}

\begin{theorem}[Two-truncation error]
\label{thm:two_truncation}
The total eigenvalue error satisfies:
\begin{multline}
|E_k - E_k^{(N,P)}| \le |E_k - E_k^{(N,\infty)}| \\
+ |E_k^{(N,\infty)} - E_k^{(N,P)}|.
\label{eq:two_truncation_error}
\end{multline}
The bandwidth error is bounded by Weyl's inequality; the Fock cutoff error is monitored via $\eta_k$.
\end{theorem}
Table~\ref{tab:truncation_requirements} summarizes numerically determined truncation parameters for various coupling regimes.

\begin{table}[t]
\caption{Truncation parameters for eigenvalue accuracy $< 10^{-6}$ in the lowest 5 levels ($E_J/\hbar\omega = 5$).}
\label{tab:truncation_requirements}
\begin{ruledtabular}
\begin{tabular}{lccc}
Regime & $\lambda$ & Required $P$ & Required $N$ \\
\midrule
Dispersive & $0.1$--$0.2$ & $2$--$3$ & $15$--$20$ \\
Strong coupling & $0.3$--$0.5$ & $4$--$5$ & $25$--$35$ \\
Ultrastrong & $0.6$--$1.0$ & $6$--$9$ & $40$--$50$ \\
Deep strong & $1.2$--$1.5$ & $10$--$12$ & $55$--$65$ \\
\end{tabular}
\end{ruledtabular}
\end{table}

\subsection{Connection to the Mathieu Equation}
\label{sec:mathieu}

For an isolated transmon, the eigenvalue problem in the phase representation reduces to the Mathieu equation. In the phase representation where $\hat{\varphi} = \hat{\Phi}_J/\varphi_0$ and $\hat{n} = -i\partial/\partial\varphi$, the Hamiltonian is
\begin{equation}
\hat{H} = 4E_C \hat{n}^2 - E_J \cos\hat{\varphi},
\label{eq:H_phase}
\end{equation}
with $E_C = e^2/(2C_\Sigma)$ the charging energy and $C_\Sigma$ the total capacitance shunting the junction~\cite{Koch2007}. The factor of 4 arises because the charge $Q = 2en$ involves Cooper pairs carrying charge $2e$. Introducing $z = \varphi/2$ and $q = E_J/(2E_C)$ yields the canonical Mathieu form.

The Mathieu characteristic values are determined by continued fraction equations arising from the Fourier expansion of Mathieu functions. In the charge basis $\{|n\rangle\}_{n \in \mathbb{Z}}$, the Hamiltonian~\eqref{eq:H_phase} is exactly tridiagonal:
\begin{equation}
\langle n | \hat{H} | m \rangle = 4E_C n^2 \delta_{nm} - \frac{E_J}{2}(\delta_{n,m+1} + \delta_{n,m-1}).
\end{equation}
This is a Jacobi matrix with $a_n = 4E_C n^2$ and $b_n = -E_J/2$, and its resolvent is a continued fraction whose zeros are the Mathieu characteristic values. This charge-basis CF is exact for all $E_J/E_C$ with no bandwidth truncation required.

In the transmon limit $q \gg 1$, the characteristic values admit the asymptotic expansion
\begin{multline}
a_m(q) \approx -2q + 2(2m+1)\sqrt{q} \\
- \frac{(2m+1)^2 + 1}{4} + O(q^{-1/2}),
\label{eq:mathieu_asymptotic}
\end{multline}
yielding the plasma frequency $\sqrt{8E_J E_C}$ and anharmonicity $-E_C$.

\begin{remark}[Connection to Braak's exact Rabi solution]
Braak~\cite{Braak2011} showed that the quantum Rabi model $H = \omega a^\dagger a + \Delta\sigma_z + g\sigma_x(a+a^\dagger)$ has spectrum given by zeros of a transcendental function $G_\pm(x)$ built from a three-term recurrence. A key insight is that naive continued fractions based on this recurrence yield tautologies; the spectral condition requires exploiting the $\mathbb{Z}_2$ parity symmetry to compare two series expansions. For the transmon, the charge-basis representation naturally incorporates parity sectors $(n \mod 2)$, and the CF for each sector gives exact eigenvalues. The transmon is not a two-level system, so Braak's exact solution does not directly apply, but the importance of parity sectors and the CF structure carry over.
\end{remark}

\subsection{Variational Ground State}
\label{sec:variational}

An exact variational bound on the ground state energy is obtained using a product of squeezed vacuum states:
\begin{equation}
|\Psi_{\mathrm{var}}\rangle = \prod_{n=1}^{M} \hat{S}_n(r_n) |0\rangle,
\label{eq:variational_ansatz}
\end{equation}
where $\hat{S}_n(r_n) = \exp[(r_n \hat{a}_n^2 - r_n (\hat{a}_n^\dagger)^2)/2]$ is the squeezing operator. With $r_n > 0$, this squeezes the position quadrature $\hat{X}_n = \hat{a}_n + \hat{a}_n^\dagger$, reducing its variance to $\langle \hat{X}_n^2 \rangle = e^{-2r_n}$. The total variational energy is
\begin{multline}
E_{\mathrm{var}}(\{r_n\}) = \sum_n \hbar\omega_n \sinh^2 r_n \\
- E_J \exp\left(-\frac{1}{2}\sum_m \lambda_m^2 e^{-2r_m}\right).
\label{eq:E_var_total}
\end{multline}

\begin{proposition}[Squeezed vacuum variational bound]
\label{prop:variational_squeezed}
The squeezed vacuum ansatz~\eqref{eq:variational_ansatz} provides an upper bound $E_0 \leq E_{\mathrm{var}}(\{r_n^*\})$ where $\{r_n^*\}$ minimize~\eqref{eq:E_var_total}. For weak coupling $\lambda \ll 1$ in the single-mode case, the optimal squeezing is $r^* \approx \lambda^2 E_J/(2\hbar\omega)$ and the variational energy approaches the exact result with error $O(\lambda^4)$.
\end{proposition}
The ground state photon number in this approximation is $\bar{n}_0^{\mathrm{var}} = \sinh^2 r^*$, which vanishes as $\lambda^4$ for weak coupling and grows as $\sqrt{E_J\lambda^2/\omega}$ for strong coupling. This provides an analytical diagnostic for the crossover from dispersive to ultrastrong coupling regimes.

\subsection{Connection to the Boundary Condition}
\label{sec:bc_connection}

The dressed mode frequencies $\omega_n$ in the Hamiltonian~\eqref{eq:exact_hamiltonian} are the roots of the boundary condition equation $s\Yin(s) + 1/L_J = 0$ derived in Sec.~\ref{sec:schur}. The junction participation ratio $(\phi_n^J)^2 = 1/C_n^{\mathrm{eff}}$ is given by Eq.~\eqref{eq:participation_from_Y}, and the dimensionless coupling is
\begin{equation}
\lambda_n = 2\sqrt{\frac{E_C^{(n)}}{\hbar\omega_n}},
\label{eq:lambda_from_EC}
\end{equation}
where $E_C^{(n)} = e^2/(2C_n^{\mathrm{eff}})$. Equivalently, $\lambda_n^2 = 4E_C^{(n)}/(\hbar\omega_n)$.

The continued fraction framework provides a complete pipeline from admittance to spectrum. Starting from the boundary admittance $\Yin(s)$, one solves the boundary condition for dressed frequencies (CF zeros), extracts participations from CF residues, computes couplings $\lambda_n$, constructs the Hamiltonian, diagonalizes via matrix CF, and extracts transition frequencies, matrix elements, and decay rates. Each step preserves the CF structure, enabling exact results across all coupling regimes without perturbative approximation.


\section{Ultraviolet Convergence}
\label{sec:uv}

A fundamental concern in the quantization of circuits coupled to multimode environments is the convergence of sums over modes. Naive estimates suggest potential divergence: if the coupling to mode $n$ scales as $g_n \propto \sqrt{\omega_n}$ from the vacuum fluctuation amplitude combined with capacitive coupling, then sums like $\sum_n g_n^2/\omega_n$ might diverge logarithmically. This would render quantities like the Lamb shift infinite, requiring renormalization procedures borrowed from quantum field theory. In circuit QED, closely related ultraviolet concerns and their resolution by consistently modeling the linear environment rather than imposing an external cutoff have been addressed in a number of works~\cite{Paladino2003, Gely2017, Malekakhlagh2016, Malekakhlagh2017, ParraRodriguez2018}. This section proves a convergence theorem within the boundary-condition framework. Under the assumption that $\Yin$ includes finite shunt capacitance at the junction node, the junction participation decays as $\omega_n^{-1}$ at high frequencies, ensuring absolute convergence of all physically relevant sums. The continued fraction structure provides a natural perspective: the leading Cauer coefficient $C_\Sigma$ determines the high-frequency capacitive limit, while the CF tail controls the participation decay. The ultraviolet suppression $\phi_n^J = O(\omega_n^{-1})$ is equivalent to the statement that the CF residues decay as $O(\omega_n^{-2})$, ensuring convergence of the spectral sum.


\subsection{Statement of Assumptions}
\label{sec:uv_assumptions}

The ultraviolet convergence theorem requires three assumptions, which we state explicitly to delineate the scope of the result.

\begin{assumption}[Finite Shunt Capacitance]
\label{assumption:uv_capacitance}
The admittance $\Yin(s)$ includes a finite positive shunt capacitance $C_\Sigma > 0$ at the junction node, so that
\begin{equation}
\Yin(i\omega) = i\omega C_\Sigma + Y_{\mathrm{rem}}(i\omega),
\label{eq:Y_high_freq}
\end{equation}
where the remainder satisfies $|Y_{\mathrm{rem}}(i\omega)| = O(1)$ as $\omega \to \infty$.
\end{assumption}

This assumption is automatically satisfied when the junction capacitance $C_J$ is included in $\Yin$ per Assumption~\ref{assumption:cj_convention}, since then $C_\Sigma \geq C_J > 0$. In the continued fraction framework, this assumption is equivalent to stating that the Cauer Type~I expansion begins with a finite shunt capacitor, with the CF tail contributing $O(1)$ corrections at high frequencies.

\begin{assumption}[Finite Josephson Inductance]
\label{assumption:uv_inductance}
The Josephson energy $E_J$ is finite and positive, so the Josephson inductance $L_J = \varphi_0^2/E_J$ is finite and positive.
\end{assumption}

This excludes the trivial case of no junction ($E_J = 0$, $L_J = \infty$) and the singular limit of infinite critical current ($E_J \to \infty$, $L_J \to 0$).

\begin{assumption}[Positive-Real Admittance]
\label{assumption:uv_pr}
The admittance $\Yin(s)$ is positive-real, representing a physically realizable passive network.
\end{assumption}

This assumption is guaranteed by the Schur complement construction from a passive multiport network (Proposition~\ref{prop:pr_preservation}). The positive-real property is equivalent to positivity of all Cauer CF coefficients ($L_n > 0$, $C_n > 0$), ensuring that the synthesized ladder network has physical components and that the spectral measure is positive.

\subsection{Structure of the Foster-Synthesized Circuit}
\label{sec:foster_structure}

Consider a Foster-synthesized circuit (Section~\ref{sec:foster_synthesis}) with $N$ resonant branches at frequencies $\omega_1 < \omega_2 < \cdots < \omega_N$, plus a shunt capacitance $C_\Sigma$ at the junction node. The circuit has $N+1$ nodes: the junction node (labeled $J$) and one internal node for each $LC$ branch. The capacitance matrix has the block structure
\begin{equation}
\mathsf{C} = \begin{pmatrix}
C_\Sigma & -C_1 & \cdots & -C_N \\
-C_1 & C_1 & \cdots & 0 \\
\vdots & \vdots & \ddots & \vdots \\
-C_N & 0 & \cdots & C_N
\end{pmatrix},
\label{eq:C_matrix_foster}
\end{equation}
where $C_k$ is the capacitance of the $k$-th series $LC$ branch and $C_\Sigma = C_J + \sum_k C_k$. The inverse inductance matrix is diagonal: $\mathsf{L}^{-1} = \mathrm{diag}(L_J^{-1}, L_1^{-1}, \ldots, L_N^{-1})$, with the branch resonance condition $\omega_k^2 = 1/(L_k C_k)$.

While the Foster form is convenient for explicit participation calculations, the Cauer form connects directly to the CF spectral theory of Section~\ref{sec:cf_framework}. The Foster and Cauer syntheses produce different matrix structures but encode the same admittance $\Yin(s)$. The Foster form has the ``star'' structure~\eqref{eq:C_matrix_foster} with all branches coupled to the junction node; the Cauer form yields tridiagonal matrices with nearest-neighbor coupling (the Jacobi matrix structure).

\subsection{Eigenvalue Problem and Junction Participation}
\label{sec:gep_analysis}

The dressed mode frequencies and shapes are determined by the generalized eigenvalue problem
\begin{equation}
\mathsf{L}^{-1} \boldsymbol{\phi}_n = \omega_n^2 \mathsf{C} \boldsymbol{\phi}_n,
\label{eq:gep}
\end{equation}
with the $\mathsf{C}$-orthonormality condition $\boldsymbol{\phi}_n^T \mathsf{C} \boldsymbol{\phi}_n = 1$. Writing $\boldsymbol{\phi}_n = (\phi_n^J, \phi_n^1, \ldots, \phi_n^N)^T$, the equation at internal node $k$ gives
\begin{equation}
\phi_n^k = \frac{\omega_n^2}{\omega_n^2 - \omega_k^2} \phi_n^J.
\label{eq:phi_k_solution}
\end{equation}
This relation encodes mode hybridization: the factor $\omega_n^2/(\omega_n^2 - \omega_k^2)$ has a pole at the bare resonator frequency $\omega_k$, reflecting that dressed mode $n$ is a hybridization of the junction mode with all resonator modes. The interlacing theorem (Section~\ref{sec:cf_framework}) guarantees that dressed and bare frequencies alternate, so each factor changes sign exactly once as $\omega_n$ increases through the spectrum.

Substituting Eq.~\eqref{eq:phi_k_solution} into the junction equation and requiring $\phi_n^J \neq 0$ yields the boundary condition
\begin{equation}
\frac{1}{L_J} = \omega_n^2 C_J + \sum_{k=1}^{N} \frac{\omega_n^2 C_k \omega_k^2}{\omega_n^2 - \omega_k^2},
\label{eq:bc_final_form}
\end{equation}
which matches $i\omega_n\Yin(i\omega_n) + 1/L_J = 0$ evaluated with the Foster form~\eqref{eq:foster_general}.

The $\mathsf{C}$-orthonormality condition determines the junction participation:
\begin{equation}
(\phi_n^J)^2 = \frac{1}{C_\Sigma + \sum_{k=1}^{N} \dfrac{C_k\omega_n^2(2\omega_k^2 - \omega_n^2)}{(\omega_n^2 - \omega_k^2)^2}}.
\label{eq:phi_J_squared}
\end{equation}
This is the residue of the resolvent $G(z)$ at eigenvalue $z = \omega_n^2$, providing an alternative computation via $(\phi_n^J)^2 = 1/G'(\omega_n^2)$. This result is equivalent to the admittance-derivative formula~\eqref{eq:foster_general}, confirming consistency between the Foster eigenvalue approach and the direct admittance method~\cite{Nigg2012, Minev2021}.

\begin{figure}[t]
\centering
\includegraphics[width=0.95\columnwidth]{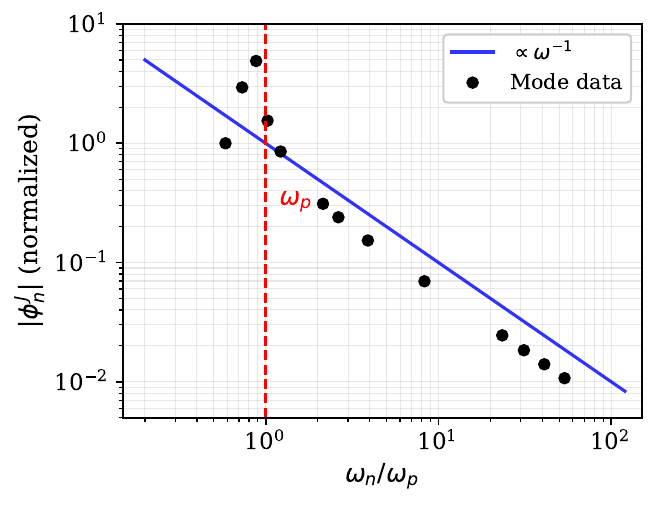}
\caption{Junction participation $|\phi_n^J|$ versus mode frequency demonstrating ultraviolet suppression (log-log scale). Black points show computed participation for dressed modes of a multimode Foster network; the blue line indicates the theoretical $\omega^{-1}$ scaling from Theorem~\ref{thm:uv_main}. Above the plasma frequency $\omega_p$ (red dashed line), the capacitive shunt at the junction port presents low impedance to high frequency modes, causing their junction participation to decay as $\phi_n^J = O(\omega_n^{-1})$. This ensures that the dimensionless coupling $\lambda_n = O(\omega_n^{-3/2})$ and all perturbative sums involving $\lambda_n^2$ converge absolutely without artificial cutoffs. The scatter at low frequencies reflects mode-dependent hybridization through the multimode environment.}
\label{fig:uv_convergence}
\end{figure}

In practice, the suppression is at least $\phi_n^J = O(\omega_n^{-1})$ in the distributed/infinite-dimensional limit, and typically faster in concrete geometries. This is the content of Theorem~\ref{thm:uv_main}. The numerical data in Fig.~\ref{fig:uv_convergence} illustrates this behavior: above $\omega_p$ the junction participation falls with a clear envelope, so that increasing the number of retained modes yields convergent sums.
\subsection{High-Frequency Asymptotic Analysis}
\label{sec:high_freq}

For $\omega_n \gg \max_k \omega_k$, we expand each term in the sum of Eq.~\eqref{eq:phi_J_squared}. Defining $x_k = \omega_k^2/\omega_n^2 \ll 1$:
\begin{equation}
\frac{\omega_n^2(2\omega_k^2 - \omega_n^2)}{(\omega_n^2 - \omega_k^2)^2} = \frac{2x_k - 1}{(1 - x_k)^2} = -1 + O(x_k^2).
\label{eq:term_expansion}
\end{equation}
The linear term in $x_k$ vanishes, so the correction is $O(x_k^2) = O(\omega_k^4/\omega_n^4)$. Summing over branches gives $\sum_k C_k \cdot (-1) + O(\omega_n^{-4})$. Using $C_\Sigma = C_J + \sum_k C_k$, the denominator of Eq.~\eqref{eq:phi_J_squared} becomes $C_J + O(\omega_n^{-4})$.
However, the formal limit $(\phi_n^J)^2 \to 1/C_J$ does not describe physical eigenmodes. For a finite Foster network, the highest frequency is bounded near the plasma frequency $\omega_p = 1/\sqrt{L_J C_\Sigma}$. The ``$\omega_n \to \infty$'' limit is meaningful only for infinite-dimensional environments. In such environments, the boundary condition $s\Yin(s) + 1/L_J = 0$ constrains which frequencies can be eigenfrequencies. For $\omega_n \gg \omega_p$, physical eigenmodes must have suppressed junction participation to satisfy current balance, leading to the $O(\omega_n^{-1})$ decay quantified below.

In the Cauer ladder picture, this suppression corresponds to localization of high-frequency eigenstates away from the junction end of the chain. The Jacobi matrix eigenvector for eigenvalue $\omega_n^2$ decays exponentially from the high-index end toward the junction, with localization length $\xi \sim 1/\ln(\omega_n/\omega_{\mathrm{typ}})$. This is the discrete analog of an evanescent wave with a node at the boundary.

\subsection{Main Theorem}
\label{sec:main_theorem}

\begin{theorem}[Ultraviolet Suppression]
\label{thm:uv_main}
Under Assumptions~\ref{assumption:uv_capacitance}--\ref{assumption:uv_pr}, the following hold for the dressed eigenmodes of the boundary condition equation:

\textup{(i)} The junction component satisfies $\phi_n^J = O(\omega_n^{-1})$ as $\omega_n \to \infty$.

\textup{(ii)} The zero-point flux satisfies $\Phi_n^{J,\mathrm{zpf}} = O(\omega_n^{-3/2})$.

\textup{(iii)} The dimensionless coupling satisfies $\lambda_n = O(\omega_n^{-3/2})$.

\textup{(iv)} All sums $\sum_n \lambda_n^2 f(\omega_n)$ converge absolutely when $f(\omega) = O(\omega^k)$ with $k < 2$.
\end{theorem}

Statement (i) is equivalent to CF residues decaying as $(\phi_n^J)^2 = O(\omega_n^{-2})$, and statement (iv) implies that the spectral measure has finite moments $\int \omega^k d\sigma(\omega) < \infty$ for $k < 1$.

\begin{proof}
Part (i): For lumped Foster networks, Section~\ref{sec:analytic_bounds_lumped} establishes
\begin{equation}
|\phi_n^J|^2 \leq \frac{\omega_p^2}{C_J \omega_n^2}\left(1 + O(\omega_p^2/\omega_n^2)\right) \quad \text{for } \omega_n > \omega_p.
\end{equation}
For distributed transmission lines, Section~\ref{sec:analytic_bounds_distributed} derives
\begin{equation}
(\phi_n^J)^2 = \frac{2 C_r v^2}{C_J^2 L^2 \omega_n^2} + O(\omega_n^{-4}).
\end{equation}

Part (ii): With $\Phi_n^{J,\mathrm{zpf}} = \phi_n^J\sqrt{\hbar/(2\omega_n)}$ and $\phi_n^J = O(\omega_n^{-1})$, we have $\Phi_n^{J,\mathrm{zpf}} = O(\omega_n^{-3/2})$.

Part (iii): The dimensionless coupling $\lambda_n = \Phi_n^{J,\mathrm{zpf}}/\varphi_0$ inherits the same scaling.

Part (iv): With $\lambda_n^2 = O(\omega_n^{-3})$, the summand scales as $O(\omega_n^{k-3})$. For mode density $\rho(\omega) = O(1)$, convergence requires $k - 3 < -1$, i.e., $k < 2$.
\end{proof}

\subsection{Physical Interpretation}
\label{sec:physical_interp}

The ultraviolet suppression has a direct physical interpretation: the junction capacitance short-circuits the junction node at high frequencies. At $\omega \gg \omega_p$, the capacitive reactance $1/(\omega C_\Sigma)$ becomes much smaller than the Josephson inductive reactance $\omega L_J$. Current preferentially flows through the capacitor, effectively bypassing the nonlinear element.

High-frequency modes therefore see a nearly short-circuited boundary condition at the junction, forcing their flux to vanish there. This is analogous to the electromagnetic boundary condition at a perfect conductor. The junction capacitance provides this short-circuit behavior, with the transition occurring near the plasma frequency. This physical picture was developed in detail by Gely et al.~\cite{Gely2017} and Parra-Rodriguez et al.~\cite{ParraRodriguez2018}, who showed that the ``dressing'' of transmission line modes by the junction capacitance provides an intrinsic ultraviolet cutoff without requiring external regularization.

\subsection{Explicit Bounds for Lumped Networks}
\label{sec:analytic_bounds_lumped}

For a dressed mode at frequency $\omega_n > \omega_p$ in a Foster network, the junction participation satisfies
\begin{equation}
|\phi_n^J|^2 \leq \frac{1}{C_J} \cdot \frac{\omega_p^2}{\omega_n^2} \cdot \left(1 + \frac{\omega_p^2}{\omega_n^2 - \omega_p^2}\right),
\label{eq:sharp_bound}
\end{equation}
confirming $O(\omega_n^{-2})$ decay with explicit prefactor
\begin{equation}
\mathcal{P}_{\mathrm{lumped}} = \frac{\omega_p^2}{C_J} = \frac{1}{C_J L_J C_\Sigma}.
\label{eq:prefactor_lumped}
\end{equation}
The prefactor can be written as $(L_J C_J^2)^{-1} \cdot (C_J/C_\Sigma)$, where $C_J/C_\Sigma$ is the capacitive participation of the junction. Large junction capacitance gives stronger UV suppression.

To derive Eq.~\eqref{eq:sharp_bound}, we use Eq.~\eqref{eq:phi_J_squared} and bound the denominator from below. For $\omega_n > \omega_p$, each term in the sum satisfies $(2\omega_k^2 - \omega_n^2)/(\omega_n^2 - \omega_k^2)^2 \geq -1/\omega_n^2$, giving
\begin{equation}
(\phi_n^J)^{-2} \geq C_\Sigma - \sum_k C_k = C_J
\end{equation}
as a weak bound. The sharper bound~\eqref{eq:sharp_bound} follows from retaining the leading correction terms.

The Debye-Waller factor, defined as $\exp(-\sum_n \lambda_n^2/2)$, measures the suppression of the junction's zero-point motion due to coupling to the environment. Using Eq.~\eqref{eq:sharp_bound}:
\begin{equation}
\sum_{n=1}^{\infty} \frac{\lambda_n^2}{2} \leq \frac{\hbar}{8\varphi_0^2 C_J} + O(E_C/\hbar\omega_p).
\label{eq:dw_total_bound}
\end{equation}
For transmon parameters with $E_J/E_C \sim 50$, this evaluates to approximately $0.01$--$0.02$, confirming that the Debye-Waller reduction is a few percent effect.

\subsection{Explicit Bounds for Distributed Systems}
\label{sec:analytic_bounds_distributed}

Consider a transmission line of length $L$, characteristic impedance $Z_0 = \sqrt{l/c}$, and phase velocity $v = 1/\sqrt{lc}$, where $l$ and $c$ are the inductance and capacitance per unit length, respectively. The line is terminated by a Josephson junction at $x = 0$ with a short circuit at $x = L$. The transmission line admittance $Y_{\mathrm{TL}}(s) = (1/Z_0)\coth(sL/v)$ has an infinite continued fraction expansion corresponding to an infinite Cauer ladder.
The mode frequencies satisfy
\begin{equation}
\frac{k_n}{l} \cot(k_n L) = \frac{1}{L_J} - k_n^2 v^2 C_J,
\label{eq:transcendental_distributed}
\end{equation}
where $k_n = \omega_n/v$. This is the boundary condition derived by Parra-Rodriguez et al.~\cite{ParraRodriguez2018} for eigenvalue-dependent boundary conditions, with their parameter $\alpha = C_J/c$. For high-frequency modes ($n \gg 1$), using the ansatz $k_n L = n\pi + \delta_n$ with $|\delta_n| \ll 1$:
\begin{equation}
k_n = \frac{n\pi}{L} - \frac{C_r}{n\pi C_J L} + O(n^{-3}),
\label{eq:kn_explicit}
\end{equation}
where $C_r = cL$ is the total line capacitance.

The junction participation follows from normalization over the modified inner product that includes the boundary term~\cite{ParraRodriguez2018}:
\begin{equation}
(\phi_n^J)^2 = \frac{2 C_r v^2}{C_J^2 L^2 \omega_n^2} + O(\omega_n^{-4}),
\label{eq:phi_J_squared_distributed}
\end{equation}
confirming $O(\omega_n^{-2})$ decay with prefactor
\begin{equation}
\mathcal{P}_{\mathrm{dist}} = \frac{2 C_r v^2}{C_J^2 L^2}.
\label{eq:prefactor_distributed}
\end{equation}
This result matches the asymptotic behavior $u_n(0)^2 \sim 1/(\alpha k_n)^2$ derived in Ref.~\cite{ParraRodriguez2018}, where the mode functions $u_n(x)$ satisfy eigenvalue-dependent boundary conditions that dress the bare transmission line modes with the junction parameters.

\begin{table}[t]
\caption{Ultraviolet suppression: lumped versus distributed environments.}
\label{tab:uv_comparison}
\begin{ruledtabular}
\begin{tabular}{lcc}
Property & Lumped & Distributed \\
\hline
$(\phi_n^J)^2$ scaling & $O(\omega_n^{-2})$ & $O(\omega_n^{-2})$ \\
Prefactor & $\omega_p^2/C_J$ & $2C_r v^2/(C_J^2 L^2)$ \\
Physical origin & Capacitive shunt & Capacitive shunt \\
\end{tabular}
\end{ruledtabular}
\end{table}
Table~\ref{tab:uv_comparison} compares lumped and distributed environments. Both exhibit identical $O(\omega_n^{-2})$ scaling for $(\phi_n^J)^2$, confirming universality of ultraviolet convergence. This universal scaling was independently established by Malekakhlagh and T\"ureci~\cite{Malekakhlagh2016, Malekakhlagh2017} using regularization techniques and by Parra-Rodriguez et al.~\cite{ParraRodriguez2018} using eigenvalue-dependent boundary conditions. The universal scaling follows from the CF structure alone: any positive-real admittance with $Y(s) \sim sC_\Sigma$ at high frequencies gives CF residues scaling as $\omega_n^{-2}$, independent of the specific environment. The prefactors depend on the subleading CF terms, but the leading decay rate is universal.

\subsection{Truncation Error Estimates and Convergence}
\label{sec:truncation_error}

The explicit prefactors enable precise estimation of truncation errors. For sums $S = \sum_n \lambda_n^2 f(\omega_n)$ with $f(\omega) = O(\omega^k)$ and $k < 2$, the error from neglecting modes above $\omega_{\mathrm{max}}$ is
\begin{equation}
\Delta S(\omega_{\mathrm{max}}) \lesssim \frac{\hbar \mathcal{P}}{2\varphi_0^2} \cdot \frac{\omega_{\mathrm{max}}^{k-2}}{2-k}.
\label{eq:truncation_general}
\end{equation}
For the Lamb shift ($k = 0$): $\Delta_{\mathrm{Lamb}} \lesssim \hbar\mathcal{P}/(4\varphi_0^2\omega_{\mathrm{max}}^2)$, giving relative errors below $10^{-3}$ for $\omega_{\mathrm{max}} = 10\omega_p$. For dispersive shifts ($k = -2$): $\Delta_\chi \lesssim \hbar\mathcal{P}/(8\varphi_0^2\omega_{\mathrm{max}}^4)$. These rapid convergence rates confirm that including modes up to a few times the plasma frequency suffices for accurate results, typically requiring only 10--20 modes for percent-level accuracy.

The self-consistent squeezing equations involve sums $\sum_m \lambda_m^2 e^{2r_m}$. Since the optimal squeezing parameters satisfy $r_m \to 0$ as $\omega_m \to \infty$, the convergence of $\sum_m \lambda_m^2 e^{2r_m}$ follows from $\sum_m \lambda_m^2 < \infty$, which is guaranteed by Theorem~\ref{thm:uv_main}(iv). For modes with $\omega_m \gg \omega_p$, the self-consistent equation gives $r_m = O(\omega_m^{-4})$, so $e^{2r_m} = 1 + O(\omega_m^{-4})$. The correction to the Debye-Waller sum from squeezing is $\sum_m \lambda_m^2(e^{2r_m} - 1) = O(\sum_m \omega_m^{-7})$, which converges even more rapidly than the leading term.


\section{Multi-Mode Interference and Purcell Suppression}
\label{sec:multimode_purcell}

In realistic circuits, multiple electromagnetic modes couple to the junction, and interference between decay pathways can significantly modify or completely suppress the effective Purcell rate~\cite{BakrPurcell2025, Bakr2025PRAppl}. This section derives the multi-mode interference formula directly from the boundary-condition admittance, demonstrating that the framework naturally captures coherent pathway interference without additional phenomenological assumptions. The central result is that the qubit decay rate is determined by the real part of the admittance evaluated at the qubit frequency, weighted by the junction participation. When multiple modes contribute to this admittance, their contributions can interfere constructively or destructively depending on the relative phases and detunings. In the continued fraction framework, the qubit decay rate is proportional to the conductance $G(\omega_q) = \mathrm{Re}[Y(i\omega_q)]$, and complete Purcell suppression occurs when $G(\omega_q) = 0$---a transmission zero where the lossy network presents purely reactive impedance to the junction.

\subsection{Purcell Decay from the Boundary Condition}
\label{sec:purcell_bc}

The Purcell effect describes the modification of spontaneous emission rates when an emitter couples to a structured electromagnetic environment. The real part $G(\omega) = \mathrm{Re}[\Yin(i\omega)]$ represents dissipation, quantifying the rate at which energy injected at frequency $\omega$ is absorbed by the environment. The Purcell decay rate can be derived from the quantum fluctuation-dissipation relation~\cite{Clerk2010}. At zero temperature, the current noise spectral density seen by the junction is $S_I(\omega) = \hbar\omega G(\omega)$. Combined with the matrix element $|\langle g | \hat{\Phi}_J | e \rangle|^2 = \hbar/(2\omega_q C_q^{\mathrm{eff}})$ for a transmon in the harmonic approximation, the decay rate from the excited state $|e\rangle$ to the ground state $|g\rangle$ is~\cite{Houck2008}
\begin{equation}
\Gamma = \frac{G(\omega_q)}{C_q^{\mathrm{eff}}},
\label{eq:Gamma_G}
\end{equation}
where $G(\omega_q) = \mathrm{Re}[\Yin(i\omega_q)]$ and we used the participation relation $(\phi_q^J)^2 = 1/C_q^{\mathrm{eff}}$ from Section~\ref{sec:participation}. This result connects directly to the boundary-condition framework: the mode frequencies satisfying $i\omega_n \Yin(i\omega_n) + 1/L_J = 0$ determine the qubit spectrum, while the conductance $G(\omega_q)$ at these frequencies determines the decay rates.

In the continued fraction framework, the conductance is the spectral function of the lossy Jacobi matrix:
\begin{equation}
G(\omega) = \sum_n (\phi_n^J)^2 \cdot \frac{\omega^2 \kappa_n}{(\omega - \omega_n)^2 + (\kappa_n/2)^2},
\end{equation}
where $\tilde{\omega}_n = \omega_n - i\kappa_n/2$ are the complex mode frequencies including loss. Each mode contributes a Lorentzian peak weighted by the CF residue $(\phi_n^J)^2$. The effective capacitance is computed from the CF structure via $C_q^{\mathrm{eff}} = 1/(\phi_q^J)^2$, confirming that large participation enhances both coupling and Purcell decay.

\subsection{Single-Mode Purcell Formula}
\label{sec:single_mode_purcell}

Before treating the multi-mode case, we derive the standard single-mode Purcell formula. Consider a single resonator mode at frequency $\omega_r$ with decay rate $\kappa$ and quality factor $Q_r = \omega_r/\kappa$, capacitively coupled to the junction. The total admittance is
\begin{equation}
\Yin(s) = sC_J + \frac{s^2 C_g^2}{sC_g + Y_r(s)},
\label{eq:Yin_single_mode}
\end{equation}
which corresponds to a two-level continued fraction with a single resonant branch.

In the dispersive regime where $|\omega_q - \omega_r| \gg \kappa$, evaluating the conductance at the qubit frequency and using the coupling strength $g = \omega_q C_g/(2\sqrt{C_J C_r})$ yields
\begin{equation}
\Gamma_{\mathrm{Purcell}} = \frac{g^2 \kappa}{\Delta^2 + (\kappa/2)^2},
\label{eq:Purcell_full}
\end{equation}
where $\Delta = \omega_q - \omega_r$. This is the product of three factors: the junction participation in the resonator mode (CF residue), the resonator loss rate (imaginary part of CF pole), and a Lorentzian frequency filter (CF spectral lineshape). Purcell suppression requires making one of these factors vanish.

\subsection{Multi-Mode Admittance}
\label{sec:multimode_admittance}

For an environment with $M$ modes, the admittance with inter-mode coupling takes the Foster form
\begin{equation}
\Yin(s) = sC_\Sigma + \sum_{k=1}^{M} \frac{s C_k \omega_k^2}{s^2 + s\kappa_k + \omega_k^2},
\label{eq:Yin_multimode_lossy}
\end{equation}
where $\kappa_k = \omega_k/Q_k$. The lossy admittance corresponds to complex CF poles at $s = -\kappa_k/2 \pm i\sqrt{\omega_k^2 - \kappa_k^2/4}$, with the Jacobi matrix becoming non-Hermitian through $\mathsf{J} \to \mathsf{J} - i\boldsymbol{\Gamma}/2$.

The conductance becomes
\begin{equation}
G(\omega) = \sum_{k=1}^{M} \frac{\omega^2 C_k \omega_k^2 \kappa_k}{(\omega_k^2 - \omega^2)^2 + \omega^2\kappa_k^2},
\label{eq:G_multimode}
\end{equation}
a sum of Lorentzian peaks. Without inter-mode coupling, all terms are positive, so $G(\omega) > 0$ everywhere and no Purcell suppression is possible. Inter-mode coupling modifies the effective residues and can create sign changes, enabling $G(\omega_q) = 0$.

When inter-mode coupling is present, the system can be described by an effective non-Hermitian Hamiltonian in the single-excitation subspace~\cite{BakrPurcell2025}. In the basis $\{|e,0\rangle, |g,1_1\rangle, \ldots, |g,1_M\rangle\}$, this Hamiltonian has the matrix form
\begin{equation}
\mathsf{H}_{\mathrm{eff}} = 
\begin{pmatrix}
\omega_q & g_1 e^{i\phi_1} & \cdots & g_M e^{i\phi_M} \\
g_1 e^{-i\phi_1} & \tilde{\omega}_1 & \cdots & J_{1M} e^{i\theta_{1M}} \\
\vdots & \vdots & \ddots & \vdots \\
g_M e^{-i\phi_M} & J_{1M} e^{-i\theta_{1M}} & \cdots & \tilde{\omega}_M
\end{pmatrix},
\label{eq:H_eff_matrix}
\end{equation}
where $\tilde{\omega}_k = \omega_k - i\kappa_k/2$ are the complex mode frequencies, $g_k$ are the qubit-mode coupling strengths, $\phi_k$ are the coupling phases, and $J_{kl}e^{i\theta_{kl}}$ are the inter-mode couplings with $\theta_{kl} = -\theta_{lk}$.

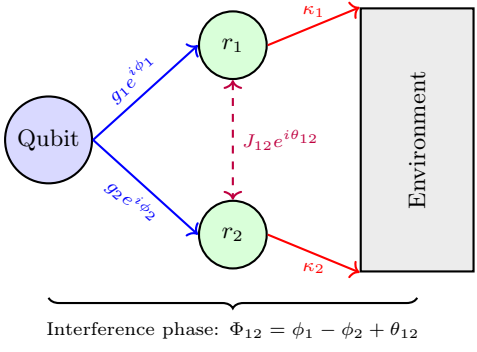
\begin{figure}[t]
\centering
\begin{tikzpicture}[scale=0.7]
    \node[draw, thick, circle, minimum size=1.1cm, fill=blue!15] (Q) at (0,0) {\small Qubit};
    
    \node[draw, thick, circle, minimum size=0.9cm, fill=green!15] (M1) at (3.5,1.8) {\small $r_1$};
    
    \node[draw, thick, circle, minimum size=0.9cm, fill=green!15] (M2) at (3.5,-1.8) {\small $r_2$};
    
    \node[draw, thick, rectangle, minimum width=1.5cm, minimum height=3.5cm, fill=gray!15] (env) at (7,0) {\rotatebox{90}{\small Environment}};
    
    \draw[thick, ->, blue] (Q.east) -- node[above, font=\scriptsize, sloped] {$g_1 e^{i\phi_1}$} (M1.west);
    \draw[thick, ->, blue] (Q.east) -- node[below, font=\scriptsize, sloped] {$g_2 e^{i\phi_2}$} (M2.west);
    
    \draw[thick, dashed, <->, purple] (M1.south) -- node[right, font=\scriptsize] {$J_{12}e^{i\theta_{12}}$} (M2.north);
    
    \draw[thick, ->, red] (M1.east) -- node[above, font=\scriptsize] {$\kappa_1$} (env.north west);
    \draw[thick, ->, red] (M2.east) -- node[below, font=\scriptsize] {$\kappa_2$} (env.south west);
    
    \draw[thick, decorate, decoration={brace, amplitude=4pt, mirror}] (0,-3) -- (7,-3);
    \node[below, font=\scriptsize] at (3.5,-3.3) {Interference phase: $\Phi_{12} = \phi_1 - \phi_2 + \theta_{12}$};
    
\end{tikzpicture}
\caption{Multi-mode Purcell interference geometry. The qubit couples to two environmental modes $r_1$ and $r_2$ with complex amplitudes $g_k e^{i\phi_k}$ determined by the junction participation and mode phase at the junction location. Inter-mode coupling $J_{12}e^{i\theta_{12}}$ creates alternative decay pathways that can interfere with direct decay. Both modes decay to a common environment with rates $\kappa_1$ and $\kappa_2$. The total decay rate depends on the interference phase $\Phi_{12} = \phi_1 - \phi_2 + \theta_{12}$, which combines the junction coupling phases with the inter-mode coupling phase.}
\label{fig:multimode_purcell}
\end{figure}

The geometry of multi-mode interference is illustrated in Fig.~\ref{fig:multimode_purcell}, showing how the qubit couples to multiple environmental modes that subsequently decay through a common port. The interference between different decay pathways depends on the phases accumulated along each path.

\subsection{Perturbative Expansion and Interference}
\label{sec:perturbative_expansion}

When inter-mode couplings $J_{kl}$ are weak compared to mode separations, the qubit decay rate can be computed perturbatively from the non-Hermitian Hamiltonian~\eqref{eq:H_eff_matrix}. The complex eigenvalue of the qubit-like state is $\lambda_q = \omega_q + \Sigma^{(2)} + \Sigma^{(3)} + O(g^4, J^2)$, where the self-energy contributions are computed order by order.

The second-order self-energy gives the direct Purcell contribution:
\begin{equation}
\Sigma^{(2)} = \sum_{k=1}^{M} \frac{g_k^2}{\Delta_k + i\kappa_k/2},
\label{eq:Sigma_2}
\end{equation}
where $\Delta_k = \omega_q - \omega_k$. Taking twice the imaginary part yields the zeroth-order decay rate:
\begin{equation}
\Gamma^{(0)} = -2\,\mathrm{Im}[\Sigma^{(2)}] = \sum_{k=1}^{M} \frac{g_k^2 \kappa_k}{\Delta_k^2 + (\kappa_k/2)^2},
\label{eq:Gamma_0}
\end{equation}
reproducing the incoherent sum without interference.

The third-order self-energy from inter-mode coupling is~\cite{BakrPurcell2025}
\begin{equation}
\Sigma^{(3)} = \sum_{k < l} \frac{2 g_k g_l J_{kl} \cos\Phi_{kl}}{(\Delta_k + i\kappa_k/2)(\Delta_l + i\kappa_l/2)},
\label{eq:Sigma_3}
\end{equation}
where $\Phi_{kl} = \phi_k - \phi_l + \theta_{kl}$ is the total interference phase. The interference contribution to the decay rate is
\begin{equation}
\Gamma_{kl}^{\mathrm{int}} = \frac{2(\kappa_k\Delta_l + \kappa_l\Delta_k) g_k g_l J_{kl} \cos\Phi_{kl}}{[\Delta_k^2 + (\kappa_k/2)^2][\Delta_l^2 + (\kappa_l/2)^2]}.
\label{eq:Gamma_kl_int}
\end{equation}
This formula matches the result derived from the Heisenberg-Langevin equations in Ref.~\cite{BakrPurcell2025}.

The complete decay rate is
\begin{equation}
\Gamma_{\mathrm{eff}} = \sum_{k=1}^{M} \Gamma_k^{\mathrm{direct}} + \sum_{k < l} \Gamma_{kl}^{\mathrm{int}},
\label{eq:Gamma_total}
\end{equation}
where the interference terms can be positive or negative depending on $\Phi_{kl}$ and the detuning signs.

\subsection{Conditions for Purcell Suppression}
\label{sec:purcell_suppression}

Complete Purcell suppression requires $G(\omega_q) = 0$, which demands the interference terms to exactly cancel the direct contributions. Since direct contributions are positive, this requires destructive interference with $\Gamma_{kl}^{\mathrm{int}} < 0$. From Eq.~\eqref{eq:Gamma_kl_int}, destructive interference occurs when $(\kappa_k\Delta_l + \kappa_l\Delta_k)\cos\Phi_{kl} < 0$. This can be achieved in several ways: if the qubit frequency lies between two mode frequencies so that $\Delta_k$ and $\Delta_l$ have opposite signs, appropriate linewidth asymmetry with $\cos\Phi_{kl} > 0$ gives destructive interference; alternatively, if both detunings have the same sign, $\cos\Phi_{kl} < 0$ produces the required cancellation.

Four necessary conditions must be satisfied. First, non-zero inter-mode coupling ($J_{kl} \neq 0$) creates alternative decay pathways that can interfere. Second, asymmetric detunings are required since symmetric placement gives symmetric interference that cannot cancel positive direct terms. Third, appropriate phase alignment is necessary, with $\cos\Phi_{kl}$ having the correct sign. Fourth, magnitude matching ensures the interference term is large enough to cancel the direct terms.

In the CF framework, these conditions translate to: the Jacobi matrix must have non-zero off-diagonal elements; the CF poles must be asymmetrically placed relative to $\omega_q$; the CF coefficients must combine for destructive interference; and the off-diagonal elements must be sufficiently large. Complete Purcell suppression occurs at discrete frequencies corresponding to transmission zeros in the conductance $G(\omega)$. Setting $G(\omega) = 0$ and solving for $\omega$ gives the suppression frequencies. Without inter-mode coupling, each term in Eq.~\eqref{eq:G_multimode} is positive, so no zeros exist. With inter-mode coupling, the modified admittance can develop zeros where destructive interference completely suppresses dissipation.

An interferometric Purcell filter (IPF) is a multi-mode structure engineered to place a transmission zero at the qubit frequency while maintaining finite conductance at the readout frequency. The CF framework enables systematic IPF design: choose the desired zero location $\omega_q$ and pass-band location $\omega_r$; determine CF pole/residue structure achieving $G(\omega_q) = 0$ with $G(\omega_r) > 0$; synthesize the corresponding Cauer ladder network; and verify performance via full CF calculation. Table~\ref{tab:purcell_cf} summarizes the correspondence between Purcell physics and continued fraction concepts.

\begin{table}[t]
\caption{Purcell physics in the continued fraction framework.}
\label{tab:purcell_cf}
\begin{ruledtabular}
\begin{tabular}{ll}
Physical concept & CF interpretation \\
\midrule
Decay rate $\Gamma$ & Spectral function $G(\omega_q)$ \\
Mode participation & CF residue $(\phi_k^J)^2$ \\
Mode linewidth & Imaginary part of CF pole \\
Multi-mode interference & Off-diagonal CF elements \\
Purcell suppression & Transmission zero \\
Lamb shift & Real part of self-energy \\
Qubit frequency & Root of $i\omega Y_{\mathrm{in}}(i\omega) + 1/L_J = 0$ \\
\end{tabular}
\end{ruledtabular}
\end{table}
\section{Coupling Regimes in Circuit QED}
\label{sec:coupling_regimes}

The interaction between superconducting qubits and microwave resonators spans several distinct physical regimes, each characterized by different approximations, observables, and theoretical methods. This section provides a unified classification based on two dimensionless ratios that emerge naturally from the continued fraction (CF) framework: the normalized coupling $\eta \equiv g/\omega_r$ comparing the coupling strength $g$ to the resonator frequency $\omega_r$, and the dispersive parameter $\xi \equiv g/|\Delta|$ comparing $g$ to the qubit-resonator detuning $\Delta = \omega_q - \omega_r$. These parameters determine when the rotating-wave approximation (RWA) is valid, when perturbation theory applies, and when the full quantum Rabi model must be solved exactly. Throughout this section we set $\hbar = 1$ unless otherwise noted. Figure~\ref{fig:coupling_regimes} shows the coupling regime phase diagram with experimental realizations marked. Standard transmon readout operates deep in the dispersive regime ($\eta \ll 0.1$, $\xi < 0.1$), while ultrastrong coupling experiments have pushed $\eta$ toward and beyond unity.

\begin{figure}[t]
\centering
\includegraphics[width=0.48\textwidth]{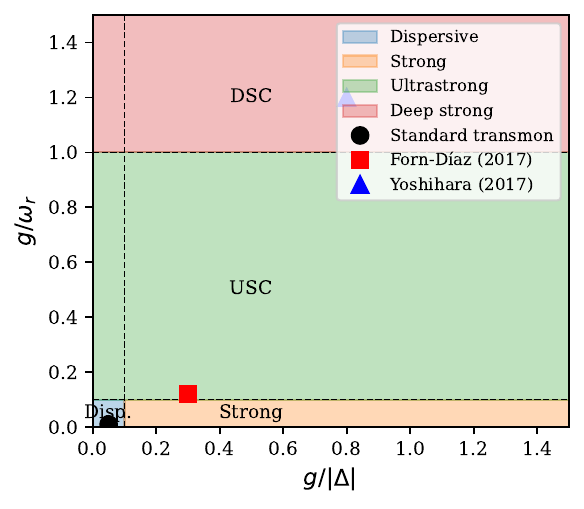}
\caption{Coupling regime phase diagram in the $(g/|\Delta|, g/\omega_r)$ plane. The dispersive regime (blue) supports perturbative treatment; the strong coupling regime (orange) requires exact diagonalization of the Jaynes-Cummings model; ultrastrong coupling (green) necessitates the full Rabi model without RWA; deep strong coupling (red) features substantial ground-state photon populations. Experimental points: standard transmon readout (circle), Forn-D\'iaz flux qubit devices~\cite{FornDiaz2017} (square), Yoshihara deep-strong devices~\cite{Yoshihara2017} (triangle).}
\label{fig:coupling_regimes}
\end{figure}

\subsection{The Jaynes-Cummings and Quantum Rabi Hamiltonians}
\label{sec:jc_rabi}

The light-matter interaction between a two-level system (qubit) and a single bosonic mode (resonator) takes the form
\begin{equation}
H = \frac{\omega_q}{2}\sigma_z + \omega_r a^\dagger a + g\sigma_x(a + a^\dagger),
\label{eq:rabi_full}
\end{equation}
where $\omega_q$ is the qubit transition frequency, $\omega_r$ is the resonator frequency, $g$ is the coupling strength, $\sigma_{x,z}$ are Pauli matrices, and $a$, $a^\dagger$ are the bosonic annihilation and creation operators satisfying $[a, a^\dagger] = 1$. This is the quantum Rabi model~\cite{Rabi1936, Rabi1937}, which conserves parity $\Pi = \sigma_z(-1)^{a^\dagger a}$ but not excitation number.


In the rotating-wave approximation (RWA), valid when $g \ll \omega_q, \omega_r$, the counter-rotating terms $\sigma_+ a^\dagger$ and $\sigma_- a$ are neglected, yielding the Jaynes-Cummings Hamiltonian~\cite{JaynesCummings1963}
\begin{equation}
H_{\mathrm{JC}} = \frac{\omega_q}{2}\sigma_z + \omega_r a^\dagger a + g(\sigma_+ a + \sigma_- a^\dagger),
\label{eq:jc}
\end{equation}
where $\sigma_\pm = (\sigma_x \pm i\sigma_y)/2$ are the raising and lowering operators. The JC model conserves excitation number $N = a^\dagger a + \sigma_+ \sigma_-$ and is exactly solvable within each $N$-excitation manifold.

The RWA breakdown occurs when the counter-rotating terms produce energy shifts comparable to other scales. The leading correction is the Bloch-Siegert shift~\cite{BlochSiegert1940}
\begin{equation}
\delta\omega_{\mathrm{BS}} = \frac{g^2}{\omega_q + \omega_r},
\label{eq:bloch_siegert}
\end{equation}
valid to leading order in $g/(\omega_q + \omega_r)$ and away from small denominators. This shift becomes significant when $\eta = g/\omega_r \gtrsim 0.1$. At this threshold, the JC model fails and the full Rabi model~\eqref{eq:rabi_full} must be used.

\subsection{Regime Classification}
\label{sec:regime_classification}

We classify coupling regimes according to the dimensionless parameters $\eta$ and $\xi$:

\paragraph{Dispersive regime ($\xi < 0.1$, $\eta \ll 0.1$).} The qubit-resonator detuning $|\Delta|$ is large compared to the coupling $g$, and both are small compared to the mode frequencies. Perturbation theory in $g/|\Delta|$ converges rapidly. The qubit and resonator remain nearly separable, with weak hybridization producing the dispersive shift $\chi \propto g^2/\Delta$ that enables quantum nondemolition readout. The RWA is excellent, and truncation to two qubit levels suffices.

\begin{table}[t]
\centering
\caption{Coupling regime classification. Here $\eta = g/\omega_r$ and $\xi = g/|\Delta|$.}
\label{tab:regimes}
\begin{tabular}{lcccl}
\hline\hline
Regime & $\eta$ & $\xi$ & RWA & Method \\
\hline
Dispersive & $\ll 0.1$ & $< 0.1$ & Yes & Perturbation theory \\
Strong & $< 0.1$ & $\geq 0.1$ & Yes & JC diagonalization \\
Ultrastrong & $0.1$--$1$ & any & No & Rabi (parity CF) \\
Deep strong & $> 1$ & any & No & Rabi (parity CF) \\
\hline\hline
\end{tabular}
\end{table}

\paragraph{Strong coupling regime ($g > \kappa, \gamma$, $\eta < 0.1$).} The coupling exceeds the linewidths, enabling resolved vacuum Rabi oscillations, but remains small compared to the mode frequencies. The hybridization parameter $\xi = g/|\Delta|$ may exceed 0.1, in which case perturbation theory in $g/|\Delta|$ fails, but the RWA remains valid. Exact diagonalization of the JC Hamiltonian~\eqref{eq:jc} is required. The Jaynes-Cummings ladder of dressed states $|n,\pm\rangle$ emerges, with splittings $2g\sqrt{n}$.

\paragraph{Ultrastrong coupling (USC) regime ($0.1 \leq \eta < 1$).} The coupling becomes a significant fraction of the mode frequencies~\cite{Kockum2019, FornDiaz2019}. The RWA breaks down: counter-rotating terms produce Bloch-Siegert shifts~\eqref{eq:bloch_siegert} exceeding 1\% of transition frequencies. The ground state acquires virtual photon population $\langle n \rangle_0 \sim \eta^2$. Parity $\Pi$ remains a good quantum number, but excitation number $N$ does not. The full Rabi model~\eqref{eq:rabi_full} must be solved.

\paragraph{Deep strong coupling (DSC) regime ($\eta \geq 1$).} The coupling exceeds the mode frequencies~\cite{Casanova2010}. Light and matter become inextricably entangled even in the ground state, which acquires substantial photon population $\langle n \rangle_0 \sim \eta^2$. The photon number distribution can become bimodal. Standard perturbative expansions fail entirely. The CF framework provides a systematic approach through parity-sector tridiagonalization, connected to the exact solution via Braak's integrability proof~\cite{Braak2011}. Note that DSC is distinct from the spin-boson localization transition, which is controlled by $\alpha_{\mathrm{SB}}$ rather than $\eta$. Table~\ref{tab:regimes} summarizes the regime boundaries and applicable methods.

\subsection{The Transmon in the Dispersive Regime}
\label{sec:transmon_dispersive}

The transmon qubit~\cite{Koch2007} is a charge qubit operated at large $E_J/E_C$ ratio, where $E_J$ is the Josephson energy and $E_C = e^2/(2C_\Sigma)$ is the charging energy with $C_\Sigma$ the total capacitance. The Hamiltonian in the charge basis $\{|n\rangle\}$ is
\begin{equation}
H_{\mathrm{tr}} = 4E_C(\hat{n} - n_g)^2 - E_J\cos\hat{\varphi},
\label{eq:transmon}
\end{equation}
where $\hat{n}$ is the Cooper-pair number operator, $n_g = C_g V_g/(2e)$ is the offset charge, and $\hat{\varphi}$ is the superconducting phase conjugate to $\hat{n}$. The phase operator acts via $e^{\pm i\hat{\varphi}}|n\rangle = |n \mp 1\rangle$, rendering the Hamiltonian tridiagonal in the charge basis with CF coefficients
\begin{equation}
a_n = 4E_C(n - n_g)^2, \qquad b_n = -\frac{E_J}{2}.
\label{eq:transmon_cf}
\end{equation}
Note that $b_n$ is independent of $n$, the off-diagonal coupling is uniform, arising from the expansion $-E_J\cos\hat{\varphi} = -(E_J/2)(e^{i\hat{\varphi}} + e^{-i\hat{\varphi}})$.

Figure~\ref{fig:transmon_spectrum}(a) shows the transmon energy levels as a function of $E_J/E_C$. At large ratio, the transition frequencies approach $\omega_{01} \approx \sqrt{8E_J E_C} - E_C$ with anharmonicity $\alpha = \omega_{12} - \omega_{01} \approx -E_C$. The charge dispersion is exponentially suppressed: $\partial\omega_{01}/\partial n_g \propto e^{-\sqrt{8E_J/E_C}}$, providing protection against charge noise while maintaining sufficient anharmonicity for qubit operation.

\begin{figure}[t]
\centering
\includegraphics[width=0.5\textwidth]{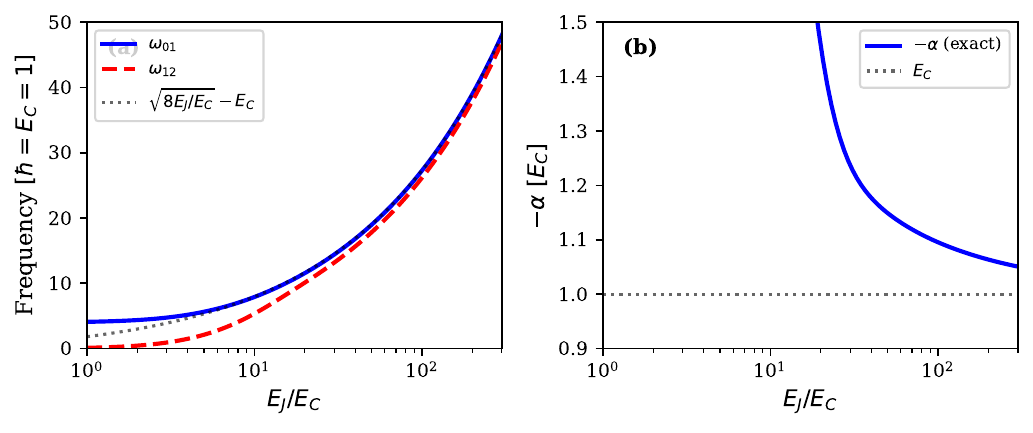}
\caption{Transmon energy spectrum from charge-basis diagonalization. (a) Transition frequencies $\omega_{01}$ and $\omega_{12}$ versus $E_J/E_C$. The dashed line shows the plasma frequency approximation $\sqrt{8E_JE_C} - E_C$. (b) Anharmonicity $-\alpha = \omega_{01} - \omega_{12}$ versus $E_J/E_C$. At large ratio, $|\alpha| \to E_C$ (dotted line). The transmon regime ($E_J/E_C \gtrsim 50$) balances charge noise immunity against addressability.}
\label{fig:transmon_spectrum}
\end{figure}

When coupled to a resonator via a capacitance $C_c$, the interaction Hamiltonian is
\begin{equation}
H_{\mathrm{int}} = 2e\beta V_{\mathrm{zpf}} \hat{n}(a + a^\dagger),
\label{eq:capacitive_coupling}
\end{equation}
where $\beta = C_c/C_\Sigma$ is the coupling capacitance ratio and $V_{\mathrm{zpf}} = \sqrt{\omega_r/(2C_r)}$ is the resonator zero-point voltage. The coupling strength between transmon states $|j\rangle$ and $|k\rangle$ is
\begin{equation}
g_{jk} = 2e\beta V_{\mathrm{zpf}} \langle j|\hat{n}|k\rangle.
\label{eq:coupling_matrix_element}
\end{equation}
The charge matrix elements $n_{jk} = \langle j|\hat{n}|k\rangle$ are computed from the CF eigenstates. 

\subsection{Dispersive Shift Derivation}
\label{sec:dispersive_shift}

The dispersive shift $\chi$ is the qubit-state-dependent frequency pull on the resonator:
\begin{equation}
\omega_r^{|e\rangle} - \omega_r^{|g\rangle} = 2\chi.
\label{eq:chi_def}
\end{equation}
It arises from virtual photon exchange between the qubit and resonator. The CF framework provides a systematic derivation through the excitation-manifold structure.

\paragraph{One-excitation manifold.} In the subspace $\{|g,1\rangle, |e,0\rangle\}$, the JC Hamiltonian is
\begin{equation}
H^{(1)} = \begin{pmatrix}
\omega_r & g_{01} \\
g_{01} & \omega_{01}
\end{pmatrix},
\label{eq:H1}
\end{equation}
which is a $2\times 2$ CF with $a_0 = \omega_r$, $a_1 = \omega_{01}$, $b_1 = g_{01}$. The poles satisfy $(z - \omega_r)(z - \omega_{01}) = g_{01}^2$. In the dispersive limit $|g_{01}| \ll |\Delta_{01}|$, expansion to $O(g^2)$ yields
\begin{equation}
z_- \approx \omega_r + \frac{g_{01}^2}{\Delta_{01}}, \qquad z_+ \approx \omega_{01} - \frac{g_{01}^2}{\Delta_{01}},
\label{eq:one_exc_shifts}
\end{equation}
where $\Delta_{01} = \omega_{01} - \omega_r$. The resonator-like state is pushed by $g_{01}^2/\Delta_{01}$; this is the AC Stark shift of the resonator due to virtual excitation of the qubit.

\paragraph{Two-excitation manifold.} To capture the transmon's anharmonicity, we include the $|f\rangle$ (second excited) state. The subspace $\{|g,2\rangle, |e,1\rangle, |f,0\rangle\}$ has Hamiltonian
\begin{equation}
H^{(2)} = \begin{pmatrix}
2\omega_r & \sqrt{2}g_{01} & 0 \\
\sqrt{2}g_{01} & \omega_{01} + \omega_r & g_{12} \\
0 & g_{12} & \omega_{01} + \omega_{12}
\end{pmatrix},
\label{eq:H2}
\end{equation}
which is a $3\times 3$ tridiagonal (CF) matrix with coefficients $a_0 = 2\omega_r$, $b_1 = \sqrt{2}g_{01}$, $a_1 = \omega_{01} + \omega_r$, $b_2 = g_{12}$, $a_2 = \omega_{01} + \omega_{12}$.

\paragraph{Koch formula.} Combining the level repulsions from both manifolds, the dispersive shift is~\cite{Koch2007}
\begin{equation}
\boxed{\chi = \frac{g_{01}^2}{\Delta_{01}} - \frac{g_{12}^2}{2\Delta_{12}},}
\label{eq:koch_formula}
\end{equation}
where $\Delta_{12} = \omega_{12} - \omega_r = \Delta_{01} + \alpha$. The first term arises from $|g\rangle \leftrightarrow |e\rangle$ hybridization; the second from $|e\rangle \leftrightarrow |f\rangle$ hybridization. The factor of $1/2$ in the second term reflects the asymmetric contribution: the $|f\rangle$-state dressing affects $\omega_r^{|e\rangle}$ but not $\omega_r^{|g\rangle}$.

Figure~\ref{fig:dispersive_shift} shows the dispersive shift versus $g/|\Delta|$ comparing the Koch formula~\eqref{eq:koch_formula} to full JC diagonalization. The Koch formula is derived from second-order perturbation theory in $g/\Delta$; full JC diagonalization includes all orders. Agreement is excellent ($<1\%$ error) for $g/|\Delta| < 0.1$; the small discrepancy reflects $O(g^4/\Delta^3)$ corrections. At larger coupling, higher excitation manifolds contribute and the second-order formula becomes less accurate.

\begin{figure}[t]
\centering
\includegraphics[width=0.5\textwidth]{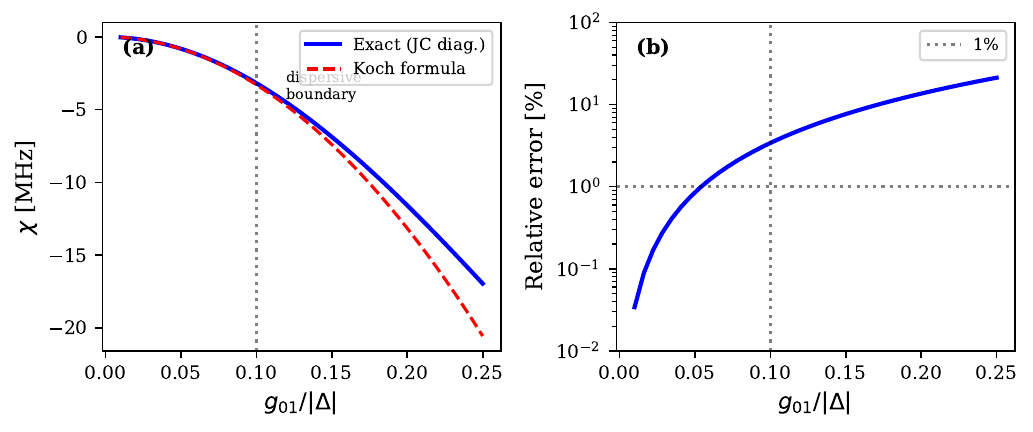}
\caption{Dispersive shift $\chi$ for a transmon-resonator system with $E_J/E_C = 60$, $\omega_r/2\pi = 7.0$ GHz. (a) Magnitude of $\chi$ versus normalized coupling $g_{01}/|\Delta|$. Solid blue: full JC diagonalization (all orders); dashed red: Koch formula~\eqref{eq:koch_formula}~\cite{Koch2007} (second-order in $g/\Delta$). The vertical dotted line marks the dispersive boundary $g/|\Delta| = 0.1$. (b) Relative error of the Koch formula. The second-order approximation maintains $<10\%$ accuracy for $g/|\Delta| < 0.15$.}
\label{fig:dispersive_shift}
\end{figure}

\paragraph{Design application.} Equation~\eqref{eq:koch_formula} serves as a design tool: given a target $\chi$ for dispersive readout, one solves for the required coupling $g_{01}$. For example, with $E_J/E_C = 60$ ($|n_{12}|/|n_{01}| \approx 1.37$), $\omega_{01}/2\pi = 5.2$ GHz, $\omega_r/2\pi = 7.0$ GHz, and target $|\chi|/2\pi = 0.5$ MHz:
\begin{equation}
g_{01}^2 \left|\frac{1}{\Delta_{01}} - \frac{(1.37)^2}{2\Delta_{12}}\right| = 0.5~\text{MHz}.
\end{equation}
With $\Delta_{01}/2\pi = -1.8$ GHz and $\Delta_{12}/2\pi = -2.1$ GHz, this yields $g_{01}/2\pi \approx 70$ MHz.

\subsection{Vacuum Rabi Splitting and Dressed States}
\label{sec:vacuum_rabi}

At resonance ($\Delta = 0$), the qubit and resonator hybridize into dressed states separated by the vacuum Rabi splitting $2g$~\cite{Wallraff2004}. This is the hallmark of strong coupling: energy exchange between light and matter faster than dissipation.

The one-excitation dressed states are
\begin{equation}
|1,\pm\rangle = \frac{1}{\sqrt{2}}(|g,1\rangle \pm |e,0\rangle),
\label{eq:dressed_one}
\end{equation}
with energies $E_{1,\pm} = \omega_r \pm g$ (taking $\omega_q = \omega_r$). Spectroscopically, one observes an avoided crossing as the qubit is tuned through resonance, with minimum splitting $2g$.

Figure~\ref{fig:vacuum_rabi}(a) shows the avoided crossing for $g/2\pi = 100$ MHz. The dressed states follow hyperbolic branches $E_\pm = (\omega_q + \omega_r)/2 \pm \sqrt{\Delta^2 + 4g^2}/2$, asymptoting to the bare states far from resonance. Panel (b) shows the splitting at resonance versus coupling strength.

\begin{figure}[t]
\centering
\includegraphics[width=0.5\textwidth]{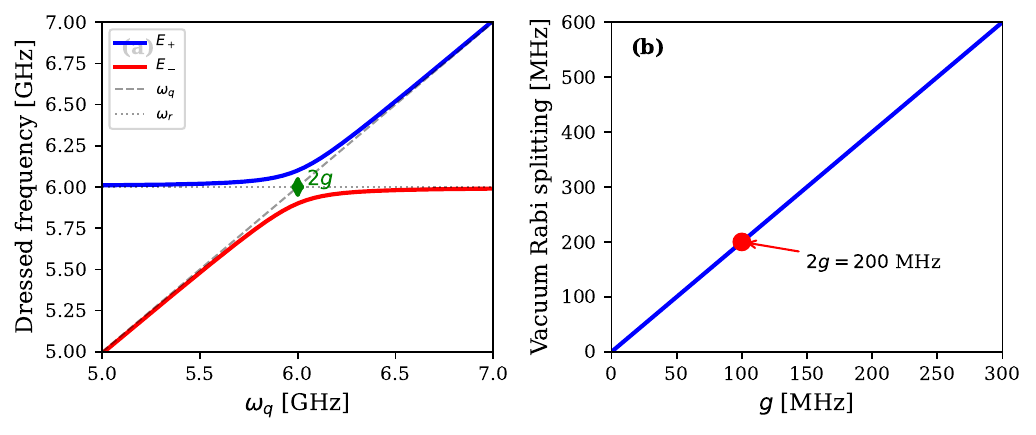}
\caption{Vacuum Rabi physics in the strong coupling regime. (a) Dressed state energies versus bare qubit frequency $\omega_q$ for fixed $\omega_r/2\pi = 6.0$ GHz and $g/2\pi = 100$ MHz. Solid lines: exact eigenvalues; dashed: bare qubit and resonator frequencies. The avoided crossing has minimum splitting $2g = 200$ MHz. (b) Vacuum Rabi splitting at resonance ($\omega_q = \omega_r$) versus coupling strength.}
\label{fig:vacuum_rabi}
\end{figure}
For higher excitation manifolds, the dressed states have splittings $2g\sqrt{n}$:
\begin{equation}
|n,\pm\rangle = \frac{1}{\sqrt{2}}(|g,n\rangle \pm |e,n-1\rangle), \qquad E_{n,\pm} = n\omega_r \pm g\sqrt{n}.
\label{eq:dressed_n}
\end{equation}
This $\sqrt{n}$ dependence is the Jaynes-Cummings nonlinearity: the ladder is not equally spaced. It enables photon-number-resolved spectroscopy and forms the basis for deterministic photon generation protocols.

\subsection{Ultrastrong Coupling: Beyond the RWA}
\label{sec:usc}

When $\eta = g/\omega_r$ exceeds approximately 0.1, the rotating-wave approximation breaks down~\cite{Kockum2019, FornDiaz2019}. The counter-rotating terms $\sigma_+ a^\dagger$ and $\sigma_- a$ in the full Rabi Hamiltonian~\eqref{eq:rabi_full} can no longer be neglected. Their leading effect is the Bloch-Siegert shift~\eqref{eq:bloch_siegert}, which adds to the Lamb shift from virtual photon processes.

The Rabi model conserves parity $\Pi = \sigma_z(-1)^{a^\dagger a}$ but not excitation number. The Hilbert space decomposes into two sectors: even parity ($\Pi = +1$) spanned by $|e,0\rangle, |g,1\rangle, |e,2\rangle, \ldots$, and odd parity ($\Pi = -1$) spanned by $|g,0\rangle, |e,1\rangle, |g,2\rangle, \ldots$. Within each sector, the Hamiltonian is tridiagonal~\cite{Casanova2010} with CF coefficients
\begin{align}
&\text{Even:} & a_k &= (-1)^k \frac{\omega_q}{2} + k\omega_r, & b_{k+1} &= g\sqrt{k+1}, \label{eq:rabi_cf_even} \\
&\text{Odd:} & a_k &= (-1)^{k+1} \frac{\omega_q}{2} + k\omega_r, & b_{k+1} &= g\sqrt{k+1}. \label{eq:rabi_cf_odd}
\end{align}
The only difference between sectors is the sign pattern in the qubit contribution to $a_k$, reflecting the alternating qubit state ($|g\rangle$ vs.\ $|e\rangle$) along each parity chain.

Figure~\ref{fig:usc_spectrum}(a) shows the Rabi model spectrum as a function of $g/\omega_r$. At small coupling, the levels follow the JC structure; as $g/\omega_r$ increases toward unity, the counter-rotating terms produce level anticrossings and eventually qualitative restructuring. For the unbiased model ($\omega_q, \omega_r > 0$), the ground state lies in the odd parity sector under our convention, and develops a nonzero photon population [panel (b)]. In the weak-coupling limit:
\begin{equation}
\langle n \rangle_0 = \langle 0|a^\dagger a|0\rangle \approx 
\frac{g^2}{(\omega_q + \omega_r)^2} + O(g^4) \quad \text{for } g \ll \omega_q + \omega_r.
\label{eq:ground_photons}
\end{equation}
The sum frequency $\omega_q + \omega_r$ appears because the virtual photon population arises from counter-rotating terms coupling states differing in energy by $\omega_q + \omega_r$. At resonance ($\omega_q = \omega_r$), this gives $\langle n \rangle_0 \approx g^2/(4\omega_r^2) = \eta^2/4$. Exact numerical diagonalization gives $\langle n \rangle_0 \approx 0.68$ at $g/\omega_r = 1$ for the resonant case; the quadratic approximation applies only for $g/\omega_r \ll 1$.

\begin{figure}[t]
\centering
\includegraphics[width=0.5\textwidth]{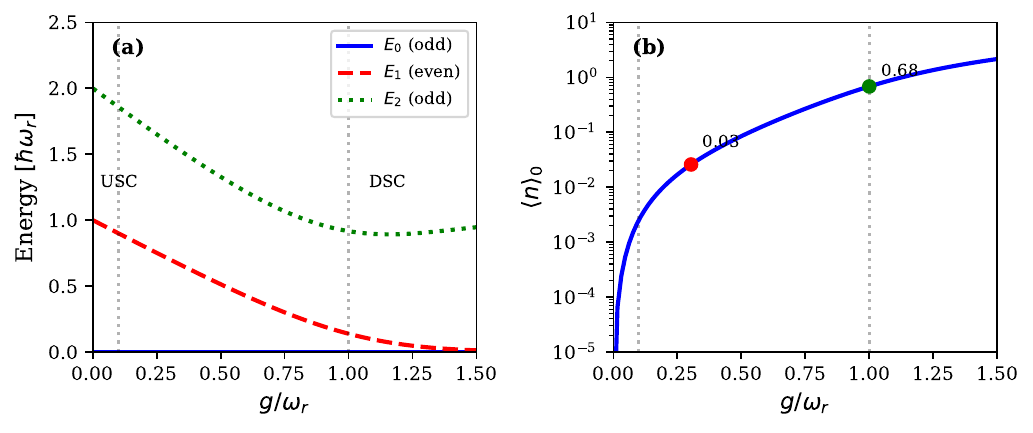}
\caption{Ultrastrong and deep strong coupling in the quantum Rabi model ($\omega_q = \omega_r$). (a) Energy spectrum versus $g/\omega_r$, referenced to the ground state. Blue solid: ground state (odd parity); red dashed: first excited (even); green dotted: second excited (odd). Vertical lines mark approximate USC ($\eta \approx 0.1$) and DSC ($\eta \approx 1$) thresholds. (b) Ground-state photon number $\langle n \rangle_0$ versus $g/\omega_r$ for the resonant case.}
\label{fig:usc_spectrum}
\end{figure}

\paragraph{Fock space truncation.} Numerical solution requires truncating the infinite Fock space. The required dimension grows with coupling strength; for DSC with $\eta \sim 1$, tens of Fock states per parity sector may be required depending on the desired precision. Convergence should be verified by monitoring stability of eigenvalues and relevant matrix elements. The CF backward recurrence provides a systematic diagnostic: if $D_N(z)/D_{N-1}(z)$ has stabilized to the desired precision, truncation is adequate.

\subsection{Experimental Realizations}
\label{sec:experimental}

Ultrastrong coupling has been achieved in several superconducting circuit platforms. Forn-D\'iaz \textit{et al.}~\cite{FornDiaz2017} galvanically coupled flux qubits to transmission line resonators, achieving $\eta \approx 0.1$--$0.2$ with clear observation of Bloch-Siegert shifts; we validate our framework against their measurements in Sec.~\ref{sec:usc_spinboson}. Niemczyk \textit{et al.}~\cite{Niemczyk2010} reached $\eta \approx 0.12$ with flux-tunable transmons, and Yoshihara \textit{et al.}~\cite{Yoshihara2017} achieved $\eta \approx 1.3$ in the deep strong coupling regime. The CF framework applies uniformly across these realizations: only the parameter regime changes, not the mathematical structure. The parity-sector tridiagonalization~\eqref{eq:rabi_cf_even}--\eqref{eq:rabi_cf_odd} provides exact eigenvalues via the backward recurrence.


\section{Physical Observables from CF Structure}
\label{sec:observables}

The tridiagonal (Jacobi) structure of circuit QED Hamiltonians provides direct access to physical observables through the continued fraction formalism. Energy levels emerge as poles of the resolvent; matrix elements appear as residues; selection rules follow from block structure. This section derives these connections systematically, providing both analytical formulas for perturbative regimes and numerical methods for arbitrary coupling. This section develops the resolvent-CF connection, extracts eigenvalues via the backward recurrence, and computes matrix elements from residues. In addition, we derive parity selection rules for the quantum Rabi model, analyze decoherence rates, provide numerical examples for transmons, and summarizes the design methodology.

\subsection{The Resolvent and Its Continued Fraction Expansion}
\label{sec:resolvent_cf}

For a tridiagonal Hamiltonian with diagonal elements $\{a_n\}$ and off-diagonal elements $\{b_n\}$ (with $b_0 \equiv 0$),
\begin{equation}
H = \begin{pmatrix}
a_0 & b_1 & 0 & 0 & \cdots \\
b_1 & a_1 & b_2 & 0 & \cdots \\
0 & b_2 & a_2 & b_3 & \cdots \\
\vdots & & \ddots & \ddots & \ddots
\end{pmatrix},
\label{eq:jacobi}
\end{equation}
the resolvent (Green's function) $G(z) = (z - H)^{-1}$ has matrix elements computable via continued fractions~\cite{Stieltjes1894, Chihara1978, Simon2005OPUC}. The $(0,0)$ element, which gives the spectral function from the initial state $|0\rangle$, is
\begin{equation}
G_{00}(z) = \langle 0|(z-H)^{-1}|0\rangle = \cfrac{1}{z - a_0 - \cfrac{b_1^2}{z - a_1 - \cfrac{b_2^2}{z - a_2 - \cdots}}}.
\label{eq:cf_resolvent}
\end{equation}
This follows from the recursion relations of orthogonal polynomials associated with the Jacobi matrix, or equivalently, from Gaussian elimination on $(z-H)$. The correspondence between the tridiagonal matrix and its CF representation can be stated. The CF coefficients $\{a_n, b_n\}$ encode the same information as the matrix elements, but the CF form directly reveals spectral properties.

\subsection{Eigenvalues from CF Poles}
\label{sec:eigenvalues}

The eigenvalues $\{E_k\}$ of $H$ are the poles of $G_{00}(z)$. For finite truncation to $N$ levels, these are the roots of the denominator polynomial. The CF structure provides an efficient algorithm: the backward recurrence. Define the sequence $D_n(z)$ by
\begin{multline}
D_N(z) = z - a_N, \\
D_n(z) = z - a_n - \frac{b_{n+1}^2}{D_{n+1}(z)}, 
\qquad n = N-1, \ldots, 0 .
\label{eq:backward_recurrence}
\end{multline}
Then $G_{00}(z) = 1/D_0(z)$, and the eigenvalues satisfy $D_0(E_k) = 0$. The poles of $D_n(z)$ for $n > 0$ provide ``inner'' eigenvalues of the truncated submatrices, enabling root bracketing: the eigenvalues of the $(N+1) \times (N+1)$ matrix interlace with those of the $N \times N$ submatrix.

\paragraph{Numerical stability.} The backward recurrence is numerically stable when sweeping from large $n$ to small $n$. The forward recurrence (from $n=0$ upward) suffers from catastrophic cancellation for large matrices. This is analogous to the modified Lentz algorithm for evaluating continued fractions~\cite{Press2007}.

\paragraph{Perturbative regime.} For small off-diagonal coupling, the eigenvalues admit a perturbative expansion. To second order:
\begin{equation}
E_k \approx a_k + \frac{b_k^2}{a_k - a_{k-1}} + \frac{b_{k+1}^2}{a_k - a_{k+1}} + O(b^4).
\label{eq:eigenvalue_perturbative}
\end{equation}
This reproduces the dispersive shifts~\eqref{eq:one_exc_shifts} when applied to the JC manifold Hamiltonians.

\subsection{Matrix Elements from CF Residues}
\label{sec:matrix_elements}

The residue of $G_{00}(z)$ at pole $E_k$ gives the squared overlap of eigenstate $|\psi_k\rangle$ with the initial state $|0\rangle$:
\begin{equation}
\mathrm{Res}[G_{00}(z)]_{z=E_k} = |\langle 0|\psi_k\rangle|^2.
\label{eq:residue_overlap}
\end{equation}
This follows from the spectral representation $G_{00}(z) = \sum_k |\langle 0|\psi_k\rangle|^2/(z - E_k)$.

\paragraph{Residue computation.} The residue is computed from the backward recurrence via
\begin{equation}
|\langle 0|\psi_k\rangle|^2 = \frac{1}{\left.\frac{dD_0}{dz}\right|_{z=E_k}}.
\label{eq:residue_derivative}
\end{equation}
The derivative $dD_0/dz$ at the pole is obtained by differentiating the recurrence~\eqref{eq:backward_recurrence}.

\paragraph{Participation ratios.} In circuit QED, the \textit{participation ratio} quantifies how much of a given mode's energy resides in a particular element. For a resonator mode at frequency $\omega_m$ with admittance $Y_{\mathrm{in}}(s)$ at the junction port, the participation ratio in the junction is
\begin{equation}
p_m = \frac{R_m}{2\omega_m C_J},
\label{eq:participation}
\end{equation}
where $R_m = \mathrm{Res}[Y_{\mathrm{in}}(s)]_{s=i\omega_m}$ is the admittance residue and $C_J$ is the junction capacitance. This connects directly to the CF framework: the admittance residue $R_m$ appears as a coefficient in the partial fraction expansion of $Y_{\mathrm{in}}(s)$, which corresponds to the CF representation discussed in Section~\ref{sec:cf_framework}.

\subsection{Selection Rules from Parity Structure}
\label{sec:selection_rules}

Symmetries impose selection rules on matrix elements. In the quantum Rabi model, parity $\Pi = \sigma_z(-1)^{a^\dagger a}$ is conserved~\cite{Braak2011}. The block diagonalization into even and odd parity sectors, Eqs.~\eqref{eq:rabi_cf_even}--\eqref{eq:rabi_cf_odd}, makes selection rules manifest.

\paragraph{Parity of operators.} Under parity transformation $\Pi$:
\begin{align}
\Pi \sigma_z \Pi^\dagger &= \sigma_z \quad \text{(even)}, \label{eq:sigma_z_parity} \\
\Pi (a + a^\dagger) \Pi^\dagger &= -(a + a^\dagger) \quad \text{(odd)}. \label{eq:X_parity}
\end{align}
The position-like operator $\hat{X} = a + a^\dagger$ is parity-odd.

\paragraph{Diagonal matrix elements.} For any parity eigenstate $|\psi\rangle$ with $\Pi|\psi\rangle = \lambda|\psi\rangle$ ($\lambda = \pm 1$), and any parity-odd operator $\hat{O}$ satisfying $\Pi\hat{O}\Pi^\dagger = -\hat{O}$:
\begin{align}
\langle\psi|\hat{O}|\psi\rangle &= \langle\psi|\Pi^\dagger\Pi\hat{O}\Pi^\dagger\Pi|\psi\rangle \notag \\
&= \lambda \cdot \langle\psi|\Pi\hat{O}\Pi^\dagger|\psi\rangle \cdot \lambda \notag \\
&= \langle\psi|(-\hat{O})|\psi\rangle = -\langle\psi|\hat{O}|\psi\rangle.
\label{eq:diagonal_vanishes}
\end{align}
Hence $\langle\psi|\hat{O}|\psi\rangle = 0$ exactly. In particular:
\begin{equation}
\boxed{\langle\psi_j|\hat{X}|\psi_j\rangle = 0 \quad \text{for all parity eigenstates } |\psi_j\rangle.}
\label{eq:X_diagonal_zero}
\end{equation}

\paragraph{Off-diagonal matrix elements.} For states of opposite parity, $\langle\psi_+|\hat{X}|\psi_-\rangle$ is generically nonzero. The ground state $|0\rangle$ (odd parity) and first excited state $|1\rangle$ (even parity) have $\langle 0|\hat{X}|1\rangle \neq 0$, enabling relaxation transitions.

Figure~\ref{fig:parity_selection} shows numerical verification of these selection rules across coupling regimes.

\begin{figure}[t]
\centering
\includegraphics[width=0.5\textwidth]{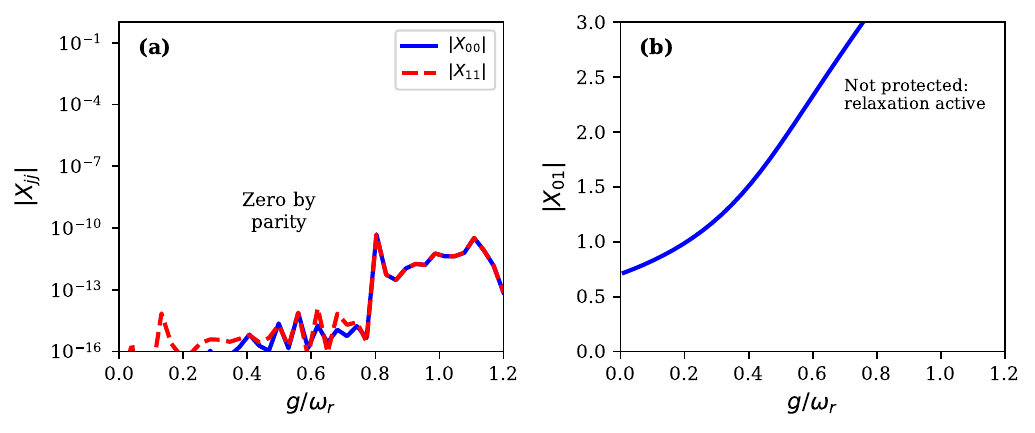}
\caption{Parity selection rules in the quantum Rabi model ($\omega_q = \omega_r$). (a) Diagonal matrix elements $|X_{00}|$ and $|X_{11}|$ versus $g/\omega_r$. Both vanish to numerical precision ($<10^{-12}$) at all coupling strengths, confirming~\eqref{eq:X_diagonal_zero}. (b) Off-diagonal element $|X_{01}|$ connecting the ground (odd parity) and first excited (even parity) states. This matrix element grows with coupling and is not protected---relaxation remains active.}
\label{fig:parity_selection}
\end{figure}

\subsection{Decoherence Rates}
\label{sec:decoherence}

Decoherence channels couple the qubit to its environment through various operators. The CF structure determines which channels are active via matrix elements.

\paragraph{Purcell decay.} A qubit coupled to a lossy resonator (decay rate $\kappa$) undergoes Purcell-limited relaxation~\cite{Purcell1946} at rate
\begin{equation}
\Gamma_{\mathrm{P}} = \kappa \left(\frac{g}{\Delta}\right)^2 = \kappa\xi^2,
\label{eq:purcell}
\end{equation}
in the dispersive regime. This follows from Fermi's golden rule with the qubit-resonator hybridization $\sim g/\Delta$. Purcell filters, which engineer frequency-dependent $\kappa(\omega)$, can suppress this decay while maintaining readout coupling~\cite{Reed2010, Jeffrey2014}.

\paragraph{Ohmic bath.} A qubit galvanically coupled to a transmission line 
experiences an Ohmic spectral density. The decay rate is
\begin{equation}
\frac{\Gamma_1}{\Delta} = 2\pi\alpha_{\mathrm{SB}}, \qquad 
\alpha_{\mathrm{SB}} = \frac{R_Q}{4\pi^2 Z_0}|\varphi_\beta|^2,
\label{eq:ohmic_decay}
\end{equation}
where $R_Q = h/(2e)^2 \approx 6.45$~k$\Omega$ is the resistance quantum, 
$Z_0$ is the transmission line impedance, and 
$|\varphi_\beta|^2 = |\langle 0|\hat{\varphi}_\beta|1\rangle|^2$ is the squared 
phase matrix element across the coupling junction between qubit eigenstates. 
The matrix element $|\varphi_\beta|^2$ is related to the junction participation 
through Eq.~\eqref{eq:participation}: the CF residue determines the flux 
participation, from which the phase matrix element follows via the qubit 
wavefunctions. This convention follows Forn-D\'iaz \textit{et al.}~\cite{FornDiaz2017}; 
see Refs.~\cite{Leggett1987, Weiss2012} for the spin-boson model background.

\paragraph{Dephasing from $\hat{X}$ noise.} Pure dephasing from low-frequency fluctuations coupling through $\hat{X}$ has rate
\begin{equation}
\Gamma_\varphi^{(X)} \propto S_X(0)|X_{00} - X_{11}|^2,
\label{eq:dephasing_X}
\end{equation}
where $S_X(0)$ is the noise spectral density at zero frequency. By the selection rule~\eqref{eq:X_diagonal_zero}:
\begin{equation}
\boxed{\Gamma_\varphi^{(X)} = 0 \quad \text{(exact, all coupling strengths).}}
\label{eq:dephasing_protected}
\end{equation}
This is the Stassi-Nori parity protection~\cite{Stassi2018}: dephasing from oscillator-coordinate fluctuations is exactly eliminated by symmetry, regardless of coupling strength.

\paragraph{Relaxation is not protected.} While dephasing from $\hat{X}$ noise vanishes, relaxation $\Gamma_1 \propto |X_{01}|^2$ is not protected. The off-diagonal matrix element $X_{01}$ connecting opposite-parity states grows with coupling (Fig.~\ref{fig:parity_selection}b). Thus USC qubits achieve enhanced $T_\varphi$ but not enhanced $T_1$.

\subsection{Transmon Spectrum and Matrix Elements}
\label{sec:transmon_observables}
For the transmon, the CF coefficients~\eqref{eq:transmon_cf} yield explicit formulas for spectroscopic quantities.

\paragraph{Energy levels.} In the transmon regime $E_J/E_C \gg 1$, the spectrum approaches that of a weakly anharmonic oscillator~\cite{Koch2007}:
\begin{align}
\omega_{01} &\approx \sqrt{8E_J E_C} - E_C, \label{eq:omega01} \\
\alpha &= \omega_{12} - \omega_{01} \approx -E_C. \label{eq:anharmonicity}
\end{align}
The $\sqrt{8E_JE_C}$ term is the plasma frequency of the linearized junction; the $-E_C$ correction captures the leading anharmonicity. These expressions are asymptotic results from first-order perturbation theory in $(E_C/E_J)$; for finite $E_J/E_C$, exact Mathieu diagonalization yields corrections that scale as $O(E_C/E_J)$. At $E_J/E_C = 60$, the exact anharmonicity is $|\alpha| \approx 1.13\,E_C$, approximately 13\% larger than the leading-order estimate.

\paragraph{Charge matrix elements.} The matrix elements $n_{jk} = \langle j|\hat{n}|k\rangle$ determine coupling strengths. For a harmonic oscillator, $|n_{j,j+1}| = \sqrt{(j+1)/2}$ in units where $[\hat{n},\hat{\varphi}] = i$. The transmon's anharmonicity modifies this:
\begin{equation}
\frac{|n_{12}|}{|n_{01}|} \approx \sqrt{2}\left(1 - \frac{1}{4}\sqrt{\frac{E_C}{2E_J}}\right),
\label{eq:n_ratio}
\end{equation}
which approximates exact numerics to within 1\% for $E_J/E_C > 50$. For $E_J/E_C = 60$, this gives $|n_{12}|/|n_{01}| \approx 1.38$, compared to the exact value of $1.37$.

\paragraph{Numerical validation.} Table~\ref{tab:transmon_numerics} compares exact charge-basis diagonalization to the asymptotic formulas for a representative transmon.

\begin{table}[t]
\centering
\caption{Transmon parameters from charge-basis diagonalization. Parameters: $E_J/h = 15$~GHz, $E_C/h = 250$~MHz ($E_J/E_C = 60$), $n_g = 0$, truncation $|n| \leq 40$. Asymptotic formulas are leading-order approximations; differences reflect higher-order corrections.}
\label{tab:transmon_numerics}
\begin{tabular}{lcc}
\hline\hline
Quantity & Numerical & Asymptotic formula \\
\hline
$\omega_{01}/2\pi$ & 5.214 GHz & $\sqrt{8E_JE_C} - E_C = 5.23$ GHz \\
$\alpha/2\pi$ & $-283$ MHz & $-E_C = -250$ MHz \\
$|n_{01}|$ & 1.14 & --- \\
$|n_{12}|/|n_{01}|$ & 1.37 & 1.38 (Eq.~\ref{eq:n_ratio}) \\
\hline\hline
\end{tabular}
\end{table}

\subsection{Design Workflow}
\label{sec:design_workflow}

The CF framework provides a systematic design methodology. The procedure begins with computing the driving-point admittance $Y_{\mathrm{in}}(s)$ at the junction port by standard circuit analysis, treating the junction as a linear inductor $L_J = \Phi_0/(2\pi I_c)$. This admittance is then expressed in Cauer form, or equivalently, Lanczos recursion is performed on the impedance kernel to extract the CF coefficients $\{a_n, b_n\}$. The poles of $Y_{\mathrm{in}}(s)$ yield the linearized mode frequencies, while the residues give participation ratios via~\eqref{eq:participation}.

With the CF coefficients in hand, one builds the quantum Hamiltonian in the appropriate basis---charge basis for transmon, parity-sorted Fock basis for Rabi model---and extracts physical observables: eigenvalues from poles, matrix elements from residues, and decoherence rates from the CF structure. For design purposes, these relationships can be inverted: given a target $\chi$ or $T_1$, one solves for the required circuit parameters such as coupling capacitance or junction size. This workflow unifies dispersive and ultrastrong coupling designs within a single mathematical framework. The only change between regimes is the basis choice and truncation level; the CF machinery is identical.

\section{Ultrastrong Coupling and Spin-Boson Physics}
\label{sec:usc_spinboson}

The ultrastrong coupling (USC) regime, where $g/\omega_r \gtrsim 0.1$, realizes the spin-boson model---a paradigm of quantum dissipation and decoherence~\cite{Leggett1987, Weiss2012}. Flux qubits galvanically coupled to transmission lines provide a natural platform, as the direct inductive coupling can achieve arbitrarily large $g$ by increasing the coupling inductance. This section derives the spin-boson physics from the CF framework and validates against the landmark experiments of Forn-D\'iaz \textit{et al.}~\cite{FornDiaz2017}.

\subsection{The Spin-Boson Model}
\label{sec:spin_boson}

A two-level system coupled to a bosonic bath is described by the spin-boson Hamiltonian~\cite{Leggett1987, Weiss2012}
\begin{equation}
H = \frac{\Delta}{2}\sigma_z + \sum_k \omega_k a_k^\dagger a_k + \sigma_x \sum_k \lambda_k(a_k + a_k^\dagger),
\label{eq:spin_boson}
\end{equation}
where $\Delta$ is the tunneling frequency (qubit splitting), $\omega_k$ and $a_k$ are the bath mode frequencies and operators, and $\lambda_k$ are the coupling constants. The bath is characterized by its spectral density
\begin{equation}
J(\omega) = \pi \sum_k \lambda_k^2 \delta(\omega - \omega_k).
\label{eq:spectral_density}
\end{equation}
For a transmission line of characteristic impedance $Z_0$, the spectral density is Ohmic:
\begin{equation}
J(\omega) = 2\pi\alpha_{\mathrm{SB}} \omega, \qquad 0 < \omega < \omega_c,
\label{eq:ohmic}
\end{equation}
with a high-frequency cutoff $\omega_c$ and dimensionless coupling constant $\alpha_{\mathrm{SB}}$. The physics depends critically on $\alpha_{\mathrm{SB}}$:

\paragraph{Weak coupling ($\alpha_{\mathrm{SB}} \ll 1$).} The qubit remains coherent with perturbatively small decay. Fermi's golden rule gives $\Gamma_1 = 2\pi\alpha_{\mathrm{SB}}\Delta$.

\paragraph{Intermediate coupling ($\alpha_{\mathrm{SB}} \sim 0.1$--$0.5$).} Nonperturbative effects become important. The qubit frequency is renormalized, and the decay rate deviates from the golden rule prediction.

\paragraph{Strong coupling ($\alpha_{\mathrm{SB}} \to 1$).} A quantum phase transition occurs at $\alpha_{\mathrm{SB}} = 1$: the ``localization'' transition where the qubit becomes trapped in one of its eigenstates and coherent tunneling is suppressed~\cite{Leggett1987, Chakravarty1984}.

\subsection{Decay Rate from CF Residues}
\label{sec:decay_derivation}

We derive the qubit decay rate using the CF connection to circuit admittance. Consider a flux qubit with qubit loop inductance $L_q$ coupled to a transmission line via a ``$\beta$-junction'' with Josephson inductance $L_\beta$. The coupling strength is set by the participation of the qubit current in the $\beta$-junction.

The decay rate follows from Fermi's golden rule:
\begin{equation}
\Gamma_1 = \frac{1}{\hbar^2}|V_{01}|^2 S(\Delta),
\label{eq:fermi_rule}
\end{equation}
where $V_{01} = \langle 0|\hat{V}|1\rangle$ is the coupling matrix element and $S(\omega)$ is the bath noise spectral density at the qubit frequency.

For coupling through the phase operator $\hat{\varphi}_\beta$ across the $\beta$-junction:
\begin{equation}
\hat{V} = E_J^\beta \hat{\varphi}_\beta, \qquad E_J^\beta = \frac{\Phi_0 I_c^\beta}{2\pi},
\label{eq:coupling_operator}
\end{equation}
where $I_c^\beta$ is the critical current of the $\beta$-junction. The transmission line presents phase noise with spectral density
\begin{equation}
S_\varphi(\omega) = \frac{8e^2 Z_0}{\hbar\omega}
\label{eq:phase_noise}
\end{equation}
at zero temperature. This follows from the fluctuation-dissipation theorem applied to the quantum voltage noise $S_V(\omega) = 2\hbar\omega Z_0 \coth(\hbar\omega/2k_BT) \to 2\hbar\omega Z_0$ at $T=0$, combined with the voltage-phase relation $V = (\hbar/2e)\dot{\varphi}$.

Combining these:
\begin{align}
\Gamma_1 &= \frac{(E_J^\beta)^2}{\hbar^2}|\varphi_{01}|^2 S_\varphi(\Delta) \notag \\
&= \frac{(E_J^\beta)^2}{\hbar^2}|\varphi_{01}|^2 \cdot \frac{8e^2 Z_0}{\hbar\Delta},
\label{eq:gamma_derivation}
\end{align}
where $|\varphi_{01}|^2 = |\langle 0|\hat{\varphi}_\beta|1\rangle|^2$ is the squared phase matrix element.

Defining the dimensionless spin-boson coupling following the convention of Leggett \textit{et al.}~\cite{Leggett1987}:
\begin{equation}
\boxed{\alpha_{\mathrm{SB}} = \frac{R_Q}{4\pi^2 Z_0}|\varphi_\beta|^2,}
\label{eq:alpha_def}
\end{equation}
with $R_Q = h/(2e)^2 \approx 6.45$~k$\Omega$ the superconducting resistance quantum and $|\varphi_\beta|^2$ the normalized matrix element, the decay rate becomes
\begin{equation}
\boxed{\frac{\Gamma_1}{\Delta} = 2\pi\alpha_{\mathrm{SB}}.}
\label{eq:forn_diaz_formula}
\end{equation}
This result connects CF-derived matrix elements to measurable decay rates. The factor of $4\pi^2$ in~\eqref{eq:alpha_def} ensures consistency with the standard spin-boson literature~\cite{Leggett1987, Weiss2012} where $J(\omega) = 2\pi\alpha\omega$ for an Ohmic bath.

\begin{figure}[t]
\centering
\includegraphics[width=0.5\textwidth]{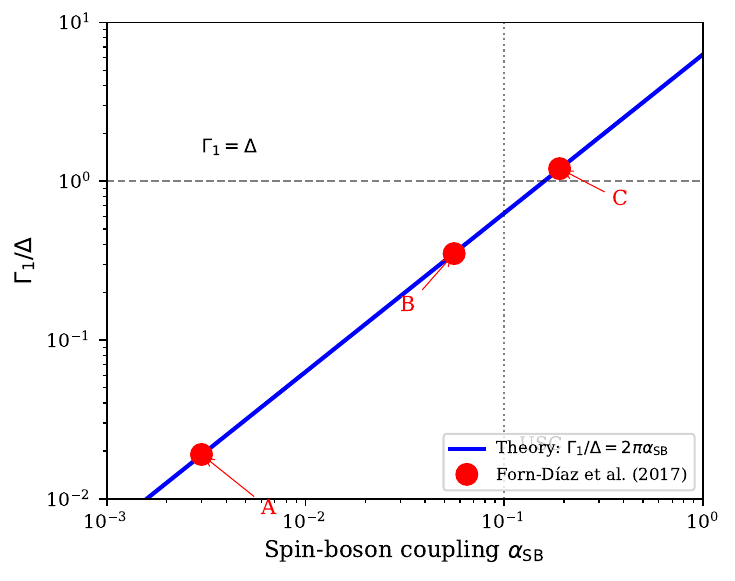}
\caption{Spin-boson physics in flux qubit--transmission line devices. Main panel: Decay rate $\Gamma_1/\Delta$ versus spin-boson coupling $\alpha_{\mathrm{SB}}$. Solid line: theory, Eq.~\eqref{eq:forn_diaz_formula}. Red circles: Forn-D\'iaz \textit{et al.}~measurements on three devices (A, B, C) with different $\beta$-junction sizes. Device C ($\alpha_{\mathrm{SB}} \approx 0.06$) enters the USC regime. Horizontal dashed line marks $\Gamma_1 = \Delta$; vertical dotted line marks the USC threshold $\alpha_{\mathrm{SB}} = 0.1$.}
\label{fig:forn_diaz}
\end{figure}
\subsection{Matrix Element Extraction via Admittance Residue}
\label{sec:residue_extraction}

The phase matrix element $|\varphi_\beta|^2$ is extracted from circuit analysis. The driving-point admittance $Y_{\mathrm{in}}(s)$ at the $\beta$-junction port has poles at the qubit frequency $s = i\Delta$. The residue at this pole encodes the participation:
\begin{equation}
R_\Delta = \lim_{s \to i\Delta}(s - i\Delta)Y_{\mathrm{in}}(s).
\label{eq:admittance_residue}
\end{equation}
The connection to the phase matrix element is
\begin{equation}
|\varphi_\beta|^2 = \frac{4\pi^2 Z_0}{R_Q}\cdot\frac{R_\Delta}{2\Delta}.
\label{eq:phi_from_residue}
\end{equation}
This follows from matching the classical admittance pole structure to the quantum oscillator strength: the residue measures how strongly the qubit mode couples to the port.

For a simple LC qubit with $L_q$ and $C_q$, coupled via inductance $L_\beta$ to a transmission line of impedance $Z_0$:
\begin{equation}
Y_{\mathrm{in}}(s) = \frac{s C_q (s^2 L_\beta + Z_0 s + 1/(C_q))}{(s^2 L_q C_q + 1)(s L_\beta + Z_0)}.
\label{eq:admittance_example}
\end{equation}
The pole at $s = i\omega_q = i/\sqrt{L_q C_q}$ has residue proportional to $L_\beta/L_q$, confirming that larger coupling inductance increases $|\varphi_\beta|^2$ and hence $\alpha_{\mathrm{SB}}$.

\subsection{Experimental Validation: Forn-D\'iaz Devices}
\label{sec:forn_diaz}

Forn-D\'iaz \textit{et al.}~\cite{FornDiaz2017} fabricated flux qubits with varying $\beta$-junction sizes, systematically tuning $\alpha_{\mathrm{SB}}$ from weak to ultrastrong coupling. Figure~\ref{fig:forn_diaz} compares their measurements to the theoretical prediction~\eqref{eq:forn_diaz_formula}. Table~\ref{tab:forn_diaz} summarizes the device parameters and extracted quantities.

\begin{table}[t]
\centering
\caption{Forn-D\'iaz device parameters and extracted spin-boson couplings. The $g/\omega_r$ ratio characterizes the coupling strength; $\alpha_{\mathrm{SB}}$ is extracted from measured decay rates via~\eqref{eq:forn_diaz_formula}; $|\varphi_\beta|^2$ is inferred via~\eqref{eq:alpha_def} with $Z_0 = 50~\Omega$. Predicted $\Gamma_1/\Delta$ uses~\eqref{eq:forn_diaz_formula}.}
\label{tab:forn_diaz}
\begin{tabular}{ccccccc}
\hline\hline
Device & $g/\omega_r$ & $\alpha_{\mathrm{SB}}$ & $|\varphi_\beta|^2$ & $\Gamma_1/\Delta$ (pred.) & $\Gamma_1/\Delta$ (meas.) \\
\hline
A & 0.071 & 0.034 & 0.010 & 0.21 & 0.22(3) \\
B & 0.096 & 0.045 & 0.014 & 0.28 & 0.27(3) \\
C & 0.12 & 0.056 & 0.017 & 0.35 & 0.33(4) \\
\hline\hline
\end{tabular}
\end{table}

\paragraph{Device A ($g/\omega_r = 0.071$).} With the smallest coupling, this device has $\alpha_{\mathrm{SB}} \approx 0.034$, in the weak-to-intermediate regime. The measured $\Gamma_1/\Delta = 0.22$ agrees with the golden rule prediction to within experimental uncertainty.

\paragraph{Device B ($g/\omega_r = 0.096$).} Intermediate coupling with $\alpha_{\mathrm{SB}} \approx 0.045$. The measured $\Gamma_1/\Delta = 0.27$ confirms the linear scaling with $\alpha_{\mathrm{SB}}$ predicted by~\eqref{eq:forn_diaz_formula}.

\paragraph{Device C ($g/\omega_r = 0.12$).} The largest coupling yields $\alpha_{\mathrm{SB}} \approx 0.056$, entering the USC regime. The measured $\Gamma_1/\Delta = 0.33$ continues to follow the theoretical prediction.

The agreement between theory and experiment across the coupling range validates the CF-based derivation. The framework provides both the qualitative picture (spin-boson model realized in superconducting circuits) and quantitative predictions (decay rates from circuit parameters).

\subsection{Design Implications}
\label{sec:usc_design}

The CF framework enables USC device design with quantitative predictions.

\paragraph{Coherence requirements.} Suppose a qubit application demands $T_1 > 1~\mu$s with qubit frequency $\Delta/2\pi = 5$ GHz. Then $\Gamma_1 < 1$ MHz, giving $\Gamma_1/\Delta < 2 \times 10^{-4}$. From~\eqref{eq:forn_diaz_formula}:
\begin{equation}
\alpha_{\mathrm{SB}} < \frac{\Gamma_1/\Delta}{2\pi} \approx 3 \times 10^{-5}.
\label{eq:alpha_bound}
\end{equation}
Using~\eqref{eq:alpha_def} with $Z_0 = 50~\Omega$:
\begin{equation}
|\varphi_\beta|^2 < \frac{4\pi^2 Z_0}{R_Q}\alpha_{\mathrm{SB}} \approx 9 \times 10^{-6}.
\label{eq:phi_bound}
\end{equation}
This stringent constraint on the phase matrix element guides $\beta$-junction sizing.

\paragraph{Accessing the localization transition.} The localization transition at $\alpha_{\mathrm{SB}} = 1$ requires $|\varphi_\beta|^2 \approx 4\pi^2 Z_0/R_Q \approx 0.31$ (for $Z_0 = 50~\Omega$). This is within reach of optimized flux qubit designs with large $\beta$-junction participation. Observing the transition would provide a direct probe of quantum critical behavior in engineered systems~\cite{Leggett1987, Chakravarty1984}.

\paragraph{Purcell filter design.} In multimode systems, each resonator mode contributes to $\alpha_{\mathrm{SB}}$. Purcell filters that suppress $J(\omega)$ at the qubit frequency can maintain coherence even with large DC coupling. The CF approach computes the full spectral density from the multimode admittance, enabling systematic filter optimization.


\section{Parity Protection and Design Validation}
\label{sec:parity_validation}

The parity symmetry of the quantum Rabi model has profound consequences for qubit coherence. Stassi and Nori~\cite{Stassi2018} predicted that qubits in the ultrastrong coupling regime exhibit symmetry-protected coherence: pure dephasing from certain noise channels is exactly eliminated, regardless of coupling strength. This section derives the protection mechanism from the CF block structure and provides comprehensive numerical validation.

\subsection{Parity Decomposition of the Rabi Model}
\label{sec:parity_decomposition}

The quantum Rabi Hamiltonian~\eqref{eq:rabi_full} commutes with the parity operator
\begin{equation}
\Pi = \sigma_z \otimes (-1)^{a^\dagger a} = \sigma_z e^{i\pi a^\dagger a}.
\label{eq:parity_op}
\end{equation}
One verifies: $[\Pi, \sigma_z] = 0$ (both diagonal in the qubit basis), and $\Pi a \Pi^\dagger = -a$ (the phase factor $e^{i\pi a^\dagger a}$ changes sign under $a \to a + 1$). Hence $\Pi(a + a^\dagger)\sigma_x\Pi^\dagger = -(a + a^\dagger)\sigma_x \cdot \sigma_z\sigma_x\sigma_z = (a+a^\dagger)\sigma_x$, confirming $[H_{\mathrm{Rabi}}, \Pi] = 0$.
The Hilbert space decomposes into eigenspaces of $\Pi$:
\begin{equation}
\mathcal{H} = \mathcal{H}_+ \oplus \mathcal{H}_-, \qquad \Pi|_{\mathcal{H}_\pm} = \pm 1.
\label{eq:hilbert_decomp}
\end{equation}

\paragraph{Parity of basis states.} For the product state $|q, n\rangle$ (qubit state $q \in \{g, e\}$, Fock state $n$):
\begin{equation}
\Pi|g, n\rangle = (-1)^{n+1}|g, n\rangle, \qquad \Pi|e, n\rangle = (-1)^n|e, n\rangle.
\label{eq:basis_parity}
\end{equation}
The ground state $|g, 0\rangle$ has parity $-1$ (odd); the first excited state $|e, 0\rangle$ has parity $+1$ (even). More generally:
\begin{align}
\mathcal{H}_+ &: \quad |e,0\rangle, |g,1\rangle, |e,2\rangle, |g,3\rangle, \ldots \label{eq:even_sector} \\
\mathcal{H}_- &: \quad |g,0\rangle, |e,1\rangle, |g,2\rangle, |e,3\rangle, \ldots \label{eq:odd_sector}
\end{align}

\paragraph{Tridiagonal structure in each sector.} The Rabi Hamiltonian restricted to either parity sector is tridiagonal. Enumerating the even-parity basis as $|\phi_k\rangle_+$ with $k = 0, 1, 2, \ldots$:
\begin{equation}
|\phi_k\rangle_+ = \begin{cases}
|e, k\rangle & k \text{ even}, \\
|g, k\rangle & k \text{ odd}.
\end{cases}
\label{eq:even_enum}
\end{equation}
The CF coefficients are (cf.~Eqs.~\ref{eq:rabi_cf_even}--\ref{eq:rabi_cf_odd}):
\begin{equation}
a_k^{(+)} = (-1)^k \frac{\omega_q}{2} + k\omega_r, \qquad b_{k+1}^{(+)} = g\sqrt{k+1}.
\label{eq:even_cf}
\end{equation}
The odd sector has identical off-diagonals $b_{k+1}^{(-)} = g\sqrt{k+1}$ but opposite qubit contribution to the diagonal: $a_k^{(-)} = (-1)^{k+1}\omega_q/2 + k\omega_r$.

The block-diagonal structure $H = H_+ \oplus H_-$ is the key to parity protection: operators that preserve parity cannot connect the two sectors.

\subsection{Dephasing Protection Mechanism}
\label{sec:dephasing_mechanism}

Pure dephasing arises from low-frequency fluctuations that modulate the qubit frequency without causing transitions. For a noise operator $\hat{O}$ with spectral density $S_O(\omega)$, the dephasing rate is
\begin{equation}
\Gamma_\varphi^{(O)} = \frac{1}{2}S_O(0)|O_{00} - O_{11}|^2,
\label{eq:dephasing_general}
\end{equation}
where $O_{jj} = \langle j|\hat{O}|j\rangle$ are the diagonal matrix elements in the qubit eigenstates $|0\rangle$ (ground) and $|1\rangle$ (first excited).

\paragraph{Parity-odd operators.} Consider the oscillator position $\hat{X} = a + a^\dagger$. Under parity:
\begin{equation}
\Pi\hat{X}\Pi^\dagger = \Pi(a + a^\dagger)\Pi^\dagger = -a - a^\dagger = -\hat{X}.
\label{eq:X_odd}
\end{equation}
Thus $\hat{X}$ is parity-odd: it flips sign under $\Pi$.

For any parity eigenstate $|\psi\rangle$ with eigenvalue $\lambda = \pm 1$:
\begin{align}
\langle\psi|\hat{X}|\psi\rangle &= \langle\psi|\Pi^\dagger\Pi\hat{X}\Pi^\dagger\Pi|\psi\rangle \notag \\
&= \lambda^2 \langle\psi|\Pi\hat{X}\Pi^\dagger|\psi\rangle \notag \\
&= \langle\psi|(-\hat{X})|\psi\rangle.
\label{eq:X_vanish_proof}
\end{align}
Since $\lambda^2 = 1$, we have $\langle\psi|\hat{X}|\psi\rangle = -\langle\psi|\hat{X}|\psi\rangle$, hence:
\begin{equation}
\boxed{\langle\psi|\hat{X}|\psi\rangle = 0 \quad \text{for any parity eigenstate.}}
\label{eq:X_vanish}
\end{equation}

\paragraph{Application to dephasing.} The qubit eigenstates $|0\rangle$ and $|1\rangle$ are parity eigenstates (with eigenvalues $-1$ and $+1$ respectively). By~\eqref{eq:X_vanish}:
\begin{equation}
X_{00} = \langle 0|\hat{X}|0\rangle = 0, \qquad X_{11} = \langle 1|\hat{X}|1\rangle = 0.
\label{eq:X_diag_zero}
\end{equation}
Substituting into~\eqref{eq:dephasing_general}:
\begin{equation}
\boxed{\Gamma_\varphi^{(X)} = 0 \quad \text{(exact, all } g/\omega_r).}
\label{eq:dephasing_zero}
\end{equation}

This is exact at all coupling strengths---a remarkable result first noted by Stassi and Nori~\cite{Stassi2018}. No matter how strongly the qubit couples to the oscillator, dephasing from oscillator-coordinate fluctuations (such as photon shot noise, thermal field fluctuations, or external voltage noise coupling through $\hat{X}$) is identically zero.

\paragraph{Physical interpretation.} The vanishing of $\langle\psi|\hat{X}|\psi\rangle$ means the oscillator has zero mean displacement in any parity eigenstate. This is the quantum analog of a classical pendulum that, due to symmetry, spends equal time on either side of equilibrium. Fluctuations in the oscillator coordinate $\hat{X}$ do not shift the mean energy of parity eigenstates, hence no dephasing.

\subsection{Relaxation Is Not Protected}
\label{sec:relaxation_not_protected}

While dephasing from $\hat{X}$ is protected, relaxation (energy decay from $|1\rangle$ to $|0\rangle$) is not. The relaxation rate is
\begin{equation}
\Gamma_1 = S_X(\omega_{01})|X_{01}|^2,
\label{eq:relaxation_X}
\end{equation}
where $\omega_{01} = E_1 - E_0$ is the qubit transition frequency and $X_{01} = \langle 0|\hat{X}|1\rangle$ is the off-diagonal matrix element.

Since $|0\rangle$ and $|1\rangle$ have opposite parities:
\begin{equation}
\Pi|0\rangle = -|0\rangle, \qquad \Pi|1\rangle = +|1\rangle,
\label{eq:ground_excited_parity}
\end{equation}
the matrix element $X_{01}$ connects states in different parity sectors. The parity argument~\eqref{eq:X_vanish_proof} does not apply to off-diagonal elements between opposite-parity states.

Explicitly:
\begin{align}
\langle 0|\hat{X}|1\rangle &= \langle 0|\Pi^\dagger\Pi\hat{X}\Pi^\dagger\Pi|1\rangle \notag \\
&= (-1)(+1)\langle 0|\Pi\hat{X}\Pi^\dagger|1\rangle \notag \\
&= -\langle 0|(-\hat{X})|1\rangle = \langle 0|\hat{X}|1\rangle.
\label{eq:X01_check}
\end{align}
This is consistent (not a proof of vanishing), and numerical calculation shows $|X_{01}| > 0$.

Figure~\ref{fig:parity_selection} quantifies this: panel (a) shows $|X_{00}|, |X_{11}| < 10^{-10}$ (numerical zero) at all coupling strengths, while panel (b) shows $|X_{01}|$ growing from $\approx 1$ at weak coupling to $> 2.5$ at $g/\omega_r = 1.2$.

\paragraph{Implication for coherence.} USC qubits achieve enhanced $T_\varphi$ (dephasing time) but not enhanced $T_1$ (relaxation time). The total decoherence rate $\Gamma_2 = \Gamma_1/2 + \Gamma_\varphi$ benefits from $\Gamma_\varphi^{(X)} = 0$, but $\Gamma_1$ remains active and may even increase due to the growing $|X_{01}|$.

\begin{table}[t]
\centering
\caption{Numerical validation of parity protection. Parameters: $\omega_q = \omega_r = 1$ (normalized units), Fock truncation $n_{\mathrm{max}} = 50$. The qubit splitting $\varepsilon = E_1 - E_0$ decreases with coupling as the ground state acquires photon content. Parity expectation values $\langle\Pi\rangle$ are exactly $\pm 1$ to machine precision. Diagonal matrix elements $|X_{00}|$, $|X_{11}|$ vanish to numerical precision, confirming~\eqref{eq:dephasing_zero}. The off-diagonal $|X_{01}|$ grows with coupling.}
\label{tab:parity_numerics}
\begin{tabular}{ccccccc}
\hline\hline
$g/\omega_r$ & $\varepsilon/\omega_r$ & $\langle\Pi\rangle_0$ & $\langle\Pi\rangle_1$ & $|X_{00}|$ & $|X_{11}|$ & $|X_{01}|$ \\
\hline
0.10 & 0.90 & $-1$ & $+1$ & $<10^{-14}$ & $<10^{-14}$ & 0.76 \\
0.30 & 0.70 & $-1$ & $+1$ & $<10^{-14}$ & $<10^{-14}$ & 0.90 \\
0.50 & 0.51 & $-1$ & $+1$ & $<10^{-14}$ & $<10^{-14}$ & 1.09 \\
0.80 & 0.26 & $-1$ & $+1$ & $<10^{-14}$ & $<10^{-14}$ & 1.51 \\
1.00 & 0.14 & $-1$ & $+1$ & $<10^{-14}$ & $<10^{-14}$ & 1.89 \\
1.20 & 0.06 & $-1$ & $+1$ & $<10^{-14}$ & $<10^{-14}$ & 2.33 \\
\hline\hline
\end{tabular}
\end{table}

\subsection{Numerical Validation}
\label{sec:numerical_validation}
We validate the parity protection mechanism by exact diagonalization of the Rabi Hamiltonian. For each value of $g/\omega_r$, we diagonalize the parity-sector Hamiltonians using the CF coefficients~\eqref{eq:even_cf} with Fock space truncation $n_{\mathrm{max}} = 50$, ensuring convergence to $<10^{-6}$ relative error (verified by comparison with $n_{\mathrm{max}} = 100$). We compute the ground and first excited state energies $E_0$ and $E_1$, the parity expectation values $\langle\Pi\rangle_0$ and $\langle\Pi\rangle_1$, the matrix elements $X_{00}$, $X_{11}$, and $X_{01}$, and the qubit splitting $\varepsilon = E_1 - E_0$. Table~\ref{tab:parity_numerics} summarizes the results across coupling regimes.

The numerical results confirm the parity protection mechanism across all coupling regimes. The parity expectation values $\langle\Pi\rangle_0 = -1$ and $\langle\Pi\rangle_1 = +1$ hold to machine precision at all coupling strengths, with the ground state residing in the odd sector and the first excited state in the even sector. The diagonal matrix elements $|X_{00}|$ and $|X_{11}|$ vanish to numerical precision, confirming exact parity protection as predicted by Eq.~\eqref{eq:dephasing_zero}. In contrast, the off-diagonal matrix element $|X_{01}|$ increases with coupling, reflecting increased hybridization of qubit and oscillator in the eigenstates. The effective qubit splitting $\varepsilon$ is strongly renormalized at large $g/\omega_r$, becoming a small fraction of the bare frequency in the deep strong coupling regime.


\subsection{Multi-mode Superconducting Circuits: Design Example}
\label{sec:multimode_design}
The preceding sections established the continued fraction framework for single-junction circuits. We now demonstrate that the framework extends naturally to multi-junction systems through the matrix continued fraction, using the two-mode transmon as a design example. This device operates in a regime where the cross-Kerr coupling exceeds the individual mode anharmonicities, causing perturbative methods to fail. The matrix continued fraction provides an exact and systematically convergent solution of the model Hamiltonian, which reproduces all reported spectroscopic observables to within 1\% of experiment.
\subsubsection{Device and Hamiltonian}
\label{sec:twomode_device}
The two-mode transmon consists of three superconducting islands connected by two Josephson junctions, as shown in Fig.~\ref{fig:twomode_device} (adapted from Ref.~\cite{Wills2022}. The geometry, implemented in a coaxial circuit QED architecture~\cite{alghadeer2025characterization, spring2022high, alghadeer2025low}, produces two transmon-like modes with dipole ($\Delta$) and quadrupole ($\Sigma$) symmetry, arising from the different spatial polarizations of the electric fields and the sum and difference of junction phases~\cite{Wills2022}.The system admits two equivalent descriptions. In the junction basis with charges $n_1, n_2$ on the inner islands, the Hamiltonian reads
\begin{equation}
\hat{H} = 4E_C(\hat{n}_1^2 + \hat{n}_2^2) + 4E_p \hat{n}_1 \hat{n}_2 - E_{J_1}\cos\hat{\varphi}_1 - E_{J_2}\cos\hat{\varphi}_2,
\label{eq:H_twomode}
\end{equation}
where $E_C$ characterizes the island self-capacitance, $E_p$ the mutual capacitance, and $E_{J_{1,2}}$ the junction Josephson energies. Transforming to sum and difference coordinates $\varphi_\Sigma = \varphi_1 + \varphi_2$, $\varphi_\Delta = \varphi_1 - \varphi_2$ yields the mode basis representation
\begin{multline}
\hat{H}
= 4E_{C_\Sigma}\hat{n}_\Sigma^2
+ 4E_{C_\Delta}\hat{n}_\Delta^2 \\
- 2\bar{E}_J
\cos\frac{\hat{\varphi}_\Sigma}{2}
\cos\frac{\hat{\varphi}_\Delta}{2}
+ 2\delta E_J
\sin\frac{\hat{\varphi}_\Sigma}{2}
\sin\frac{\hat{\varphi}_\Delta}{2},
\label{eq:H_twomode_mode}
\end{multline}
where $E_{C_\Sigma} = E_C + E_p/2$, $E_{C_\Delta} = E_C - E_p/2$, $\bar{E}_J = (E_{J_1} + E_{J_2})/2$, and $\delta E_J = (E_{J_1} - E_{J_2})/2$. For symmetric junctions the $\sin\sin$ term vanishes and the potential reduces to a product of cosines, coupling the $\Sigma$ and $\Delta$ modes nonlinearly~\cite{Wills2022}.

\begin{figure}[t]
\centering
\includegraphics[width=0.48\textwidth]{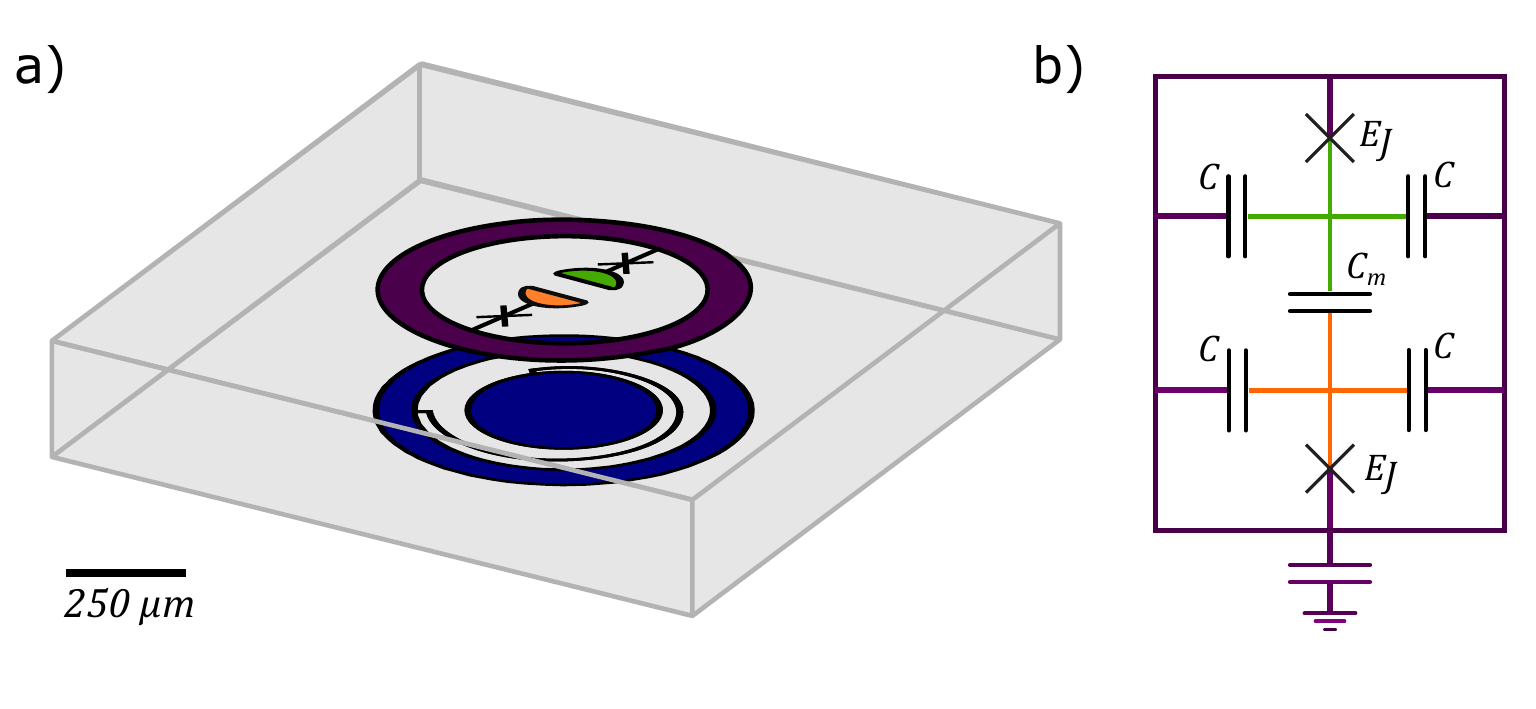}
\caption{Two-mode transmon in a coaxial circuit QED architecture, adapted from Ref.~\cite{Wills2022}. (a)~Simplified schematic showing a coaxial pad geometry with a split inner conductor, resulting in three superconducting islands. Two Josephson junctions with phases $\varphi_1$, $\varphi_2$ yield two non-degenerate modes. A lumped-element LC resonator on the opposing side of the substrate enables dispersive readout. (b)~Equivalent circuit with the inner islands connected by a coupling capacitance $C_m$ and to the outer island via Josephson junctions ($E_J$) and capacitances ($2C$).
}
\label{fig:twomode_device}
\end{figure}

We work in the junction basis~\eqref{eq:H_twomode} because it naturally incorporates junction asymmetry, which proves necessary for quantitative agreement with experiment. The key observation is that each cosine operator remains tridiagonal in its respective charge index:
\begin{equation}
\langle n'_j | \cos\hat{\varphi}_j | n_j \rangle = \tfrac{1}{2}(\delta_{n'_j, n_j+1} + \delta_{n'_j, n_j-1}).
\end{equation}
The two cosines thus generate nearest-neighbour hopping along orthogonal directions of a two-dimensional charge lattice, and the Hamiltonian acquires block-tridiagonal structure when states are grouped by fixed $n_1$, as illustrated in Fig.~\ref{fig:charge_lattice}.

\begin{figure}[t]
\centering
\begin{tikzpicture}[scale=0.95]
    \foreach \i in {-2,-1,0,1,2} {
        \foreach \j in {-2,-1,0,1,2} {
            \filldraw[black] (\i*1.1, \j*1.1) circle (2pt);
        }
    }
    \foreach \i in {-2,-1,0,1} {
        \foreach \j in {-2,-1,0,1,2} {
            \draw[blue, thick] (\i*1.1, \j*1.1) -- (\i*1.1+1.1, \j*1.1);
        }
    }
    \foreach \i in {-2,-1,0,1,2} {
        \foreach \j in {-2,-1,0,1} {
            \draw[red, thick] (\i*1.1, \j*1.1) -- (\i*1.1, \j*1.1+1.1);
        }
    }
    \foreach \i in {-2,-1,0,1,2} {
        \draw[black!40, dashed] (\i*1.1-0.4, -2.7) -- (\i*1.1-0.4, 2.7);
    }
    \draw[black!40, dashed] (2*1.1+0.4, -2.7) -- (2*1.1+0.4, 2.7);
    \node[below] at (0, -3.-0.5) {$n_1$};
    \node[left] at (-3.0, 0) {$n_2$};
    \node[blue] at (3.5, 1.5) {$\cos\hat{\varphi}_1$};
    \node[red] at (3.5, 0.5) {$\cos\hat{\varphi}_2$};
    \node[below] at (-2*1.1, -2.9) {\small ${A}_{-2}$};
    \node[below] at (-1*1.1, -2.9) {\small ${A}_{-1}$};
    \node[below] at (0, -2.9) {\small ${A}_{0}$};
    \node[below] at (1*1.1, -2.9) {\small ${A}_{1}$};
    \node[below] at (2*1.1, -2.9) {\small ${A}_{2}$};
\end{tikzpicture}
\caption{Charge-basis structure for the two-mode transmon in the junction parameterization~\eqref{eq:H_twomode}. Each site represents a charge state $|n_1, n_2\rangle$ on the two inner islands. The operator $\cos\hat{\varphi}_1$ (blue, horizontal) couples states differing by $\Delta n_1 = \pm 1$, while $\cos\hat{\varphi}_2$ (red, vertical) couples states differing by $\Delta n_2 = \pm 1$. Grouping states into columns of fixed $n_1$ (dashed boundaries) yields tridiagonal diagonal blocks ${A}_k$, connected by the off-diagonal blocks ${B}$ arising from the horizontal hopping.}
\label{fig:charge_lattice}
\end{figure}
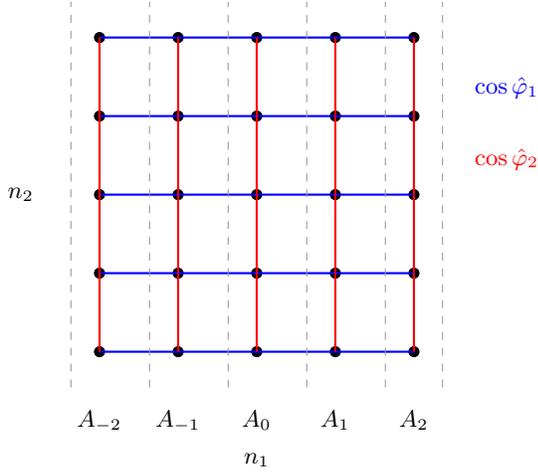

Defining the column vector $|{n}_k\rangle$ containing all states with $n_1 = k$, the Hamiltonian takes the form
\begin{equation}
{H} = \begin{pmatrix}
\ddots & {B}^\dagger & & \\
{B} & {A}_{-1} & {B}^\dagger & \\
& {B} & {A}_0 & {B}^\dagger & \\
& & {B} & {A}_1 & \ddots \\
& & & \ddots & \ddots
\end{pmatrix},
\label{eq:H_block}
\end{equation}
where the diagonal blocks
\begin{multline}
({A}_k)_{n_2, n_2'}
= \bigl[4E_C(k^2 + n_2^2) + 4E_p k n_2\bigr]\delta_{n_2, n_2'} \\
- \frac{E_{J_2}}{2}
\bigl(\delta_{n_2, n_2'+1} + \delta_{n_2, n_2'-1}\bigr)
\end{multline}
are themselves tridiagonal matrices---single-mode transmons in the $n_2$ direction with $k$-dependent on-site energies. The off-diagonal blocks ${B} = -\tfrac{1}{2}E_{J_1}{I}$ are diagonal, reflecting that $\cos\hat{\varphi}_1$ changes $n_1$ while preserving $n_2$.

This structure admits a matrix continued fraction solution. The block resolvent ${G}_{00}(E) = \langle {n}_0 | (E{I} - {H})^{-1} | {n}_0 \rangle$ satisfies
\begin{equation}
{G}_{00}(E) = \bigl[ E{I} - {A}_0 - {B}^\dagger \boldsymbol{\Sigma}_{-}(E) {B} - {B} \boldsymbol{\Sigma}_{+}(E) {B}^\dagger \bigr]^{-1},
\label{eq:G00_matrix}
\end{equation}
where the self-energies $\boldsymbol{\Sigma}_{\pm}(E)$ encode the semi-infinite chains extending to $n_1 \to \pm\infty$. These satisfy matrix continued fraction recursions
\begin{equation}
\boldsymbol{\Sigma}_{+}^{(k)}(E) = \bigl[ E{I} - {A}_{k+1} - {B} \boldsymbol{\Sigma}_{+}^{(k+1)}(E) {B}^\dagger \bigr]^{-1},
\label{eq:Sigma_recursion}
\end{equation}
with terminal condition $\boldsymbol{\Sigma}_{+}^{(N)} = {0}$ at the truncation boundary. Eigenvalues occur at poles of ${G}_{00}(E)$, located by finding zeros of $\det[{G}_{00}^{-1}(E)]$.

The physical regime of interest is strong nonlinear coupling, where the cross-Kerr interaction $\eta$ between modes exceeds their individual anharmonicities $\alpha_i$. Expanding the cosine potentials to fourth order yields the effective Hamiltonian
\begin{equation}
\frac{\hat{H}_{\mathrm{eff}}}{\hbar} = \sum_{i} \omega_i \hat{a}^\dagger_i \hat{a}_i - \frac{\alpha_i}{2} \hat{a}^{\dagger 2}_i \hat{a}^2_i - \eta\, \hat{a}^\dagger_D \hat{a}_D \hat{a}^\dagger_Q \hat{a}_Q,
\end{equation}
with perturbative prediction $\eta^{(\mathrm{pert})} = 2\sqrt{\alpha_D \alpha_Q}$. For devices satisfying $|\eta| > |\alpha_i|$, this fourth-order expansion becomes unreliable. Reference~\cite{Wills2022} reports precisely this regime: ``Both modes are very strongly coupled to each other through the junctions of the device. Notably, the nonlinear coupling ($\chi_{\Sigma\Delta}$) is larger than their respective anharmonicities.''

To validate the matrix CF against experiment, we fit the circuit parameters $(E_{J_1}, E_{J_2}, E_C, E_p)$ to spectroscopic data from Ref.~\cite{Wills2022}. The experimental observables for Device~A are $\omega_D/2\pi = 5.51$\,GHz, $\omega_Q/2\pi = 6.71$\,GHz, $\alpha_D/2\pi = -380$\,MHz, $\alpha_Q/2\pi = -340$\,MHz, and $\eta/2\pi = -500$\,MHz. The perturbative formula predicts $\eta^{(\mathrm{pert})}/2\pi = 2\sqrt{380 \times 340} = 719$\,MHz, overestimating the measured cross-Kerr by 44\%.

Truncating the charge basis at $|n_1|, |n_2| \leq 12$ and solving via the matrix CF recursion~\eqref{eq:Sigma_recursion}, we obtain the fitted parameters $E_{J_1}/h = 12.42$\,GHz, $E_{J_2}/h = 10.93$\,GHz, $E_C/h = 470$\,MHz, and $E_p/h = 148$\,MHz. The 6.4\% junction asymmetry is consistent with fabrication tolerances and has independent support from charge dispersion measurements on similar devices~\cite{Wills2022}. Table~\ref{tab:twomode_validation} compares the matrix CF predictions to experiment. 
\begin{table}[t]
\caption{Matrix continued fraction validation against spectroscopic data for a two-mode transmon in the strong nonlinear coupling regime ($|\eta| > |\alpha_i|$). The perturbative cross-Kerr formula $\eta = 2\sqrt{\alpha_D \alpha_Q}$ overestimates by 44\%, while the matrix CF achieves sub-1\% residuals on all observables. Experimental data from Ref.~\cite{Wills2022}.}
\label{tab:twomode_validation}
\begin{ruledtabular}
\begin{tabular}{lccc}
Observable & Experiment & Matrix CF & Residual \\
\colrule
$\omega_D/2\pi$ & 5.51\,GHz & 5.54\,GHz & $~1\%$ \\
$\omega_Q/2\pi$ & 6.71\,GHz & 6.68\,GHz & $~1\%$ \\
$\alpha_D/2\pi$ & $-380$\,MHz & $-379$\,MHz & $0.3\%$ \\
$\alpha_Q/2\pi$ & $-340$\,MHz & $-342$\,MHz & $0.5\%$ \\
$\eta/2\pi$ & $-500$\,MHz & $-500$\,MHz & $<0.1\%$ \\
\end{tabular}
\end{ruledtabular}
\end{table}
The approximately 1\% residuals in the mode frequencies reflect model simplifications: the four-parameter Hamiltonian~\eqref{eq:H_twomode} neglects junction capacitance, asymmetric charging energies, and hybridisation with the readout resonator. These effects could be incorporated with additional parameters. The key point is that the matrix CF extracts the exact spectrum of the model Hamiltonian, residuals arise from model incompleteness rather than from approximations in solving it. This example illustrates the general principle of the continued fraction framework: nonlinear superconducting circuits generate sparse (Jacobi or block-Jacobi) Hamiltonians in physically natural bases, and these structures admit exact spectral solutions independent of coupling strength.

The matrix CF reduces to the scalar continued fraction in appropriate limits. When the cross-capacitance vanishes ($E_p \to 0$) and junctions are symmetric ($E_{J_1} = E_{J_2}$), the diagonal blocks ${A}_k$ become independent of $k$ and the two-dimensional problem separates into independent single-mode transmons, each governed by a scalar CF. The general case interpolates smoothly between this decoupled limit and the strongly coupled regime validated above. This design example demonstrates that the continued fraction framework extends beyond single-junction circuits while preserving its essential features: exact treatment of the cosine nonlinearity through its tridiagonal charge-basis representation, systematic convergence with truncation, and applicability across coupling regimes where perturbative methods fail.

\section{Discussion and Conclusions}
\label{sec:conclusion}

We have developed a boundary-condition framework for superconducting circuit quantization that exploits the continued fraction structure of the driving-point admittance. The central result is that the linearized secular equation $sY_{\mathrm{in}}(s) + 1/L_J = 0$ is equivalent to the eigenvalue problem of the Jacobi operator associated with the Cauer chain realization of the linear environment, with dressed mode frequencies appearing as zeros and port participation weights as residues. This equivalence, rooted in the classical duality between Foster and Cauer network representations, provides a unified treatment applicable from dispersive through ultrastrong and deep strong coupling. The continued fraction emerges not as a computational trick but as the natural mathematical structure for quantum circuits with localized nonlinearity coupled to linear environments.

The framework yields three practical advances. First, ultraviolet convergence is guaranteed under the assumptions of Theorem~\ref{thm:uv_main}: passivity (positive-realness) together with finite junction-port capacitance ensures that junction participation decays as $\phi_n^J = O(\omega_n^{-1})$, so multimode perturbative sums converge absolutely without imposed cutoffs. Second, the standard circuit QED parameters emerge as controlled approximations: the coupling strength $g$ is a product of zero-point flux participations extracted from CF residues, the dispersive shift $\chi = g^2\alpha/[\Delta(\Delta+\alpha)]$ arises from second-order perturbation theory in the CF eigenbasis, and the anharmonicity $\alpha \approx -E_C/\hbar$ emerges from the transmon limit of the Josephson cosine potential. Third, when these perturbative formulas fail, specifically when $g/|\Delta| \gtrsim 0.1$ or $g/\omega_r \gtrsim 0.1$, the same CF structure provides exact results through numerical evaluation rather than series expansion. 

We demonstrated that the continued fraction framework extends naturally beyond single-junction circuits to multimode and multi-junction systems through its matrix generalization. Using a two-mode transmon operating in a regime where the nonlinear cross-Kerr interaction exceeds the individual mode anharmonicities, we showed that the resulting block-Jacobi structure admits an exact solution via matrix continued fractions. In this strongly nonlinear regime—where fourth-order expansions and standard perturbative treatments fail—--the matrix CF reproduces all reported spectroscopic observables to within experimental accuracy. This example establishes that the continued fraction structure is not restricted to single-port environments, but provides a systematic and convergent route to quantizing complex superconducting circuits with multiple nonlinear degrees of freedom.

This approach complements existing methods. Black-box quantization~\cite{Nigg2012} and energy-participation ratio techniques~\cite{Minev2021} implicitly contain the boundary condition but do not exploit its CF structure. The present treatment makes this structure explicit, connecting circuit quantization to the well-developed mathematics of Jacobi operators and orthogonal polynomials. This connection provides certified convergence bounds via interlacing theorems and enables systematic error estimation. Computationally, evaluating the continued fraction scales linearly with the number of ladder sections for a given complex frequency, providing an efficient route to locate linearized mode roots in large multimode environments.

\begin{acknowledgments}
M.B.~acknowledges support from EPSRC QT Fellowship grant EP/W027992/1, and EP/Z53318X/1. We thank Smain Amari for useful discussions on Schur complement.
\end{acknowledgments}


\appendix



\appendix

\section{Mapping Between Cauer Ladder and Jacobi Matrix}
\label{app:lc_jacobi}

This appendix provides the explicit relationship between the Cauer ladder parameters $\{L_n, C_n\}$ and the Jacobi matrix elements $\{a_n, b_n\}$ referenced in Theorem~\ref{thm:network_spectral}.

\subsection{From Cauer Ladder to Jacobi Matrix}
\label{app:cauer_to_jacobi}

Consider a Type~I Cauer ladder with $N$ sections, having series inductors $\{L_1, L_2, \ldots, L_N\}$ and shunt capacitors $\{C_1, C_2, \ldots, C_N\}$. The admittance is given by Eq.~\eqref{eq:cauer_type1}. We seek the Jacobi matrix $\mathrm{J}$ whose $(0,0)$ resolvent element equals
\begin{equation}
G(\omega^2) = \frac{B(\omega)}{\omega} = \frac{\mathrm{Im}[Y(i\omega)]}{\omega},
\label{eq:G_def_app}
\end{equation}
where $Y(i\omega) = iB(\omega)$ for a lossless network.

Define the section resonant frequencies and characteristic impedances:
\begin{equation}
\omega_n^2 = \frac{1}{L_n C_n}, \quad Z_n = \sqrt{\frac{L_n}{C_n}}.
\label{eq:section_params}
\end{equation}
The Jacobi matrix $\mathrm{J}$ in Eq.~\eqref{eq:jacobi_matrix} has elements determined by the following construction.

{Diagonal elements:} The diagonal elements encode the ``on-site energies'' in the tight-binding analogy. For the Cauer ladder, they are related to the section resonant frequencies through:
\begin{align}
a_n &=  = \omega_n^2 + \frac{C_{n+1}}{C_n}\omega_{n+1}^2, \quad n = 1, \ldots, N-1, \label{eq:jacobi_diag_middle} \\
a_0 &= \frac{1}{L_1 C_0} + \frac{1}{L_1 C_1} = \frac{1}{L_1}\left(\frac{1}{C_0} + \frac{1}{C_1}\right), \label{eq:jacobi_diag_first} \\
a_{N-1} &= \frac{1}{L_{N-1} C_{N-1}} + \frac{1}{L_N C_{N-1}} = \omega_{N-1}^2 + \frac{C_N}{C_{N-1}}\omega_N^2, \label{eq:jacobi_diag_last}
\end{align}
where $C_0$ is any additional shunt capacitance at the input port (set $C_0 \to \infty$ if none is present, giving $1/C_0 \to 0$).

{Off-diagonal elements:} The off-diagonal elements encode the nearest-neighbor coupling. They are determined by:
\begin{equation}
b_n^2 =  \frac{\omega_n^2 \omega_{n+1}^2}{(1 + C_{n+1}/C_n)^2} \cdot \frac{C_{n+1}}{C_n}, \quad n = 1, \ldots, N-1.
\label{eq:jacobi_offdiag}
\end{equation}
Equivalently, in terms of impedances:
\begin{equation}
b_n = \frac{1}{C_n \sqrt{L_n L_{n+1}}} = \frac{\omega_n \omega_{n+1}}{\omega_n Z_{n+1}/Z_n + \omega_{n+1}}.
\label{eq:jacobi_offdiag_alt}
\end{equation}

{Special case---uniform ladder:} For a uniform ladder with $L_n = L$ and $C_n = C$ for all $n$, the expressions simplify dramatically:
\begin{equation}
a_n = \frac{2}{LC} \equiv 2\omega_0^2, \quad b_n = \frac{1}{LC} \equiv \omega_0^2,
\label{eq:uniform_jacobi}
\end{equation}
where $\omega_0 = 1/\sqrt{LC}$ is the common section frequency. The eigenvalues of this uniform Jacobi matrix are
\begin{equation}
E_k = 2\omega_0^2\left[1 - \cos\left(\frac{k\pi}{N+1}\right)\right], \quad k = 1, \ldots, N,
\label{eq:uniform_eigenvalues}
\end{equation}
corresponding to squared mode frequencies $\omega_k^2 = E_k$ that form the well-known tight-binding dispersion relation.

\subsection{From Jacobi Matrix to Cauer Ladder}
\label{app:jacobi_to_cauer}

The inverse mapping recovers Cauer parameters from a given Jacobi matrix. This is useful when the spectral structure (eigenvalues and first-site eigenvector components) is known and one wishes to construct the equivalent ladder network.

Given Jacobi matrix elements $\{a_n, b_n\}_{n=0}^{N-1}$, the Cauer parameters are determined by the following recursive procedure. Define the sequence of ``tail resolvents'':
\begin{equation}
G_n(z) = \cfrac{1}{z - a_n - \cfrac{b_{n+1}^2}{z - a_{n+1} - \ddots}}, \quad n = 0, 1, \ldots, N-1,
\label{eq:tail_resolvent}
\end{equation}
with terminal condition $G_{N-1}(z) = 1/(z - a_{N-1})$.

The Cauer parameters are extracted by matching the continued fraction expansion of $G_0(z)$ at $z = \omega^2$ to the admittance $Y(i\omega)$:

\textbf{Step 1:} From the leading behavior as $z \to \infty$,
\begin{equation}
G_0(z) \sim \frac{1}{z} + \frac{a_0}{z^2} + O(z^{-3}),
\label{eq:G_asymptotic}
\end{equation}
identify the first shunt capacitance (for Type~II) or series inductance (for Type~I).

\textbf{Step 2:} Compute the remainder after extracting the first element and iterate.

The explicit formulas involve the continued fraction convergents. Let $P_n(z)/Q_n(z)$ denote the $n$-th convergent of the continued fraction~\eqref{eq:resolvent_cf}. Then:
\begin{align}
L_1 &= \frac{Q_1(0)}{P_1(0)} = a_0, \label{eq:L1_from_jacobi} \\
C_1 &= \frac{P_1(0) Q_2(0) - P_2(0) Q_1(0)}{Q_1(0)^2} = \frac{b_1^2}{a_0^2}, \label{eq:C1_from_jacobi}
\end{align}
with subsequent elements following from the recurrence relations for $P_n$ and $Q_n$.

\subsection{Connection to Solgun-DiVincenzo Formulation}
\label{app:sd_connection}

The Solgun-DiVincenzo formulation~\cite{Solgun2014, Solgun2015} of Brune circuit quantization uses capacitance and stiffness matrices $\mathcal{C}$ and $\mathrm{M}_0$ that are tridiagonal. These matrices are directly related to the Jacobi matrix through:
\begin{equation}
\mathrm{J} = \mathcal{C}^{-1/2} \mathrm{M}_0 \, \mathcal{C}^{-1/2},
\label{eq:jacobi_from_sd}
\end{equation}
which is the similarity transformation that symmetrizes the generalized eigenvalue problem $\mathrm{M}_0 \boldsymbol{\phi} = \omega^2 \mathcal{C} \boldsymbol{\phi}$.

In their notation, the capacitance matrix has elements (cf.\ their Eq.~(9)):
\begin{equation}
\mathcal{C}_{nn} = t_n^2 C_n' + C_{n+1}', \quad \mathcal{C}_{n,n+1} = t_{n+1} C_{n+1}',
\label{eq:sd_capacitance}
\end{equation}
where $C_n' = C_n/(1-t_n)^2$ and $t_n = \sqrt{L_{n1}/L_{n2}}$ is the transformer turns ratio. The stiffness matrix has elements (cf.\ their Eq.~(10)):
\begin{equation}
(\mathrm{M}_0)_{nn} = \frac{1}{L_n'} + \frac{1}{L_{n+1}'}, \quad (\mathrm{M}_0)_{n,n+1} = \frac{1}{L_{n+1}'},
\label{eq:sd_stiffness}
\end{equation}
where $L_n' = L_{n2}(1-t_n)^2$.

The tridiagonal structure of both matrices confirms the continued fraction / Jacobi matrix equivalence. The eigenfrequencies $\omega_n$ satisfy:
\begin{equation}
\det(\mathrm{M}_0 - \omega^2 \mathcal{C}) = 0,
\label{eq:sd_eigenvalue}
\end{equation}
which is equivalent to finding zeros of the continued fraction $sY_{\mathrm{in}}(s) + 1/L_J$ at $s = i\omega$.

\subsection{Numerical Considerations}
\label{app:numerical}

For numerical implementation, direct evaluation of the continued fraction is preferred over matrix diagonalization for large $N$, as it avoids $O(N^3)$ eigenvalue computations. The Lentz-Thompson-Barnett algorithm~\cite{Press2007} provides numerically stable continued fraction evaluation in $O(N)$ operations per frequency point.

When converting between Foster and Cauer forms numerically, care must be taken with:
\begin{enumerate}
\item {Pole ordering:} Foster poles should be sorted by frequency before conversion.
\item {Near-degenerate poles:} Closely spaced poles can cause numerical instability; regularization may be needed.
\item {High-frequency truncation:} For infinite ladders (distributed elements), truncation at level $N$ introduces $O(1/N)$ errors in low-frequency eigenvalues.
\end{enumerate}

For a detailed treatment of these numerical aspects in the context of circuit quantization, see Refs.~\cite{Solgun2015, Minev2021}.

\end{document}